\documentclass[a4paper,12pt]{article}

\usepackage[T1]{fontenc}
\usepackage{subcaption}
\usepackage{amsmath}
\usepackage{amsbsy}
\usepackage{amsthm}
\usepackage{amsfonts}
\usepackage{amssymb}
\usepackage{mathtools}
\usepackage{array}
\usepackage{bbm} %
\usepackage[textwidth=3cm]{todonotes}
\usepackage{tikz}

\usepackage{newunicodechar}
\newunicodechar{ȩ}{\c{e}} %
\usepackage{stmaryrd}
\SetSymbolFont{stmry}{bold}{U}{stmry}{m}{n}
\usepackage[scr=rsfs]{mathalpha} %

\usepackage[style=alphabetic,url=false,doi=false,backend=biber,maxbibnames=20]{biblatex}
\renewbibmacro{in:}{%
  \ifentrytype{article}{}{\printtext{\bibstring{in}\intitlepunct}}}
\AtEveryBibitem{
	\clearfield{number}
}

\reversemarginpar

\usetikzlibrary{math}
\usetikzlibrary{calc}
\usetikzlibrary{fit}
\usetikzlibrary{shapes}
\usetikzlibrary{backgrounds}

\usepackage{hyperref}

\usepackage[noabbrev,capitalize]{cleveref}
\usepackage{lastpage}

\newcounter{Ccounter}
\newcommand{\Clast}{\ensuremath{C_{\theCcounter}}}
\newcommand{\C}{
	{\refstepcounter{Ccounter}\Clast}
}
\newcommand{\Cp}[1]{
	\refstepcounter{Ccounter}
	\label{CC#1}
}
\newcommand{\Cl}[1]{
	\refstepcounter{Ccounter}
	\label{CC#1}
	\Clast
}
\makeatletter
\RequirePackage{ifthen}
\@ifpackageloaded{hyperref}{
	\newcommand{\Cr}[1]{\ensuremath{C_{\ref*{CC#1}}}}
}{ %
	\newcommand{\Cr}[1]{\ensuremath{C_{\ref{CC#1}}}}
}
\makeatother

\newtheorem{theorem}{Theorem}
\newtheorem{lemma}[theorem]{Lemma}

\newtheorem{corollary}[theorem]{Corollary}

\newtheorem{remark}[theorem]{Remark}
\newtheorem{definition}[theorem]{Definition}

\numberwithin{equation}{section}

\newcommand{\bs}{\boldsymbol}
\newcommand{\wt}{\widetilde}

\newcommand{\ind}{\mathbbm{1}}
\DeclareMathOperator{\degree}{deg}

\DeclareMathOperator{\Pf}{Pf}
\DeclareMathOperator{\PPf}{PPf}
\DeclareMathOperator{\range}{range}

\DeclareMathOperator{\supp}{supp}
\DeclareMathOperator{\sign}{sign}

\newcommand{\wickord}[1]{:\mathrel{#1}:}
\newcommand{\floor}[1]{\left\lfloor {#1} \right\rfloor}

\newcommand{\intrange}[2]{{\llbracket#1,#2\rrbracket}}

\newcommand{\dd}{\,\text{\rm d}}             %

\renewcommand{\Re}{\operatorname{Re}}

\newcommand{\ul}[1]{{\ensuremath{\underline{#1}}}}
\newcommand{\lis}[1]{{\ensuremath{\overline{#1}}}}

\newcommand{\B}{\ensuremath{\textup{B}} }
\newcommand{\D}{\ensuremath{\textup{D}} }
\newcommand{\Dmin}{\ensuremath{D_\textup{min}} }
\newcommand{\E}{\ensuremath{\textup{E}} }
\newcommand{\GH}{\ensuremath{\textup{GH}} }
\newcommand{\R}{\ensuremath{\textup{R}} }
\newcommand{\T}{\ensuremath{\textup{T}} }

\newcommand{\hk}{\ensuremath{\textup{hk}}}

\newcommand{\inter}{\ensuremath{\textup{int}} }
\newcommand{\many}{\ensuremath{\textup{many}} }

\newcommand{\chiduo}{\ensuremath{\chi_{\textup{duo}}}}

\newcommand{\chimono}{\ensuremath{\chi_{\textup{mono}}}}
\newcommand{\chirel}{\ensuremath{\chi_{\textup{irrel}}}}
\newcommand{\chitree}{\ensuremath{\chi_{\textup{tree}}}}

\newcommand{\bC}{\ensuremath{\mathbb{C}}}

\newcommand{\bN}{\ensuremath{\mathbb{N}}}

\newcommand{\bR}{\ensuremath{\mathbb{R}}}

\newcommand{\bZ}{\ensuremath{\mathbb{Z}}}

\newcommand{\cC}{\ensuremath{\mathcal{C}}}

\newcommand{\cF}{\ensuremath{\mathcal{F}}}

\newcommand{\cH}{\ensuremath{\mathcal{H}}}
\newcommand{\cI}{\ensuremath{\mathcal{I}}}
\newcommand{\cL}{\ensuremath{\mathcal{L}}}

\newcommand{\cO}{\ensuremath{\mathcal{O}}}
\newcommand{\cP}{\ensuremath{\mathcal{P}}}
\newcommand{\cQ}{\ensuremath{\mathcal{Q}}}
\newcommand{\cR}{\ensuremath{\mathcal{R}}}
\newcommand{\cS}{\ensuremath{\mathcal{S}}}
\newcommand{\cT}{\ensuremath{\mathcal{T}}}

\newcommand{\cV}{\ensuremath{\mathcal{V}}}
\newcommand{\cW}{\ensuremath{\mathcal{W}}}
\newcommand{\cY}{\ensuremath{\mathcal{Y}}}

\newcommand{\fA}{\ensuremath{\mathfrak{A}}}
\newcommand{\fD}{\ensuremath{\mathfrak{D}}}
\newcommand{\fH}{\ensuremath{\mathfrak{H}}}
\newcommand{\fI}{\ensuremath{\mathfrak{I}}}
\newcommand{\fL}{\ensuremath{\mathfrak{L}}}
\newcommand{\fP}{\ensuremath{\mathfrak{P}}}

\newcommand{\fT}{\ensuremath{\mathfrak{T}}}
\newcommand{\fW}{\ensuremath{\mathfrak{W}}}

\newcommand{\sB}{\ensuremath{\mathscr{B}}}

\newcommand{\sV}{\ensuremath{\mathscr{V}}}
\newcommand{\sY}{\ensuremath{\mathscr{Y}}}

\newcommand{\fin}{\textrm{finite}}

\newcommand{\diffy}{\ensuremath{\mathsf{diffy}}}
\newcommand{\diffyn}{\ensuremath{\diffy_n}}
\newcommand{\symm}{\ensuremath{\mathsf{symm}}}
\newcommand{\slow}{\ensuremath{\mathsf{slow}}}
\newcommand{\slown}{\ensuremath{\mathsf{slow}_n}}
\newcommand{\kernel}{\ensuremath{\mathsf{kernel}}}
\newcommand{\kerneln}{\ensuremath{\mathsf{kernel}_n}}
\newcommand{\kernelfin}{\ensuremath{\kernel_\fin}}
\newcommand{\mi}{\ensuremath{\mathsf{mi}}}
\newcommand{\pot}{\ensuremath{\mathsf{pot}}}
\newcommand{\potfin}{\ensuremath{\pot_\fin}}
\newcommand{\test}{\ensuremath{\mathsf{test}}}
\newcommand{\testn}{\ensuremath{\mathsf{test}_n}}
\newcommand{\tn}{\ensuremath{\mathsf{tn}}}
\newcommand{\Grassmann}{\ensuremath{\mathsf{Grassmann}}}
\newcommand{\Grassmannn}{\ensuremath{\mathsf{Grassmann}_n}}

\newcommand{\leg}{\ensuremath{\mathsf{leg}}}
\newcommand{\edge}{\ensuremath{\mathsf{edge}}}
\newcommand{\el}{\ensuremath{\mathsf{el}}}
\newcommand{\EL}{\ensuremath{\mathsf{EL}}}
\newcommand{\IG}{\ensuremath{\mathsf{IG}}}
\newcommand{\LL}{\ensuremath{\mathsf{LL}}}
\newcommand{\tl}{\ensuremath{\mathsf{tl}}}
\newcommand{\TL}{\ensuremath{\mathsf{TL}}}

\title{Constructing a weakly-interacting fixed point of the Fermionic Polchinski equation}
\author{Rafael Leon Greenblatt\thanks{E-mail address: \href{mailto:greenblatt@mat.uniroma2.it}{greenblatt@mat.uniroma2.it}} %
\\ Dipartimento di Matematica \\ Universit\`a degli Studi di Roma ``Tor Vergata'' \\ Rome, Italy.}  
\date{\today}

\begin{document}

\maketitle

\begin{abstract}
	I rigorously prove the existence of a nontrivial fixed point of a family of continuous renormalization group flows corresponding to certain weakly interacting Fermionic quantum field theories with a parameter in the propagator allowing the scaling dimension to be tuned in a manner analogous to dimensional regularization.
\end{abstract}

\tableofcontents

\section{Introduction}
\label{sec:intro}

The (Fermionic)\footnote{\Cref{eq:flow} also has a better known Bosonic version where a complex-valued function $\phi$ takes the place of $\psi$ and $\cV_\Lambda$ is an ordinary functional; this was the main subject of \cite{Pol}, where the equation was first introduced.} Polchinski equation is 
\begin{equation}
	\frac{\partial}{\partial \Lambda} \cV_\Lambda(\psi)
	=
	\frac{1}{2} \left\langle \frac{\delta}{\delta \psi}, G_\Lambda \frac{\delta}{\delta\psi} \right\rangle \cV_\Lambda (\psi)
	-
	\frac12 \left\langle \frac{\delta}{\delta \psi} \cV_\Lambda(\psi), G_\Lambda \frac{\delta}{\delta \psi} \cV_{\Lambda}(\psi) \right\rangle
	\label{eq:flow}
\end{equation}
where $\Lambda$ is a positive real parameter, 
$\cV_\cdot (\psi)$ is a map from $(0,\infty)$ to a Grassmann algebra (so the product is noncommutative) with a set of generators $\left\{ \psi_a (x) \right\}$ indexed by a position $x$ and a discrete index $a$ taking $N$ allowed values (see \cref{sec:Grassmann} for a more precise definition); the variational derivative is a Grassmann derivative, $G_\Lambda$ is an integral operator (I will also use the same symbol for its kernel) and $\left\langle \cdot,\cdot \right\rangle$ is an $L^2$ inner product contracting the arguments of the derivatives.  
It is often convenient to ignore the constant or C-number term of $\cV_\Lambda$ (that is, the part in the subspace of the Grassmann algebra spanned by its unit), since it may be ill-defined and in any case it does not appear on the right hand side of \cref{eq:flow} since all of its Grassmann derivatives vanish; I will do so implicitly for now, and explicitly in \cref{sec:flow_formal} below when I rewrite \cref{eq:flow} in the notation I use for the construction of a fixed point.

This equation was introduced by Polchinski \cite{Pol} as an alternative, simpler formulation of the continuous-parameter Renormalization Group (RG) introduced by Wilson \cite{Wilson1}.  In physical terms, $\cV_\Lambda$ is an effective interaction, which together with the cutoff propagator 
$P_\Lambda := \int_0^{\Lambda} G_\Theta \dd \Theta$
defines an effective description of a quantum field theory depending on the scale parameter $\Lambda$.  \Cref{eq:flow} was obtained by comparing different expressions for the same functional (path) integral, so that the different effective descriptions corresponding to a single solution of the equation evaluated at different values of $\Lambda$ describe the same physical system, and are consistent insofar as they describe the same phenomena at long distance scales, in a sense that depends somewhat on the details of how the propagator depends on $\Lambda$.
In this framework, a \emph{(quantum field) theory} or \emph{model} corresponds to a specification of a set of possible interactions (usually involving symmetry properties) and the full propagator $\int_0^\infty G_{\Lambda} \dd \Lambda$, while the specific form of $G_\Lambda$ corresponds to the choice of an \emph{(ultraviolet) cutoff}.  
Typically there are a variety of reasonable cutoffs for any given model which are expected to produce the same physical outcome.

One quirk of this interpretation is that \cref{eq:flow} is usually thought of as describing an evolution from initial conditions (the ``bare action'') at a large value of $\Lambda$ (or in the ``ultraviolet'' limit $\Lambda \to \infty$) to smaller values, contrary to the usual convention for dynamical systems.  On closer examination, however, many of the interesting problems (including the one which the equation was originally presented to study \cite{Pol}) are mixed boundary value problems with a combination of conditions on the bare action and conditions at some finite $\Lambda$ (``renormalization conditions'') or on the asymptotic behavior for $\Lambda \to 0$ (the ``infrared limit'').

Since I wish to stress that the question of the existence of solutions of the Polchinski equation is conceptually and logically separate from the question of its applicability to physical problems or its relationship to other mathematical descriptions (such as path integrals) I will not say anything further about them here.  Instead I take \cref{eq:flow} as a starting point (which is also why I mostly treat $G_\Lambda$ as a more basic object than $P_\Lambda$).

Solutions of \cref{eq:flow} of the following special type, known as scaling solutions or as fixed points of the renormalization group flow, play a central role in many applications, in particular to critical phenomena.
Suppose that $G_\Lambda$ has the self-similarity property
\begin{equation}
	G_\Lambda(x,y) = \Lambda^{2[\psi]-1} G ( \Lambda x,  \Lambda y)
	\label{eq:propagator_scaling}
\end{equation}
(understanding this as an equality between $N\times N$ matrices, so that the discrete indices are unaffected) for some $[\psi] > 0$, which is known as the scaling dimension of the field $\psi$.
Introducing $\wt D_{\Lambda}(\psi)(x) := \psi( \Lambda x)$ and $[\fD_\Lambda \cV] (\psi) := \cV \left( \Lambda^{-[\psi]} \wt D_{ \Lambda} (\psi) \right)$ (more carefully defined in \cref{eq:dilation_maindef} below), a solution of the form
\begin{equation}
	\cV_\Lambda = \fD_\Lambda \cV^*
	\label{eq:scaling_solution}
\end{equation}
is called a scaling solution.
For such a solution, \cref{eq:flow} becomes
\begin{equation}
	\frac{\partial}{\partial \Lambda} \left[ \fD_\Lambda \cV^* \right]
	=
	\begin{aligned}[t]
		&
		\frac{1}{2} \left\langle \frac{\delta}{\delta \psi}, G_\Lambda \frac{\delta}{\delta\psi} \right\rangle \left[ \fD_\Lambda \cV^* \right] (\psi)
		\\ & \quad
		-
		\frac12 \left\langle \frac{\delta}{\delta \psi} \left[ \fD_\Lambda \cV^* \right](\psi), G_\Lambda \frac{\delta}{\delta \psi} \left[ \fD_\Lambda \cV^* \right](\psi) \right\rangle
		.
	\end{aligned}
	\label{eq:stationary_Grassmann}
\end{equation}
note that
when $G_\Lambda$ satisfies \cref{eq:propagator_scaling}, the definition of $\fD_\Lambda$ is exactly such that this condition is actually independent of $\Lambda$, i.e.\ if it holds for $\Lambda=1$ then it also holds for all $\Lambda \in (0,\infty)$.
This scale independence can be made more prominent as follows: considering the reparameterization
\begin{equation}
	\cV_\Lambda(\psi) 
	=
	[\fD_\Lambda \cW_{\log \Lambda}] (\psi)
	,
	\label{eq:dilation_cW}
\end{equation}
(note the logarithmic scale for $\cW$)
and letting $\Delta$ denote the generator of the one-dimensional semigroup $D_{e^\eta}$,
if $\cW_\eta$ satisfies the equation
\begin{equation}
	\frac{\partial}{\partial \eta}
	\cW_\eta (\psi)
	=
	\frac{1}{2} \left\langle \frac{\delta}{\delta \psi}, G \frac{\delta}{\delta\psi} \right\rangle \cW_\eta (\psi)
	-
	\Delta \cW_\eta (\psi)
	-
	\frac12 \left\langle \frac{\delta}{\delta \psi} \cW_\eta(\psi), G \frac{\delta}{\delta \psi} \cW_\eta(\psi) \right\rangle
	\label{eq:pol_functional}
\end{equation}
this implies that $\cV_\Lambda$ satisfies \cref{eq:flow}.
\Cref{eq:pol_functional}, unlike \eqref{eq:flow} is autonomous (i.e.\ there is no explicit dependence on the scale parameter), and so it may well have fixed points, i.e.\ solutions of the form $\cW_\eta \equiv \cV^*$.
Comparing \cref{eq:stationary_Grassmann,eq:dilation_cW}, fixed points of \cref{eq:pol_functional} (that is, fixed points of the RG flow) correspond to scaling solutions of \cref{eq:flow} and vice versa.  I will sometimes find it useful to useful to pass back and forth between the two formulations.

Such fixed points are associated with self-similar behavior and the presence of scaling exponents in translation-invariant models, though of course there are further generalizations (for example, for disordered systems one can consider solutions which are only fixed points in distribution, which seems to be necessary to accommodate infinite-order phase transitions \cite{Fisher.PRB,GG.critical,CGG.continuum}) and other mechanisms for producing such behavior.

Actually, I will make a further manipulation which simplifies the fixed-point problem somewhat.   Considering the Wick ordering of the interaction with respect to the UV-cutoff propagator $P_\Lambda $, defined by
\begin{equation}
	\wickord{\cV(\psi)}_\Lambda
	\ 
	:=
	\sum_{m=0}^\infty
	\frac{1}{m!}
	\left[ -\frac{1}{2} \left\langle \frac{\delta}{\delta\psi} , P_\Lambda \frac{\delta}{\delta \psi} \right\rangle \right]^m \cV(\psi)
	,
	\label{eq:Wick_order_formal}
\end{equation}
then at least formally
\begin{equation}
	\begin{split}
		\frac{\partial}{\partial \Lambda} 
		\wickord{\cV(\psi)}_\Lambda
		&=
		-\frac{1}{2}
		\left\langle \frac{\delta}{\delta\psi} , G_\Lambda \frac{\delta}{\delta \psi} \right\rangle  \wickord{\cV(\psi)}_\Lambda
	\end{split}
	\label{eq:Wick_deriv_formal}
\end{equation}
which combines with \cref{eq:flow} to give
\begin{equation}
	\begin{aligned}
		\frac{\partial}{\partial \Lambda} \wickord{\cV_\Lambda(\psi)}_\Lambda
		& =
		-
		\frac12 \wickord{\left\langle \frac{\delta}{\delta \psi} \cV_\Lambda(\psi), G_\Lambda \frac{\delta}{\delta \psi} \cV_{\Lambda}(\psi) \right\rangle}_\Lambda
		\\ 
		& = 
		-
		\frac12 
		\sum_{m=0}^\infty
		\frac{(-1)^m}{m!}
		\left\langle \frac{\delta}{\delta\psi} , P_\Lambda \frac{\delta}{\delta \varphi} \right\rangle^m
		\left\langle \frac{\delta}{\delta \psi} \wickord{\cV_\Lambda(\psi)}_\Lambda, G_\Lambda \frac{\delta}{\delta \varphi}\wickord{\cV_\Lambda(\varphi)}_\Lambda \right\rangle
		\Big|_{\varphi=\psi}
	\end{aligned}
	\label{eq:flow_Wick_formal}
\end{equation}
that is, a differential equation for $\wickord{\cV_\Lambda(\psi)}_\Lambda$, in which the right hand side is a quadratic form with no linear part.
Graphically, it corresponds to taking two vertices and making an arbitrary nonzero number of contractions between the two (no self contractions), assigning the value $-G_\Lambda$ to one of the lines and $-P_\Lambda$ to the others.  

\medskip

In this work I will present some techniques for rigorously constructing weakly interacting fixed points in this sense, in particular for which the various equivalences mentioned above can be verified.
In particular I give an alternative construction of a nontrivial fixed point 
of the model studied in \cite{GMR}\footnote{An almost identical proof also holds for the model considered in \cite{GK85}, with the main change being that \cref{eq:beta_with_error} should be replaced by a rearranged version of \cite[Equations~(2.20-2.21)]{GK85}; I do not spell this out because I only became aware of \cite{GK85} after this article had been submitted for publication and one example seemed sufficient.}, which I present in more detail in \cref{sec:example}.  
Compared to \cite{GMR} who study a discrete RG transformation, I study a continuous flow equation, which has the advantage that the resulting fixed point is evidently invariant (or covariant) under arbitrary rescaling.  
I would like to think that the proof in this paper is also somewhat more straightforward, but it is not as direct as one might hope: 
I solve the fixed point problem by an ansatz in which the irrelevant parts are given by a function of the local relevant part, defined via a tree expansion.  Once this is done, however, the possible choices of the relevant part form a finite-dimensional space, and the ansatz reduces the fixed point condition to that of finding the zeroes of a vector field on this space, which admits a very well-behaved expansion in for small values of a parameter $\varepsilon$ in the definition of the propagator.
Like \cite{GMR}, this construction is ``nonperturbative'' in the sense that it corresponds to a solution of the exact renormalization group transformation and does not assume the validity of perturbation theory, but ``perturbative'' in the sense that the fixed point can be approximated well by the solution of a truncated problem.

In \cref{sec:Grassmann} I begin the more technical part of the present work by defining the Grassmann algebra \Grassmann\ and associated operations used in the statement of \cref{eq:flow} and the related expressions above.
The Grassmann algebra I introduce is extremely large; for example it allows arbitrary finite sums of terms of the form
\begin{equation}
	\int
	V(x_1,\dots,x_n)
	\frac{\partial^{m_{11}}}{\partial (x_1)_1^{m_{11}}}\dots\psi_{a_1}(x_1)
	\dots
	\frac{\partial^{m_{n d}}}{\partial (x_n)_d^{m_{n d}}}\psi_{a_n}(x_n)
	\dd^d x_1 \dots \dd^d x_n
	,
\end{equation}
with derivatives of any order and for any $V:(\bR^d)^n \to \bC$, which need not be globally bounded, as well as similar terms where some or all of the position arguments are repeated, and some infinite sums.
The only possible further generalization I am aware of with any significant applications would be replacing $\bR^d$ with another Euclidean manifold, which is entirely straightforward at this level of abstraction.  In the context of this space,   we can speak of solutions of the Polchinski equation (for a given choice of the ``parameter'' $ G_\cdot$) in the same sense as any other ordinary differential equation.  

In fact this space is much too large to be of any practical use in constructing solutions of \cref{eq:flow}, and so in \cref{sec:formalism} I introduce a much smaller space and a much finer topology, along with the associated formalism used in the rest of the paper, and prove some basic estimates and other properties that I will use later in the proof. 
The scope of this section will probably be more recognizable than those of \cref{sec:Grassmann} to readers familiar with some of the recent works I have cited above, but there are some important differences.
Previous works on constructive Fermionic renormalization group have described the interaction by a collection of functions or measures with added indices indicating when derivatives should be applied to fields or propagators, using norms which treat the derivative indices like other indices of the fields (see for example \cite{AGG_part1,AGG_part2,GMR,Duch} and citations therein); this has the result that there are several different functions which are physically indistinguishable and correspond to the same element of the original Grassmann algebra (this is discussed in \cite{GMR} in terms of null fields and in \cite{AGG_part1,AGG_part2} in terms of an {\it ad hoc} equivalency relationship).
Instead of doing so, I introduce a family of seminorms describing the interaction as a sequence of distributions, and which consequently do not involve explicit derivative indices.
This defines a locally convex space of ``potentials'' \pot\footnote{It would also be possible to use a Banach space as in \cite{GMR} or a family of Banach spaces as in \cite{Duch}, but this involves further restricting the space in a way which depends on a parameter which I find more natural to introduce later in the demonstration.}, which consists of a subset of the elements of \Grassmann\ with a finer topology than that of the original space.  
I then rewrite the flow equation in a form which is more adapted to the structure of the space \pot, and prove some preliminary estimates about the elements of the resulting expression.  
I also set up the flow equation for the Wick ordered interaction implementing \cref{eq:flow_Wick_formal} in this context.

In \cref{sec:relevant_parts} I introduce my notion of local parts, and prove scaling estimates which verify that this gives an appropriate decomposition into (local) relevant and irrelevant components\footnote{As in \cite{Pol} I will not distinguish marginal terms from relevant ones, so strictly speaking ``relevant'' should be read ``relevant or marginal'' here and later.}.  
Well-informed readers may worry that by doing away with derivative indices I lose track of the fact that derivatives change the scaling dimension of fields/propagators, which is the reason that the remainder after removing the local relevant part of the interaction is dimensionally irrelevant.  In fact, the usual ``interpolation'' estimates can be easily adapted to this new framework, as shown in \cref{lem:f_cR_scaling,lem:fI_few_scaling}.
Furthermore, since I do not introduce derivative indices the irrelevant part is equal to the full interaction minus the local part (rather than a physically equivalent but unequal element of the abstract space of interactions) a trivial bound, \cref{eq:trivial_interpolation_bound}, avoids an apparent divergence in these estimates which was one of the obstacles to taking a continuum limit for the RG flow in \cite{GMR} and required a special treatment of so-called ``dangerous 1-particle irreducible'' graphs in \cite{DR}, complicating the construction and leading to stronger restrictions on the form of $G_\Lambda$ than are actually necessary.

In \cref{sec:fp_construction} I present the central part of the construction of the fixed point.  This is based on a tree expansion adapted from \cite{DR}, which seems to have been the first rigorous construction of a nontrivial continuous RG flow (albeit one with only a trivial fixed point in the region under examination), and perhaps the only one apart from a few subsequent works by the same authors \cite{DR.2D.letter,DR.2D.1,DR.2D.2} until very recently (the preprints \cite{Duch,KMS.localSolution} appeared shortly after the first version of the present manuscript).  This tree expansion (considered together with scale labels) combines features of the Battle-Brydges-Federbush expansion with an expansion in Gallavotti-Nicolò trees into a single expansion; as well as being a more compact formulation this allows better combinatorial bounds which are important to avoiding the additional obstacle to taking a continuum limit in \cite[Appendix~J]{GMR}.

Finally in \cref{sec:conclusions} I conclude the construction of the fixed point and discuss some conclusions and prospects for future work.

\bigskip

Before I proceed let me pause to make some further comparisons with the existing literature.  
As mentioned above, similar scaling solutions were constructed in \cite{GMR} for the same model with a discrete renormalization group flow.  Modulo some technical issues, these correspond to nontrivial solutions of \cref{eq:flow} with the property that
\begin{equation}
	\cV_\Lambda = (\fD_\gamma)^n \cV^*
	, \quad \forall n \in \bZ
	\label{eq:discrete_scaling}
\end{equation}
for some $\gamma \in \bR$.  The construction in question holds for a range of $\gamma$, so these solutions should in fact have the form in \cref{eq:scaling_solution}, but this is not clear because it is not clear that this construction gives the same formulation for different values of $\gamma$, due to complications arising from the treatment of the explicit derivative indices \cite[Remark~6.4]{GMR}.
Unfortunately the proof I give here -- using the tree expansion as an ansatz -- cannot be guaranteed to give all solutions characterized by any more explicit criterion, so as things stand it has not been rigorously established that the fixed points I construct here are equivalent to the ones constructed in \cite{GMR}.  
It should be possible to  overcome this problem by reformulating the construction in \cite{GMR} in terms of the spaces and norms defined in \cref{sec:Grassmann,sec:formalism}.  Since the construction of these solutions was based on the solution of an implicit function problem, the outcome will be the unique solution with a certain set of properties, which should then imply that the solutions of the original discrete problem with different $\gamma$ all correspond to the same element of the space $\Grassmann$ defined in \cref{sec:Grassmann} below, as does the solution I construct.

\medskip

The first rigorous result on the existence of nontrivial solutions to the Polchinski equation was \cite{BK87}, which showed the existence of local solutions (that is, solutions on a compact interval of scales) for Bosonic models with a bounded initial datum (such as the Sine-Gordon model, but not the $\phi^4$ model; see in particular \cite[Equation~(2.6)]{BK87}).  This was followed by an attempt to prove a similar result for the Fermionic case where the relevant boundedness assumption is much less restrictive \cite{BW.local} which was initially unsuccessful \cite{BW.erratum} but has recently been completed by a different set of authors \cite{KMS.localSolution}.  It is probably possible to combine the latter result with constructions of discrete-time renormalization group flows to generate global solutions of the Polchinski equation (although this can probably also be done by simply taking the derivative of the discrete RG transformation with respect to the step size).  Doing so with the fixed point of \cite{GMR} could plausibly lead to an alternative proof of \cref{thm:main_FP}; this involves some subtleties related to the distinction between equality and equivalence there but at worst this should be resolvable using the techniques of \cref{sec:formalism} below.

In 2000 Disertori and Rivasseau \cite{DR}, as mentioned above, presented a global solution for the two-dimensional Gross-Neveu model using a tree expansion strategy \cite[Section~IV.2]{DR} to control the irrelevant parts of the effective action (which, as so often in constructive renormalization group work, is where most of the complicating details of the proof are found).  That strategy is based on a single tree expansion (in the sense of the Battle-Brydges-Federbush formula), rather than repeated tree expansions as in \cite{AGG_part1,AGG_part2,BM01,GMR,GMRS.correlations}, among many others.  
I use a version of this expansion, presented in \cref{sec:tree_def} below, however it plays a different role in the current proof than in its original context, and reusing it in this way is probably the main element which made some form of the main result possible.
The subject of \cite{DR} is the construction of the two-dimensional Gross-Neveu model via a continuous renormalization group procedure, which gives a family of solutions of the Polchinski equation for $\Lambda \in [1,\infty)$ (an ``ultraviolet limit''); so \cite{DR} starts with a regularized path integral, uses a forest formula \cite{AR.botany} to derive a tree expansion for the vertex functions (that is, for the effective action) which is then seen to be absolutely convergent, uniformly in the ultraviolet limit, and to satisfy the Polchinski equation.
I instead introduce the  tree expansion as part of an ansatz for the irrelevant part of a scaling solution of the Polchinski equation.  

The most important feature of this expansion is that makes it possible to substantially reduce the number of terms obtained when further expanding the resulting expression to obtain a sum over products of Pfaffians, 
corresponding to the combinatorial factors discussed which were one of the obstacles (along with the apparent divergence of the ``interpolation'' estimates) to taking a continuum limit of the discrete flow defined in \cite{GMR}.
As discussed in \cite[Appendix~J]{GMR}, carefully controlling the number of terms in this expansion makes it possible to construct a fixed point for $\gamma$ arbitrarily close to $1$, but with estimates that deteriorate in such a way to make it impossible to control the limit $\gamma \to 1^+$.
The expansion I write in \cref{eq:C_ell_tau_explicit} initially contains a similar problematically large number of terms; 
as I detail in \cref{sec:regroup}, regrouping these terms in a way which preserves the Pfaffian structure is nontrivial because of the need to extract the relevant part of the effective interaction at all scales, but it can be done following the procedure of \cite[Section~IV.2]{DR}.
Using the seminorms introduced in \cref{sec:formalism}, I am also able to radically simplify the treatment the ``interpolated'' part of the effective action compared to \cite{DR}, so that \cref{eq:C_sum_integral_bound_basic}, the counterpart of \cite[Equation~IV.98]{DR}, follows from the decomposition along with the bounds in \cref{sec:relevant_parts} as soon as the number of terms is controlled, corresponding to eliminating nearly all of \cite[Section~IV.3]{DR}.
The first part of this simplification is a result of bounding the effect of interpolation (corresponding to removal of the relevant part here, and as cancellation with a part of the counterterm in \cite{DR}) iteratively in terms of a weighted norm, which eliminates the need for a careful accounting of the length of tree lines (parallel the simplification of the comparable treatment in \cite{AGG_part2,GMR} compared to \cite[Section~3.3]{BM01} in the case of a discrete renormalization group transform).
In fact this makes it possible to handle the interpolated part of the effective action outside of the tree expansion (in \cref{thm:interpolated_stationary}), which allows some modest simplifications to the definition of the decomposition in \cref{sec:regroup} compared to \cite{DR}.

Additionally and more consequentially for the outcome of the analysis, considering the interactions as distributions rather than functions allows for some simplifications which also weaken the assumptions which would otherwise be needed. 
Using \cref{lem:fI_few_scaling} (and in particular the trivial bound which follows from \cref{eq:trivial_interpolation_bound}) corresponds to a version of \cite[Lemma~6]{DR} with a constant bound rather than a divergent one, and similarly a version of \cite[Lemma~8]{DR} without the first product on the right hand side of \cite[Equation~(IV.84)]{DR}; as a result the factors with negative exponents in \cite[Equation~(IV.93)]{DR} are absent and there is no need for \cite[Lemma~9]{DR}.  
This is particularly helpful because the degree of divergence in the terms mentioned above depends on the decay of $G_\Lambda (x)$ at long distances, which in \cite{DR} was taken to be exponential, corresponding to the Fourier transform of an analytic function, which is quite restrictive since analytic functions cannot have compact support (it excludes the cutoffs used in \cite{GMR}, although the compact support property appears to have only played a minor role there).  A more careful analysis in \cite[Section~III.3.2]{DR.2D.2} found that an equivalent of \cite[Lemma~9]{DR} holds with a decay like $\exp(-|\Lambda x|^{1/s})$ (corresponding to the Fourier transform of a function of Gevrey class $s$) for $s < 3/2$, which would include some of the cutoffs considered in \cite{GMR} but is still an annoying restriction.  
On top of this, the lemmas referred to above depend on some details of the analysis being carried out, essentially via the integer $R$ introduced in \cref{sec:formalism}.  Admittedly the equivalent of $R$ can be taken to be $2$ in all of the specific examples I am aware of in the constructive renormalization group literature, but relying too much on this could create a pitfall for future work.

\medskip

In a recent preprint \cite{Duch}, Paweł Duch revisited the problem studied in \cite{DR} and provided an alternate proof without the use of tree formulas or cluster expansions.  
This is very promising for future developments, however for now it seems to have led to a work which is still complicated even by the standards of the field.  
Part of this complication appears to be due to the fact that (in the currently cited version of \cite{Duch}, at least) although Duch considers interactions as distributions (apparently in order to include Dirac deltas in a less {\it ad hoc} fashion than, for example, \cite{GMR}), he includes a derivative index leading to a topology which includes unphysical distinctions between equivalent representations of the same interaction, as in \cite{GMR}, also leading to a restriction on the cutoff similar to that discussed above for \cite{DR.2D.2}.
Although Duch does not study a fixed point in the RG sense\footnote{The fixed point discussed in \cite[Section~13]{Duch} is the solution of an implicit function problem using the Banach fixed point theorem, as in the construction of solutions of ordinary Cauchy problems by Picard iteration (\cref{thm:interpolated_stationary} below is also based on a similar argument, as is \cite[Theorem~6.1]{GMR}, where however the result does actually correspond fairly directly to a fixed point in the RG sense).  This has the same basic structure as the construction of the running coupling constants in \cite[Section~4]{BM01} and many subsequent works where it arguably occupies a similar position as in the proof of the main result of \cite{Duch}.}
it seems likely that a combination of Duch's techniques (especially the introduction of the norm of the space $\mathscr{W}_{\tau,\varepsilon}^{\alpha,\beta;\gamma}$ in an unnumbered equation on \cite[21]{Duch}, which unlike the other norms involved does not seem to me to have any analogue in \cite{GMR} or the present work) with those I present here can produce proofs which are significantly simpler and more flexible than either.

\section{The model under consideration}
\label{sec:example}
As in \cite{GMR}, I consider a model of an $N$-component Fermionic field in $d \in \{1,2,3\}$ dimensions, intended as a renormalized version of the one defined by the action
\begin{equation}
	\int
	\left[ 
		(\psi, \Omega \otimes (- \mathscr L)^{\tfrac14 d + \tfrac12 \varepsilon} \psi) (x)
		+ 
		\nu (\psi, \Omega \psi) (x)
		+
		\lambda (\psi, \Omega \psi)^2 (x)
	\right]
	\dd^d x
	,
	\label{eq:bare_action}
\end{equation}
where $\psi(x)$ is an $N$-component vector which transforms as a scalar under rotations, $\mathscr L$ is the Laplacian\footnote{I use this symbol to avoid confusion with $\Delta$, the generator of continuous dilations introduced above} on $\bR^d$, $\Omega$ is the antisymmetric $N \times N$ matrix
\begin{equation}
	\Omega
	:=
	\begin{pmatrix}
		0 & 1 & 0 & 0 & \dots 
		\\
		-1 & 0 & 0 & 0 & \dots
		\\
		0 & 0 & 0 & 1 & \dots 
		\\
		0 & 0 & -1 & 0 & \dots
		\\
		\vdots & \vdots & \vdots & \vdots & \ddots
	\end{pmatrix}
	\label{eq:Omega_matrix}
\end{equation}
(hence $N$ is even, and for reasons that will become clear later we also take $N \neq 8$) and $\varepsilon \in [0,d/6)$ is a nonnegative real parameter; see \cite{FP.FractionalLaplacian} for a discussion of the fractional Laplacian in the context of quantum field theory.  
A quadratic part with a variable power (originally introduced by Gawędzki and Kupiainen to define a version of the Gross-Neveu model with a nontrivial fixed point, and simultaneously by Felder \cite{Felder} in the context of Bosonic fields\footnote{In \cite{Felder}, like \cite{GMR}, the theory is introduced without reference to the fractional Laplacian, which however appears in \cite{GK85}.}) gives a continuously varying scaling dimension for the appropriate field as in dimensional regularization, without the associated problem of defining spaces whose dimension is an arbitrary real number.
As we shall see shortly, for this particular choice of the power of the fractional Laplacian the scaling dimension of the field (denoted $[\psi]$, see \cref{eq:propagator_scaling}) takes the value $\tfrac14 d - \tfrac12 \varepsilon $, 
so that the quartic term transforms under dilations (see \cref{eq:dilation_cW}) as
\begin{equation}
	\fD_\Lambda 
	\int
	(\psi, \Omega \psi)^2 (x)
	\dd^d x
	=
	\Lambda^{2 \varepsilon}
	\int
	(\psi, \Omega \psi)^2 (x)
	\dd^d x
	\label{eq:quartic_scaling}
\end{equation}
(note $d - 4[\psi] = 2 \varepsilon$)
and thus is stable or growing (\emph{(dimensionally) relevant} in the usual terminology of the renormalization group) for $\varepsilon \ge 0$.

Note that this action is translation and rotation invariant, and has an internal symplectic symmetry; if we restrict to $\varepsilon < d/6$ the only local polynomials in $\psi$ and its gradients with these symmetries transforming as in \cref{eq:quartic_scaling} with a nonnegative exponent are the ones already present in \cref{eq:bare_action}; for example for the sextic term
\begin{equation}
	\fD_\Lambda 
	\int
	(\psi, \Omega \psi)^3 (x)
	\dd^d x
	=
	\Lambda^{3 \varepsilon - \tfrac{1}{2} d}
	\int
	(\psi, \Omega \psi)^3 (x)
	\dd^d x
	,
	\label{eq:sextic_scaling}
\end{equation}
so that the exponent (or scaling dimension) is negative and this term is \emph{irrelevant} for $\varepsilon < d/6$ (hence the restriction to this range of values of $\varepsilon$).

\medskip

To define the flow I need to specify a choice of $G_\Lambda$, or in other words choose a cutoff.
The simplest choice I am aware of with the properties I will later need closely resembles that of \cite{DR,GK85,Felder}. 
It is closely related to the heat kernel, which makes it possible to extend to a wide variety of domains and boundary conditions, separating contributions geometrically along the same lines as \cite{Greenblatt.zeta}; I will return to this point in a subsequent work.

Let
\begin{equation}
	p^\hk(x)
	:=
	\int \frac{\dd^d k}{(2 \pi)^d} \hat p^\hk (k) \ e^{i \; k \cdot x}
	, \quad
	\hat p^\hk(k)
	:=
	\frac{1}{\Gamma\left( 1+\tfrac{d}4 + \tfrac{\varepsilon}2 \right)}
	\int_1^\infty
	t^{\tfrac14 d + \tfrac12 \varepsilon} 
	k^2 e^{-t k^2}
	\dd t
	.
	\label{eq:P_def}
\end{equation}
Noting that
\begin{equation}
	\int_0^\infty
	t^{\frac14 d + \frac12 \varepsilon} 
	k^2 e^{-t k^2}
	\dd t
	=
	\frac{\Gamma\left( 1+\tfrac{d}4 + \tfrac{\varepsilon}2 \right)}{(k^2)^{\frac{d}4 + \frac{\varepsilon}2}}
\end{equation}
for $k^2 > 0$,
$\hat p^\hk(k)$ (like $\hat P(k)$ of \cite[Equation~(2.1)]{GMR}) is a version of $\left|k\right|^{-\frac12 d - \varepsilon}$ (the Fourier transform of the inverse fractional Laplacian) with an ultraviolet cutoff; note that \cref{eq:P_def} makes $\hat p^\hk$ analytic for $\Re k^2 > 0$. 

Since
\begin{equation}
	\int k^2 e^{-t k^2} \dd^d k
	=
	\frac{d}{2} \frac{\pi^{d/2}}{t^{1+d/2}}, \quad t > 0
\end{equation}
for $\varepsilon < \tfrac12 d$ we can interchange the $t$ and $k$ integrals in \cref{eq:P_def}.  Changing the variable of integration $t=1/\Lambda^2$ to have a momentum scale,
\begin{equation}
	p^\hk(x) = 
	\int_0^1 g^\hk_\Lambda(x) \dd \Lambda
	, \qquad
	g^\hk_\Lambda(x) 
	:=
	\int \frac{\dd^d k}{(2 \pi)^d} 
	\hat g^\hk_\Lambda(k)
	e^{i \; k \cdot x}
	,
	\label{eq:gt_def}
\end{equation}
with
\begin{equation}
	\hat g^\hk_\Lambda(k)
	:=
	\frac{2 \Lambda^{-1 - \frac{1}{2} d - \varepsilon}}{\Gamma\left( 1+\tfrac{d}4 + \tfrac{\varepsilon}2 \right)}
	\left( \frac{k}{\Lambda} \right)^2
	e^{-(k/\Lambda)^2}
	,
	\label{eq:ghat_t_def}
\end{equation}
these have the property that 
\begin{equation}
	g^\hk_\Lambda (x)
	=
	\Lambda^{-1 + \tfrac12 d - \varepsilon}
	g^\hk (\Lambda x),
	\label{eq:g_scaling_check}
\end{equation}
where $g^\hk := g^\hk_1$,
cf.\ \cref{eq:propagator_scaling}; note that the exponent indeed corresponds to $[\psi] = \tfrac14 d - \tfrac12 \varepsilon$, which is positive for $\varepsilon < \tfrac12 d$.

Since 
\begin{equation}
	\begin{split}
		\int \frac{\dd^d k}{(2 \pi)^d} 
		k^2 e^{- k^2}
		e^{i \; k \cdot x}
		&
		=
		\left.
		\frac{\partial}{\partial t}
		\int \frac{\dd^d k}{(2 \pi)^d} 
		e^{- k^2}
		e^{i \; k \cdot x}
		\right|_{t=1}
		\\ &
		=
		\left.
			\frac{\partial}{\partial t}
			\frac{
				e^{-x^2/4 t}
			}{
				(4 \pi t)^{d/2}
			}
		\right|_{t=1}
		=
		\frac{x^2 - 2 d}{2^{d+2} \pi^{d/2}}
		e^{-x^2/4 }
		,
		\end{split}
\end{equation}
we also have the more explicit formula
\begin{equation}
	g^\hk (x)
	=
	\frac{x^2 - 2d }{2^{d+1} \pi^{d/2}\Gamma\left( 1+\tfrac{d}4 + \tfrac{\varepsilon}2 \right)}
	e^{-x^2/4 }
\end{equation}
and
\Cp{gt}
\begin{equation}
	g^\hk_\Lambda(x)
	=
	\frac{
		\Lambda^{-1 + \tfrac12 d - \varepsilon} (\Lambda^2 x^2 - 2d  )
}{2^{d+1} \pi^{d/2}\Gamma\left( 1+\tfrac{d}4 + \tfrac{\varepsilon}2 \right)}
	e^{-\Lambda^2 x^2/4 }
	=:
	\Cr{gt} 
	\Lambda^{-1 + \tfrac12 d - \varepsilon} 
	(\Lambda^2 x^2 - 2d  )
	e^{-x^2/4 t}
	;
	\label{eq:gt_explicit}
\end{equation}
note the (super)exponential decay
\Cp{g_decay}
\begin{equation}
	\left|g^\hk_\Lambda(x)\right|
	\le
	\Cr{g_decay}
	\Lambda^{-1 + \tfrac12 d - \varepsilon} 
	e^{-\Lambda^2 x^2/5}
	.
	\label{eq:gt_decay}
\end{equation}
The propagator involves an internal index taking $N$ values, via the antisymmetric $N \times N$ matrix $\Omega$ defined in \cref{eq:Omega_matrix}.
The full infinitesimal propagator (including the internal indices) is $G^\hk_\Lambda := \Omega g^\hk_\Lambda$, and the full UV-cutoff propagator is $P^\hk_\Lambda := \int_0^\Lambda P_\Theta \dd \Theta = \Omega_{ab} p^\hk_\Lambda$.

Among the alternatives is $G^{\chi}$ corresponding to replacing $\hat p^\hk(k)$ in \cref{eq:P_def} with $\hat P(k)$ of \cite{GMR}, that is
\begin{equation}
	\hat g^{\chi}_\Lambda (k)
	:=
	\frac{\dd}{\dd \Lambda}\left[ \frac{\chi(|k|/\Lambda)}{\left|k\right|^{d/2+\varepsilon}} \right]
	=
	-\frac{|k|}{\Lambda^2}
	\,
	\frac{\chi'(|k|/\Lambda)}{\left|k\right|^{d/2+\varepsilon}}
	=:
	\frac{1}{\Lambda}
	\frac{\tilde \chi(|k|/\Lambda)}{\left|k\right|^{d/2+\varepsilon}}
	\label{eq:ghat_chi_def}
\end{equation}
where $\chi:\bR \to \bR$ is an even function in the Gevrey class $s \in (1,\infty)$%
\footnote{A function $f:\bR \to \bR$ of compact support is in the Gevrey class $s$ (or ``is Gevrey-$s$'') iff there is a constant $C=C(f)$ such that
	\begin{equation}
		|f^{(n)}(x)|
		\le 
		C^{n+1}
		(n!)^s
		\quad
		\forall x \in \bR, n=0,1,2,\dots
		,
		\label{eq:Gevray_def}
	\end{equation}
	and a function $\bR^d \to \bR$ is if the same holds for all partial derivatives.
	Of course all such functions are $C^\infty$, and any derivative of $f$ is always as member of the same Gevrey class.
	
	Examples of Gevrey-$s$ functions with the properties required for $\chi$ are given in \cite[Appendix~A.1]{GMR}.
}
with $\chi([0,\tfrac12])=\{1\}$ and $\chi([1,\infty))=\{0\}$; then $\tilde \chi$ is also a Gevrey-$s$ function with compact support, and furthermore it vanishes on $[0,\tfrac12]$ so that $g^{\chi}_\Lambda (k)$ is also Gevrey-$s$ with compact support.  Then defining $g^\chi_\Lambda$ by the inverse Fourier transform as before gives the same scaling property as in \cref{eq:g_scaling_check} and proceeding as in \cite[Appendix~A.2.2]{GMR} there are are $\chi$-dependent constants $\Cr{g_chi_pre}, \Cr{g_chi_rate}$ such that
\Cp{g_chi_pre} \Cp{g_chi_rate} 
\begin{equation}
	\left|g^\chi_\Lambda(x)\right|
	\le
	\Cr{g_chi_pre}
	\Lambda^{-1 + \tfrac12 d - \varepsilon} 
	e^{-\Cr{g_chi_rate} |\Lambda x|^{1/s}}
	,
	\label{g_chi_decay}
\end{equation}
and similar for their derivatives of given order.

\medskip

I will show the following:
\begin{theorem}
	For $d=1,2,3$ and for $G_\Lambda = G^\hk_\Lambda$ or $G^\chi_\Lambda$ with any $\chi$ of the form discussed above, there exists $\varepsilon_0 = \varepsilon_0 (d,G_\cdot)$ such that for each $\varepsilon \in (0,\varepsilon_0]$ there exists $\cV^* \neq 0$ such that $\Lambda \mapsto \fD_\Lambda \cV^*$ is a solution of the Polchinski equation~\eqref{eq:flow} on $(0,\infty)$.
	\label{thm:main_FP}
\end{theorem}

A number of other properties of this fixed point follow from the construction; for example it has the same symplectic, rotation, and translation symmetries as \cref{eq:bare_action}, and as a result its relevant part has the same form.  
A fully precise version of these statements is cumbersome, so I postpone it to the conclusion (\cref{sec:conclusions}).

The culmination of the proof is the construction and solution of a set of exact equations for the relevant part, which admits an asymptotic treatment which can be viewed as a kind of epsilon expansion \cite{WilsonKogut} (cf.\ \cite{GK85,GMR}).
Let me introduce this at first with the following approximate version of the fixed point problem.
If I restrict to interactions with the same symmetries as \cref{eq:bare_action}, as noted above the only local relevant terms are $\int \sum_{a,b}\Omega_{ab} \psi_a (x) \psi_b(x) \dd x$ and $\int [\sum_{a,b}\Omega_{ab}\psi_a(x) \psi_b(x)]^2$.  Denoting their coefficients in the Wick-ordered interaction by $\tilde \nu$ and $\lambda$ respectively, discarding all irrelevant terms the fixed point version of \cref{eq:flow_Wick_formal} becomes a pair of coupled equations for $\tilde \nu,\lambda$ in a second order approximation, which can be written
\begin{equation}
	\begin{cases}
		(d - 2 [\psi]) \tilde \nu
		= 
		-48 (N+2) \lambda^2 \int g(x) p^2 (x) \dd^d x
		, 
		\\
		(d-4[\psi]) \lambda
		=
		8(8-N) \lambda^2 \int g(x) p(x) \dd^d x
	\end{cases}
	\label{eq:truncated_FP_prelim}
\end{equation}
after some slightly tedious calculations about the factors of $\Omega$; note the presence of factors of $p$ associated with the Wick ordering.
Since we have assumed $N \neq 8$, and since
\begin{equation}
	\int g(x) p(x) \dd^d x
	=
	\int \hat g(k) \hat p(-k) \dd^d k
	> 0
	,
\end{equation}
discarding the trivial solution $\lambda = \tilde \nu =0$, \cref{eq:truncated_FP_prelim} becomes
\begin{equation}
	\begin{cases}
		\tilde \nu
		 = 
		-\frac{96}{d+2\varepsilon} (N+2)  \int g(x) p^2 (x) \dd^d x
		\ \lambda^2
		,
		\\
		\lambda
		=
		4 \varepsilon \ \big/\left[ (8-N)\int g(x) p(x) \dd^d x \right] 
		.
	\end{cases}
	\label{eq:truncated_FP}
\end{equation}
Note that the integrals depend on $\varepsilon$ via $p$, but the dependence is smooth for $\varepsilon < d/2$, and with nonzero limits for $\varepsilon \to 0$.  
Compared to the discussion in \cite[Section~3]{GMR}, note that integral in the expression for $\lambda$ includes a part corresponding to the ``semilocal'' contribution in \cite[Equations~(3.6-7)]{GMR}.  Also, although $\tilde \nu = O(\varepsilon^2)$, this is because I am considering the Wick ordered interaction; taking this into account, the full quadratic coefficient is
\begin{equation}
	\nu = \tilde \nu
	+ (2 N +2) p(0) \lambda
	.
\end{equation}

I will construct functions $B_{\chi} (\lambda,\wt \nu)$, $\chi = \wt \nu, \lambda$ (see \cref{eq:B_implicit_def}) such that
\begin{equation}
	\begin{cases}
		(d - 2 [\psi]) \tilde \nu
		= 
		-48 (N+2) \lambda^2 \int g(x) p^2 (x) \dd^d x
		+
		B_{\wt \nu} (\wt \nu, \lambda)
		, 
		\\
		(d-4[\psi]) \lambda
		=
		8(8-N) \lambda^2 \int g(x) p(x) \dd^d x
		+
		B_{\lambda} (\wt \nu, \lambda)
	\end{cases}
	\label{eq:beta_with_error}
\end{equation}
is a sufficient condition for $\wt \nu,\lambda$ to correspond to the relevant part of a fixed point of the full RG flow, with the irrelevant part given by a convergent power series in $\lambda$ and $\tilde \nu$; 
and furthermore 
\begin{align}
	\left| B_\chi (\wt \nu,\lambda) \right|& \le \C (|\lambda| + |\tilde \nu|)^3
	\label{eq:B_small}
	,
	\\[1ex]
	\left| B_\chi (\wt \nu,\lambda)-  B_\chi (\wt \nu',\lambda')\right|& \le \C (|\lambda| + |\tilde \nu|)^2 (|\wt \nu - \wt \nu'|+ |\lambda - \lambda'|)
	\label{eq:B_Lipschitz}
\end{align}
uniformly for small $\varepsilon$,
from which it is easy to see that \cref{eq:beta_with_error} has a solution close to \cref{eq:truncated_FP}, or more precisely
\begin{equation}
	\tilde \nu = \cO(\varepsilon^2)
	,
	\quad 
	\lambda 
	=
	\left. \left( 4  \ \Big/\left[ (8-N)\int g(x) p(x) \dd^d x \right] \right) \right|_{\varepsilon = 0} \varepsilon
	+ \cO( \varepsilon^2)
	\label{eq:rel_FP_asymp}
\end{equation}
for $\varepsilon \to 0^+$; in particular this is a nontrivial fixed point (with $\lambda \neq 0$) for small strictly positive $\varepsilon$.

\section{The Grassmann algebra}
\label{sec:Grassmann}
In this section I formally define the Grassmann algebra alluded to in \cref{eq:flow}, which is also a topological vector space, so that derivatives with respect to a real parameter have a natural meaning.  The other operations on the right hand side of \cref{eq:flow} are also well defined, as is the dilation operation, so this gives a clear notion of solutions of the Polchinski equation (and scaling solutions, in particular) without the additional structure introduced in later sections, for example without reference to the norms which I will later use to construct a specific solution.

This Grassmann algebra is intended to be extremely large; in particular it corresponds to an extremely weak assumption on the locality and regularity of the interaction.  Correspondingly, the notion of convergence (and so of differentiability) is extremely weak.

Note that this section along with \cref{sec:formalism,sec:fp_construction,sec:relevant_parts} below are stated for a generic model apart from a few comments, both in the hope that they may be useful elsewhere in the future and to make it easier for the reader to compare with other techniques.

Let $\intrange{a}{b} := [a,b] \cap \bZ$ denote the set of integers from $a$ to $b$.
For each positive integer $n$ let $\testn$ be the set of smooth, compactly supported functions $\intrange1N^n \times (\bR^d)^n \to \bC$, with the locally convex topology\footnote{In the sense of \cite[Chapter~V]{ReedSimon1}, that is the topology defined by convergence of each of the seminorms, without any uniformity.} defined by the seminorms 
\begin{equation*}
	\tn_{n,K,\ul m, \bs a}[f]
	:=
	\sup_{x \in K}
	\left|
		\partial^{\ul m}
		f_{\bs a} (\bs x)
	\right|
\end{equation*}
for all $\bs a \in \intrange1N^n $, $K \subset  (\bR^d)^n$ compact, and all $\ul m \colon \intrange1n \times \intrange1d \to \bN$ (which I will call a multi-index; note that I use the convention $0 \in \bN$), where I use the notation
\begin{equation}
	\partial^{\ul m}
	f_{\bs a} (\bs x)
	:=
	\frac{\partial^{\ul m}}{\partial \bs x^{\ul m}}
	f_{\bs a} ( \bs x)
	:=
	\frac{\partial^{m_{11}}}{\partial x_{11}}
	\dots
	\frac{\partial^{m_{nd}}}{\partial x_{nd}}
	f_{\bs a} ( \bs x)
	\label{eq:bf_continuity}
\end{equation}
for the partial derivatives.  Apart from the presence of the index $\bs a$, this is the usual topology of test functions used in the definition of distributions

Let $\testn^*$ be the (continuous) dual of $\testn$, i.e.\ the set of continuous  linear maps $\testn \to \bC$, with the topology of weak convergence defined by stipulating that $\lim_{k \to \infty} V_k = V$ means that $V_k [f] \to V[f]$ for all $f \in \testn$.

\begin{definition}
	Adopting the convention that any permutation $\pi \in \Pi_n$ acts on all $n$-tuples as $\pi (\bs x) = (x_{\pi(1)},\dots,x_{\pi(n)})$, and on $\testn$ as $\pi f_{\bs a} (\bs x) := f_{\pi(\bs a)}(\pi(\bs x))$, I say that $W \in \testn^*$ is \emph{antisymmetric under perturbations} (\emph{antisymmetric} for short) if $W[\pi f_{\bs a}] = \sign(\pi) W[ f_{\bs a}]$ for all permutations $\pi$.

	\label{def:antisymm}
\end{definition}

\begin{definition}
	For $n \ge 1$, let \Grassmannn\ denote the set of antisymmetric elements of $\testn^*$, and let $\Grassmann_0 := \bC$.

	Let $\Grassmann$ denote the topological product of $\Grassmann_n$ over $n \in \bN$, with the natural vector space structure and the product defined by 
\begin{equation}
		(\cV \wedge \cW)_n [f]
		:=
		\sum_{k=0}^n
		\sum_{\pi \in \Pi_n}
		\frac{\sign(\pi)}{n!}
		V_k\left[ \bs a, \bs x \mapsto W_{n-k}[\bs b, \bs y \mapsto (\pi f)_{\bs a, \bs b}(\bs x , \bs y)] \right]
		,
		\label{eq:wedge_def}
	\end{equation}
	with the obvious shorthand in the arguments of $f \in \testn$.
	\label{def:Grassmann}
\end{definition}
(For the benefit of the reader who finds this expression to be overly terse, the argument of $V_k$ in \cref{eq:wedge_def} is the function obtained by acting $W_{n-k}$ on $\pi f$ understood as a function of only $\bs b, \bs y$, with the result then acted on by $V_k$ as a function of the remaining parts $\bs a, \bs x$ of the argument.)

This means that each $\cV \in \Grassmann$ is an infinite sequence $(V_0,V_1,\dots)$, and that convergence $\cV^{j} \to \cV$ is defined by $V^j_n [f] \to V_n[f]$ for each $f \in \testn$, $n \in \bN_+:= \bN \setminus \{0\}$, and $V^j_0 \to V_0$ as complex numbers.
Each $\cV \in \Grassmann$ can also be identified with the map $\bC \cup \bigcup_{n=1}^{\infty} \testn \to \bC$, $f \mapsto V_n[f]$ (understood as $c \mapsto V_0 c$ for $n=0$).

\Grassmann\ is indeed a Grassmann algebra; 
the only subtlety is that associativity of the product requires that there is no issue with changing the order of operations in the right hand side, or in other words
\begin{equation}
	V_k\left[ \bs a, \bs x \mapsto W_{n-k}[\bs b, \bs y \mapsto f_{\bs a, \bs b}(\bs x , \bs y)] \right]
	\equiv
	W_{n-k}\left[ \bs b, \bs y \mapsto V_k[\bs a, \bs x \mapsto f_{\bs a, \bs b}(\bs x , \bs y)] \right]
	\label{eq:distro_product}
\end{equation}
which can be seen by noting first that it follows for modified polynomials of the form
\begin{equation}
	g_{\bs a } (\bs x)
	=
	\sum_{\ul m \in M}
	\gamma_{\bs a, \ul m}
	\prod_{j=1}^n \prod_{k=1}^d
	\left[ \chi_K(x_{jk}) \right]^{m_{jk}}
	\label{eq:mod_poly}
\end{equation}
for $M$ a finite set and $\chi_K$ a smooth function with support $K$, and then noting that these are dense in \testn\ using the Stone-Weierstrass theorem and a diagonal trick along with the observation that all elements of $\testn^*$ are by assumption continuous. 

Since \Grassmann\ is a topological vector space, derivatives and Riemann integrals with respect to real parameters are defined as usual via algebraic operations and limits.  Singling out the elements $\left\{ \psi_a (x) \right\}_{a \in \intrange1N, x \in \bR^d}$ defined by
\begin{equation}
	\psi_a (x) [f]
	=
	\begin{cases}
		f_a (x), & f \in \test_1
		\\
		0,
		& \text{otherwise},
	\end{cases}
	\label{eq:psi_def}
\end{equation}
it is possible to define elements of \Grassmann\ using conventional expressions which might be expected to be only formal; for example 
$ \int  (\psi, \Omega \psi) (x)\dd^d x$ 
(with the conventions of \cref{eq:bare_action}) is an element of $\Grassmann$, characterized equivalently as the map
\begin{equation}
	f
	\mapsto
	\begin{cases}
		\sum_{a,b \in \intrange1N}
		\Omega_{ab}
		\int f_{ab}(x,x)
		\dd^d x
		,& 
		f \in \test_2
		\\
		0
		, &
		\text{otherwise}
		.
	\end{cases}
	\label{eq:mass_term_as_map}
\end{equation}
Finally, identifying the variational derivative $\frac{\delta}{\delta \psi}$ with the linear map $\Grassmann_n \to (\test_1 \times \test_{n-1})^*$ characterized by 
\begin{equation}
	V_n
	\mapsto
	\big[
		g,f \mapsto
		n V_n (g \otimes f)
	\big]
	, \quad
	V_n \in \Grassmann_n, \ n \ge 1
\end{equation}
(taking advantage of antisymmetry) and $V_0 \mapsto 0$, and likewise for repeated derivatives, the right hand side of \eqref{eq:flow} is well defined on at least a nontrivial subspace of $\Grassmann$. 

Finally, the dilation operations introduced in \cref{eq:dilation_cW} generalize as follows.
Let $\wt D_\Lambda: \testn \to \testn$, $f \mapsto (\bs a , \bs x \mapsto f_{\bs a} ( \bs x / \Lambda ))$ and 
\begin{equation}
	\fD_\Lambda \cV : f \mapsto
	\begin{cases}
		V_0 f
		, &
		f \in \test_0
		\\
		D_\Lambda V_n [f] 
		:=
		\Lambda^{-[\psi]n}
		V_n [\wt D_\Lambda f]
		, &
		f \in \testn, n \ge 1
		.
	\end{cases}
	\label{eq:dilation_maindef}
\end{equation}
Recall that $\Delta$ was introduced in \cref{sec:intro} as the generator the one-parameter semigroup obtained by reparameterizing this family of operators as $\Lambda = e^\eta$, which is equivalent to
\begin{equation}
	\Delta \cV 
	:=
	\left.
	\Lambda \frac{\dd}{\dd \Lambda} \fD_\Lambda \cV
	\right|_{\Lambda=1}
	;
	\label{eq:Delta_explicit}
\end{equation}
the derivative can be expressed in terms of the action of $\cV$ on $f$ and its first derivatives with respect to position arguments, and so $\Delta$ is well defined on all of $\Grassmann$. 

\section{Norms, formalism, and estimates used to construct the fixed point}
\label{sec:formalism}

Having introduced the concepts which were left implicit in the statement of \cref{thm:main_FP}, I now turn to the elements that will be needed for the proof.  
I start by making more precise statements of the features of the model and cutoff introduced in \cref{sec:example} which play a role in the proof, then in \cref{sec:potentials} I introduce  a collection of seminorms on a subset \kernel\ of \Grassmann\ which are more suitable for checking convergence, and then introduce a subspace \pot\ of \kernel\ characterized by symmetry properties.
In \cref{sec:flow_formal,sec:Wick} I translate the other key formulae from \cref{sec:intro} into a new notation associated with this space.  In the course of this I prove a number of basic estimates which will play a role later on.

Let me stress that \pot\ is a (vector) subspace of \Grassmann, so that each element of \pot\ is also an element of \Grassmann.  There are no inequivalent representations involving derivative indices (in other words there is no jet extension \cite{Duch}). 
The seminorms give $\pot$ a different topology from that of \Grassmann; but to be precise it is a finer topology (one with a strictly stronger criterion for convergence), so that differentiability in \pot\ implies differentiability in \Grassmann, cf.\ \cref{rem:derivability_kernel}.

\subsection{Preliminaries}
\label{eq:setup_prelim}

Let $R$ be a fixed integer such that $R \in \bN$, $2[\psi] + R - d > 0$ (in the model under consideration $R=2$).  

For $n \in \bN_+$, let \diffyn\ be the set of functions $f \colon \intrange1N^n \times (\bR^d)^n \to \bC$ 
such that for any multi-index $\ul m \colon \intrange1n \times \intrange1d \to \bN$ with $\degree \ul m := \max_{j=1,\dots,n} \sum_{k=1,\dots,d} m_{jk} \le R$, the derivative  
exists and is continuous as a function of $\bs x$.
Denote this set of multiindices by $\mi_{n;R}$.
For $f \in \diffyn$, let 
\begin{equation}
	D_{\bs a}[ f] (\bs x)
	:=
	\max_{\ul m \in \mi_{n,R}}
	\left|
	\frac{\partial^{\ul m}}{\partial \bs x^{\ul m}}
	f_{\bs a} ( \bs x)
	\right|
	.
	\label{eq:D_ell_def}
\end{equation}
Note that the Cartesian product $(f \otimes g)_{\bs a , \bs b} (\bs x , \bs y) := f_{\bs a}(\bs x) g_{\bs b} (\bs y)$ takes $\diffy_\ell,\diffy_m$ to $\diffy_{\ell+m}$, with
\begin{equation}
	D_{\bs a, \bs b} [f \otimes g] (\bs x, \bs y)
	=
	D_{\bs a}[f](\bs x) 
	D_{\bs b}[g](\bs y)
	.
	\label{eq:D_ell_product}
\end{equation}

In addition to the scaling property \eqref{eq:propagator_scaling}, I assume that the propagator $G$ is in $\diffy_2$ and antisymmetric with the matrix indices taken into account, and is translation invariant.  I also assume that there is an increasing, subadditive function 
$\eta : [0,\infty) \to [0,\infty)$ such that 
\begin{equation}
	\lim _{x \to \infty} x e^{-c \eta(x)} = 0
	, \quad
	\forall c >0
	\label{eq:eta_growth}
\end{equation}
for which the propagator satisfies
\begin{equation}
	| G|_+
	:=
	\sup_{\substack{ y \in \bR^d }}
	\int \dd^d x
	\max_{\bs a \in \intrange1{N}^2}
	e^{\eta(|x - y|)}
	D_{\bs a}[G] (x,y)
	<
	\infty
	;
	\label{eq:plus_norm}
\end{equation}
(for $G^\hk$ we can simply take $\eta(x)=x$, for $G^\chi$ instead $\eta(x) = \tfrac12 \Cr{g_chi_rate} x^{1/s}$ where $s$ is the Gevrey class of $\chi$, and $\Cr{g_chi_rate}$ is the same constant appearing in \cref{g_chi_decay}).
Furthermore, I assume that there is a constant $\Cl{gram_g}$ and a Hilbert space containing families of vectors $\gamma_{a,\ul{m}}(x), \ \tilde \gamma_{b,\ul n}(y)$ such that
\begin{equation}
	\partial_x^{\ul m} \partial_y^{\ul n} G_{ab}(x,y)
	=
	\left\langle \tilde \gamma_{a,\ul m}(x),\gamma_{b,\ul n}(y) \right\rangle
\end{equation}
and
\begin{equation}
|\tilde \gamma_{a,m}(x)|^2, \ |\gamma_{b,n}(y))|^2
	\le 
	\Cr{gram_g}
	.
\end{equation}
This is known as a Gram decomposition; as reviewed in \cref{sec:Pfaff} it gives a way of estimating Pfaffians of matrices whose elements are given in terms of propagators, and in the model under consideration (as for most translation-invariant systems) it can be obtained using the Fourier transform.  As an example of the bounds obtained, letting
\begin{equation}
	P(x,y)
	:=
	\int_{0}^1
	G_\Lambda (x,y)
	\dd \Lambda
	=
	\int_0^1
	\Lambda^{2 [\psi] -1}
	G(\Lambda x, \Lambda y)
	\dd \Lambda
	,
	\label{eq:_proper_def}
\end{equation}
evidently $P \in \diffy_2$ and $P$ is antisymmetric; for such functions, let 
\begin{equation}
	\begin{aligned}
		\PPf_n [P] \colon & \intrange1{N}^{2n} \times (\bR^d)^{2n} \to \bC
		, 
		\\ &
		\bs a , \bs x
		\mapsto
		\frac{1}{2^n n!}
		\sum_{\pi \in \Pi_{2n}} (-1)^{\pi}
		\prod_{j=1}^{n}
		P_{a_{\pi(2j-1)},a_{\pi(2j)}} (x_{\pi(2j-1)},x_{\pi(2j)})
	\end{aligned}
	\label{eq:PPf_def}
\end{equation}
(here writing the internal indices explicitly)
so that $\PPf_n[P]$ is a function in $\diffy_{2n}$ whose value is the Pfaffian of a matrix with entries given by values of $P$.
Then letting $C_\GH := \Cr{gram_g}/2[\psi]$,
using \cref{lem:gram_Lambda_integral,thm:gram} in \cref{sec:Pfaff}
\begin{equation}
	\max_{\bs a \in \intrange1{N}^{2n}}
	\sup_{x \in (\bR^d)^{2n}}
	D_{\bs a} [ \PPf_n [P] ]
	\le
	C_\GH^n 
	.
	\label{eq:good_Pfaff}
\end{equation}

Finally, I assume that I am given a collection $\symm$ of functions $S$ with the properties that
\begin{enumerate}
	\item $S$ is an automorphism of each $\diffyn$.
	\item For every $f \in \diffy_m, g \in \diffy_n$ ($m,n =1,2,\dots$), $S[f \otimes g] = S[f] \otimes S[g]$.
	\item The propagator satisfies $S[G_\Lambda] = G_\Lambda$.
	\item $S$ commutes with dilations, $S[\widetilde D_\Lambda f] = \widetilde D_\Lambda S[f]$.
\end{enumerate}
For the current case, I consider \symm\ to consist of rotations, acting as $ (S[f])_{\bs a} (\bs x) = f_{\bs a} (\mathbf Rx_1,\dots,\mathbf R x_n)$ for  $\mathbf R \in SO(3)$
(corresponding to $\psi$ being a scalar in \cref{sec:example}) and the symplectic transformations acting as $ (S[f])_{\bs a} (\bs x) = \sum_{b_1,\dots,b_n} M_{b_1,a_1} \dots M_{b_n,a_n} f_{\bs b} (\bs x)$ for some some symplectic matrix $M$ (that is, satisfying $M^T \Omega M = \Omega$).
I could also include translations in \symm, but I prefer to consider consider translation invariance separately as it plays a more important role.
Until \cref{sec:concluding_ansatz} everything that follows would make sense taking a smaller set as \symm\ (even $\symm = \emptyset$), which would lead to considering a set of interactions characterized by fewer symmetries, and hence a larger (but still finite-dimensional) space of local relevant terms.

\subsection{Seminorms and the space of interactions used in the construction of the fixed point}
\label{sec:potentials}

I now turn to the definition of the subspaces of \Grassmann\ used in the rest of this paper.  The first subspace, \kernel, is still defined as a dual space, but to obtain a finer topology I introduce it as the dual of a larger space of test functions defined as follows.

For $f \in \diffyn$, let
\begin{equation}
	\left| f \right|
	:=
	\int \dd^d y
	\max_{\bs a \in \intrange1{N}^n}
	\sup_{\substack{\bs x \in (\bR^d)^n: \\ x_n = y}}
	e^{- \eta(\cT (\bs x))}
	D_{\bs a}[f] (\bs x)
	\label{eq:bf_norm}
\end{equation}
with $\cT(\bs x)$ the tree (Steiner) diameter of $x_1,\dots,x_n$, see \cref{sec:Steiner}. 
The peculiar form of this norm compared to those of e.g.\ \cite{GMR,AGG_part2} is because this is only an auxiliary norm on test functions used to introduce a more important norm below by duality; its main features can be understood by noting that it can be used to bound both $\cV[f]$ for translation invariant ``local'' terms like $\cV = \int (\psi,\Omega \psi)(x) \dd^d x$ (cf.\ \cref{eq:mass_term_as_map}), and also integrals of the form
\begin{equation*}
	\int V_{\bs a, \ul m}(x_1 - x_n, x_2-x_n,\dots,x_{n-1}-x_n) \partial^{\ul m} f(\bs x)
	\dd^{n \, d}\bs x
\end{equation*}
with $V$ decaying somewhat faster than $e^{-\eta(\cT(\bs x))}$.
For example, the fact that the integral in \eqref{eq:bf_norm} diverges if $f$ is translation invariant and nonzero corresponds to the fact that, for such an $f$, $\cV[f]$ diverges for many $\cV$ of these forms.
Correspondingly, note that the supremum and integral appear in the opposite order in \cref{eq:plus_norm} compared to \cref{eq:bf_norm}, so that $| G |_+$ may be finite for a translation invariant $G$.
For $|f|$ to be finite also requires (loosely speaking) that the function not grow too quickly in the distance between the points, in contrast to $|\cdot|_+$ in \cref{eq:eta_growth} which imposes a complementary decay condition.

In the degenerate case $n=0$, $\slow_0 = \bC$ and $| \cdot |$ is the absolute value.
In the case $n=1$ (which will not be particularly important) \cref{eq:bf_norm} simplifies to 
\begin{equation}
	|f|
	=
	\int \dd^d x
	\max_{a \in \intrange1{N}}
	D_{a}[f] (x)
	.
\end{equation}

\begin{remark}
	Singling out one of the position arguments is very helpful for some manipulations later; the specific choice won't play a particular role because all of the important applications involve antisymmetric functions.
\end{remark}

With the continuity assumption in the definition of $\diffy_n$, this defines a normed vector space which I shall call \slown.  
Comparing this to the previous section, note that $\testn \subset \slown$
Although there is no direct relationship between this norm and the seminorms $\tn$ in general, for any sequence $f_m \to f$ in \testn, it is easy to see that $f_m \to f $ in \slown, since these are necessarily compactly supported; in other words the topology of \slown\ (restricted to \testn) is coarser than that of \testn.  This implies that any continuous map on \slown\ is also continuous on \testn.

Let $\slow_n^*$ be the dual space of \slown, that is the space of linear functions $H \colon \slown \to \bC$ bounded with respect to the norm
\begin{equation}
	\left\| H \right\|
	:=
	\sup_{f \in \slown} \frac{|H[ f]|}{| f|}
	.
	\label{eq:teststar_norm}
\end{equation}
Note that $\slown^*$ is automatically a Banach space, and that $\slown^* \subset \testn^*$ (i.e.\ boundedness on \slown\ implies continuity on \testn).
Convergence $H^m \to H$ in $\slown^*$ (i.e.\ in norm) implies $H^m [f] \to H[f]$ for all $f \in \testn$, i.e.\ convergence in the topology of $\testn^*$.

As a notational matter, we can identify an $H \in \slown^*$ with a collection of distributions by
\begin{equation}
	H[ f]
	=
	\sum_{\bs a \in \intrange1{N}^n}
	\int \dd^{n \, d} (\bs x)
	H_{\bs a} (\bs x)
	f_{\bs a}(\bs x)
	.
	\label{eq:H_a_def}
\end{equation}
This is useful for verifying the following estimate, which I will use repeatedly;
recall that $| \cdot |_+$ was defined in \cref{eq:plus_norm}, which differs from \cref{eq:bf_norm} in the sign of the exponent.

\begin{lemma}
	Given $G \in \diffy_2$ with $|G|_+ < \infty$,
	$f \in \slow_{\ell+m-2}$, and $H \in \slow^*_m$ for positive integers $\ell,m$, 
	\begin{equation}
		h_{a, \bs b} (x, \bs y)
		=
		\sum_{\bs a', \bs b'}
		\int \int
		H_{\bs b',a'}( \bs y',x')
		G_{a,a'}(x,x')
		f_{\bs b', \bs b}(\bs y', \bs y)
		\dd^d x' \dd^{(m-1)d} \bs y'
		\label{eq:contraction_test_function}
	\end{equation}
	defines an $h \in \slow_\ell$ with
	\begin{equation}
		\left|  h\right|
		\le 
		\| H \| \, |  f | \, | G|_+
		.
		\label{eq:contraction_bound_generic}
	\end{equation}
	\label{lem:contraction_bound_generic}
\end{lemma}
Note that the ordering of the arguments of $H$ and $f$ in \cref{eq:contraction_test_function} is deliberate.
\begin{proof}
	Writing out the definitions in \cref{eq:D_ell_def,eq:bf_norm},
	\begin{equation}
		\left|h \right|
		=
		\begin{aligned}[t]
			\int \dd^d z \,
			&
			\max_{a \in \intrange1N}
			\max_{\bs b \in \intrange1N^{\ell-1}}
			\sup_{x \in \bR^d}
			\sup_{\substack{\bs y \in (\bR^d)^{\ell-1} \\ y_{\ell-1} = z}}
			e^{- \eta(\cT (x, \bs y))}
			\\ & \times
			\max_{\ul k \in \mi_{1,R}}
			\max_{\ul m \in \mi_{\ell-1,R}}
			\left|
			\frac{\partial^{\ul k}}{\partial x^{\ul k}}
			\frac{\partial^{\ul m}}{\partial \bs y^{\ul m}}
			h_{ \bs b,a} (x, \bs y)
			\right|
			,
		\end{aligned}
	\end{equation}
	where I use the abbreviation $\cT(x,\bs y)$ for the Steiner diameter of the set obtained by combining $x$ and $\bs y$, and so on. 
	Taking the derivatives inside the integral defining $h$, this has the form
	\begin{equation}
		\begin{aligned}
			\left|h \right|
			&=
			\int \dd^d z \,
			\max_{a \in \intrange1N}
			\max_{\bs b \in \intrange1N^{\ell-1}}
			\sup_{x \in \bR^d}
			\sup_{\substack{\bs y \in (\bR^d)^{\ell-1} \\ y_{\ell-1} = z}}
			e^{- \eta(\cT (x, \bs y))}
			\max_{\ul k \in \mi_{1,R}}
			\max_{\ul m \in \mi_{\ell-1,R}}
			\left| H( \tilde h) \right|
			\\ & \le
			\| H \|
			\int \dd^d z \,
			\max_{a \in \intrange1N}
			\max_{\bs b \in \intrange1N^{\ell-1}}
			\sup_{x \in \bR^d}
			\sup_{\substack{\bs y \in (\bR^d)^{\ell-1} \\ y_{\ell-1} = z}}
			e^{- \eta(\cT (x, \bs y))}
			\max_{\ul k \in \mi_{1,R}}
			\max_{\ul m \in \mi_{\ell-1,R}}
			\left| \tilde h\right|
		\end{aligned}
	\end{equation}
	for $\tilde h \in \slow_m$,
	\begin{equation}
		\tilde h _{\bs b',a'} (\bs y',x')
		: =
		\frac{\partial^{\ul k}}{\partial x^{\ul k}}
		G_{a,a'}(x,x')
		\frac{\partial^{\ul m}}{\partial \bs y^{\ul m}}
		f_{\bs b', \bs b}(\bs y', \bs y)
		.
	\end{equation}
	Then writing out $|\tilde h|$ and noting that the last argument of $\tilde h$ (that is, the one which is integrated in the norm) is $y'_{m-1}$ which I rename $z$, 
	\begin{equation}
		\label{eq:contraction_h_norm}
		\begin{aligned}
			\left|h\right|
			&\le
			\begin{aligned}[t]
				\| H \|
				\int \dd^d z \,
				\max_{a \in \intrange1N}
				\max_{\bs b \in \intrange1N^{\ell-1}}
				\sup_{x \in \bR^d}
				\sup_{\substack{\bs y \in (\bR^d)^{\ell-1} \\ y_{\ell-1} = z}}
				\int \dd^d x' \,
				\max_{a' \in \intrange1N}
				\max_{\bs b' \in \intrange1N^{m-1}}
				\sup_{\substack{\bs y' \in (\bR^d)^{m-1}}}
				&
				\\ \qquad \times
				e^{- \eta(\cT (x, \bs y)) - \eta(\cT\left(x', \bs y' \right))}
				D_{a, a'}[G](x, x')
				D_{\bs b' , \bs b}[f](\bs y' , \bs y)
				&
			\end{aligned}
			\\ &
			\le
			\begin{aligned}[t]
				\| H \|
				\int \dd^d z \,
				\max_{\bs b \in \intrange1N^{\ell-1}}
				\sup_{x \in \bR^d}
				\sup_{\substack{\bs y \in (\bR^d)^{\ell-1} \\ y_{\ell-1} = z}}
				\int \dd^d x' \,
				\max_{a,a' \in \intrange1N}
				\max_{\bs b' \in \intrange1N^{m-1}}
				\sup_{\substack{\bs y' \in (\bR^d)^{m-1}}}
				&
				\\ \qquad \times
				e^{- \eta(\cT (x, \bs y)) - \eta(\cT\left(x', \bs y' \right))}
				D_{a, a'}[G](x, x')
				D_{\bs b' , \bs b}[f](\bs y' , \bs y)
				&
			\end{aligned}
		\end{aligned}
	\end{equation}

	Comparing the relevant trees, note that
	\begin{equation}
		\cT\left( x, \bs y  \right)
		+
		\cT\left( x', \bs y' \right)
		+
		\cT\left( x,x' \right)
		\ge
		\cT\left(  x , \bs y ,  x' , \bs y' \right)
		\ge
		\cT\left(   \bs y , \bs y' \right)
	\end{equation}
	and so using the subadditivity of $\eta$ this implies 
	$$e^{- \eta(\cT ( x, \bs y)) - \eta(\cT\left(  x' , \bs y' \right))} \le e^{\eta(\cT\left( x, x' \right)) - \eta(\cT\left( \bs y, \bs y ' \right))}; $$
	applying this to the right hand side of \cref{eq:contraction_h_norm} eliminates the factors depending on both $x,x'$ and $\bs y,\bs y'$, so that rearranging the result gives
	\begin{equation}
		\begin{aligned}
			\left|h\right|
			&\le
			\begin{aligned}[t]
				&
				\| H \|
				\sup_{\substack{ x \in \bR^d}}
				\int \dd^d x' \,
				\max_{a,a' \in \intrange1N}
				e^{\eta(\cT (x, x'))}
				D_{a, a'}[G](x, x')
				\\ &
				\times
				\int \dd^d z \,
				\max_{\bs b \in \intrange1N^{\ell-1}}
				\max_{\bs b' \in \intrange1N^{m-1}}
				\sup_{\substack{\bs y \in (\bR^d)^{\ell-1} }}
				\sup_{\substack{\bs y' \in (\bR^d)^{m-1}, \\ y'_{m-1}=z}}
				e^{- \eta(\cT\left(\bs y, \bs y' \right))}
				D_{\bs b , \bs b'}[f](\bs y , \bs y')
			\end{aligned}
			\\ &
			=
			\| H \| \, | G|_+\, |  f | 
			.
		\end{aligned}
	\end{equation}
\end{proof}

\medskip

\begin{definition}
	Let \kerneln\ be the set of $H \in \slown^*$ which are antisymmetric, in the sense that, letting $\pi f_{\bs a} (\bs x) := f_{\pi(\bs a)}(\pi(\bs x))$,
			$W(\pi f_{\bs a}) \equiv \sign(\pi) W( f_{\bs a})$.

			Let \kernel\ be the topological product of \kerneln over $n \in \bN$.
	\label{def:kernel}
\end{definition}
This means that $\cV \in \kernel$ is an infinite sequence $V_0,V_1,\dots$, and that convergence $\cV_{j} \to \cV$ is defined by $\lim_{j \to \infty} \|V^j_n - V_n \| = 0$ for each $n$, with $\| \cdot \|$ the norm of $\slown^*$ defined in \cref{eq:teststar_norm}, or $\lim_{j \to \infty} \| \cV^j - \cV\|_n =0\ \forall n \in \bN$ with the shorthand $\|\cV\|_n := \| V_n \|$.
This is a locally convex vector space, and it is complete in the sense of \cite[Chapter~V]{ReedSimon1}, as a result of the completeness of each of the spaces $\kernel_n$, which follows from the closedness of $\slow^*_n$.  

\begin{remark}
	In light of the relationship discussed above between $\testn^*$ and $\slown^*$, $\kernel \subset \Grassmann$, and convergence in $\kernel$ implies convergence in $\Grassmann$.  Since both spaces have the same notion of addition and scalar multiplication, this also means that any differentiable function $\bR \to \kernel$ is also a differentiable function $\bR \to \Grassmann$, with the same derivative.

	At the risk of belaboring the obvious, each element of \kernel\ is a (unique) element of \Grassmann.  This is in contrast to the space of trimmed sequences defined in \cite[Section~4]{GMR} or the family of spaces $\sV^{\alpha,\beta;\gamma}$ of \cite{Duch}, where many different elements correspond to the same element of \Grassmann.
	\label{rem:derivability_kernel}
\end{remark}

Antisymmetry leads to another important bound, using \cref{eq:good_Pfaff}.
For $V_\ell \in \kernel_\ell$ and $2n < \ell$, using the antisymmetry of $V_\ell$
\begin{equation}
	W_\ell [P^{\otimes n} \otimes f]
	=
	\frac{2^n n!}{(2n)!}
	W_\ell [\PPf_n [P] \otimes f]
	\label{eq:cont_to_Pf}
\end{equation}
where $\PPf$ is the function defined in \cref{eq:PPf_def},
and so 
\begin{equation}
	\left| 
	W_\ell [P^{\otimes n} \otimes f]
	\right|
	\le
	\frac{2^n n!}{(2n)!}
	C_\GH^n  \| W_\ell \| |f|
	.
	\label{eq:mult_Pfaff_norms}
\end{equation}

\medskip

Finally I introduce a subspace of \kernel\ with some further restrictions related to symmetry, where I will look for the actual solution in question.

\begin{definition}
	Let \pot\ be the set of $\cV \in \kernel$ which are 
	\begin{enumerate}
		\item Even order: $V_n = 0$ unless $n$ is a positive even integer 
		\item Translation invariant: defining $T_y \bs f$ by $f_{\bs a}(\bs x + y)$ with $(\bs x + y)_j = x_j + y$, $V_n[T_x  f] \equiv V_n[ f]$.
		\item Invariant under the other symmetries:
			$V_n[S[f]] = V_n[f]$ for all ${f \in \diffyn}$, ${S \in \symm}$.
	\end{enumerate}
	\label{def:pot}
\end{definition}
Note that $\pot$ is a (topological vector) subspace of $\kernel$, in particular it also complete.

\subsection{Reformulation of the flow equation}
\label{sec:flow_formal}
Although \kernel\ is a vector subspace of \Grassmann, it is not a subalgebra, in that the antisymmetric product of two elements of \kernel\ can very easily be unbounded in the norms introduced above, essentially because the locality property enforced by these norms is lost.
For example 
$\cV = \int  (\psi, \Omega \psi) (x)\dd^d x \in \kernel$ 
(see \cref{eq:mass_term_as_map}) but 
\begin{equation}
	\cV \wedge \cV
	: f \mapsto
	\begin{cases}
		\sum_{\bs a \in \intrange1N^4}
		\Omega_{a_1a_2}
		\Omega_{a_3a_4}
		\int f_{\bs a}(x,x,y,y)
		\dd^d x
		\dd^d y
		,& 
		f \in \test_4
		\\
		0
		, &
		\text{otherwise}
	\end{cases}
	\label{eq:bad_wedge}
\end{equation}
is an element of $\Grassmann$ but not of $\kernel$, since $(\cV \wedge \cV)[f]$ diverges for e.g.\ $f_{\bs a}(\bs x) = \exp (- |x_4|^2)$, which is an element of $\slow_4$ (but not $\test_4$).
Consequently the elementary operations used on the right hand side of \cref{eq:flow} are not necessarily well defined on \pot, so I will rewrite the full expressions in a more tractable fashion.

Let $\fA^{(\Lambda)} $ be the linear operator on \Grassmann\ defined by 
\begin{equation}
	\fA^{(\Lambda)} \cV (\psi)
	=
	(1 - \cP_0)
	\frac{1}{2} \left\langle \frac{\delta}{\delta \psi}, G_\Lambda \frac{\delta}{\delta\psi} \right\rangle \cV_\Lambda (\psi)
	,
	\label{eq:fA_in_Grassmann}
\end{equation}
where $\cP_0$ is the projection onto $\Grassmann_0 \simeq \bC$, which was left implicit in \cref{eq:flow} and elsewhere in \cref{sec:intro}.  This is equivalent to  
\begin{equation}
	( \fA^{(\Lambda)} \cV)_\ell [f] 
	=
	\frac{(\ell + 1)(\ell+2)}{2}
	V_{\ell+2} [G_\Lambda \otimes f ]
	\label{eq:Lflow_def}
\end{equation}
for $\ell>0$, and $ ( \fA^{(\Lambda)} V)_0 = 0$.  With this last stipulation, $\fA^{(\Lambda)}$ is defined on all of \kernel\ and maps \pot\ into itself.  Since
\begin{equation}
	\begin{aligned}
		\left| G_\Lambda \otimes f \right|
		&
		=
		\int \dd^d z
		\max_{\bs a \in \intrange1N^2} \sup_{\bs b \in \intrange1N^\ell}
		\max_{\bs x \in (\bR^d)^2}
		\sup_{\substack{\bs y \in (\bR^d)^{\ell}, \\ y_{\ell}=z}}
		e^{- \eta(\cT (\bs x, \bs y))}
		D_{\bs a}[G_\Lambda] (\bs x)
		D_{\bs b}[f] (\bs y)
		\\ & 
		\le \Lambda^{2 [\psi] - 1}C_{\GH} |f|
		,
	\end{aligned}
\end{equation}
it is easy to see that each $\| (\fA^{(\Lambda)} \cV)_\ell\|$ is finite, and so $\fA^{(\Lambda)}$ maps $\kernel \to \kernel$ and $\pot \to \pot$.

Similarly, the bilinear map $\Grassmann^2 \to \Grassmann$ defined by 
\begin{equation}
	\cQ^{(\Lambda)} (\cV,\cW)
	:=
	(1 - \cP_0)
	\left\langle \frac{\delta}{\delta \psi} \cV_\Lambda(\psi), G_\Lambda \frac{\delta}{\delta \psi} \cW_{\Lambda}(\psi) \right\rangle
	\label{eq:cQ_Grassman_def}
\end{equation}
which (taking advantage of antisymmetry) is equivalent to
\begin{equation}
	Q^{(\Lambda)}_\ell (\cV,\cW) [f ]
	=
	\sum_{\substack{j,k \in \bN + 1 \\ j+k=\ell+2}}
	j k 
	\left( V_j \owedge W_k \right)[ G_\Lambda \otimes f ]
	\label{eq:Qflow_def}
\end{equation}
where $\owedge$ is a partially antisymmetrized product, defined in distributional notation by
\begin{equation}
	\left[ V_j \owedge W_k \right] (z_1,\dots,z_{j+k})
	=
	\frac{1}{(j+k-2)!}
	\begin{aligned}[t]
		\sum_{\substack{\pi \in \Pi_{j+k}: \\ \pi(1) =1, \pi(j+1) = 2}}
		&
		\sigma(\pi)
		V_j (z_1,z_{\pi(2)},\dots,z_{\pi(j+1)})
		\\ & \times
		W_k (z_2,z_{\pi(j+2)},\dots,z_{\pi(j+k)})
		,
	\end{aligned}
\end{equation}
where $z_j = (a_j,x_j)$.

It follows from \cref{lem:contraction_bound_generic} that 
\begin{equation}
	\left\| Q^{(\Lambda)}_\ell (\cV,\cW)\right\|
	\le
	| G_\Lambda |_+
	\sum_{\substack{j,k \in \bN_+ \\ j+k=\ell+2}}
	j k
	\| V_j \|
	\ 
	\| W_k \|
	\label{eq:Q_Lam_norm}
\end{equation}
and in particular $\cQ^{(\Lambda)}(\kernel,\kernel) \subset \kernel$, since the sum is finite.  Examining the parity and symmetry properties, also $\cQ^{(\Lambda)}(\pot,\pot) \subset \pot$.

With this notation, we can rewrite \cref{eq:flow} as
\begin{equation}
	\frac{\dd}{ \dd \Lambda} \cV
	=
	\fA^{(\Lambda)} \cV
	-
	\frac12
	\cQ^{(\Lambda)} (\cV,\cV)
	=:
	\cC_\Lambda (\cV)
	,
	\label{eq:flow_kernel}
\end{equation}
which also makes sense as a differential equation on \pot, and this is how I will approach it.

Recall the definition of $\fD_\Lambda$ and the related maps in \cref{eq:dilation_maindef}, which are equivalent to $\wt D_\Lambda: \slown \to \slown$, $f \mapsto (\bs a , \bs x \mapsto f_{\bs a} ( \bs x / \Lambda ))$ and $\fD_\Lambda \cV = \left( D_\Lambda V_j \right)_{j \in \bN}$, where
\begin{equation}
	D_\Lambda V_j [f] 
	:=
	\Lambda ^{-j [\psi]} 
	V_j \left[ \wt D_\Lambda f \right]
	.
	\label{eq:DLambda_kernel_def}
\end{equation}
This definition is designed to make $D_\Lambda V_{2n} [P_\Lambda^{\otimes n}]$ independent of $\Lambda$.
Evidently $\fD_\Lambda $ is a jointly-continuous one-parameter semigroup on \kernel\ up to a reparameterization $\Lambda = e^\eta$, and also $\fD_\Lambda (\pot)=\pot$.

Then the problem of finding a scaling solution can be restated as
\begin{equation}
	\frac{\dd}{\dd \Lambda} [ \fD_\Lambda \cV^*]
	=
	\cC_\Lambda (\fD_\Lambda \cV^*)
	;
	\label{eq:kernel_FP}
\end{equation}
by the semigroup property of $\fD$, this holds either for all $\Lambda \in (0,\infty)$ or not at all.

I now list some other helpful formulae about the relationship of $\fD_\Lambda$ with the other objects defined above.
\cref{eq:propagator_scaling} can be rephrased as
\begin{equation}
	\wt D_\Lambda G_\Lambda 
	=
	\Lambda^{-1 + 2 [\psi]}G
	,
\end{equation}
and as a consequence 
\begin{equation}
	\begin{split}
		D_\Lambda Q_\ell (\cV, \cW) [f]
		& =
		\Lambda^{-\ell [\psi]}
		\sum_{\substack{j,k \in \bN_+ \\ j+k=\ell+2}}
		j k 
		\left( V_j \owedge W_k \right)[ G \otimes \wt D_\Lambda f ]
		\\ & =
		\Lambda^{1-(\ell+2) [\psi]}
		\sum_{\substack{j,k \in \bN_+ \\ j+k=\ell+2}}
		j k 
		\left( V_j \owedge W_k \right)[ \wt D_\Lambda G_\Lambda \otimes \wt D_\Lambda f ]
		\\ & =
		\Lambda
		\sum_{\substack{j,k \in \bN_+ \\ j+k=\ell+2}}
		j k 
		\left( D_\Lambda V_j \owedge D_\Lambda W_k \right)[  G_\Lambda \otimes  f ]
	\end{split}
\end{equation}
(as with the propagator, I use the abbreviation $\cQ = \cQ^{(1)}$),
implying
\begin{equation}
	\cQ^{(\Lambda)}(\fD_\Lambda \cV , \fD_\Lambda \cW)
	=
	\Lambda^{-1}
	\fD_\Lambda
	\cQ(\cV,\cW)
	.
	\label{eq:Q_rescaling}
\end{equation}
Similarly (abbreviating $\fA = \fA^{(1)}$)
\begin{equation}
	(\fA^{(\Lambda)} \fD_\Lambda \cV)_\ell [f]
	=
	\Lambda^{-1}
	D_\Lambda
	(\fA \cV)_\ell [f]
	\label{eq:fA_rescaling}
\end{equation}
and so
\begin{equation}
	\cC_\Lambda (\fD_\Lambda \cV)
	\equiv
	\Lambda^{-1}
	\fD_\Lambda \cC(\cV)
	\label{eq:F_rescaling}
\end{equation}
with $\cC := \cC_1$.

\subsection{Wick ordering}
\label{sec:Wick}
$\fA$ has an important relationship to Wick ordered polynomials, corresponding to \cref{eq:Wick_deriv_formal}: it acts on them as an infinitesimal change in the propagator which defines the Wick ordering.

Wick ordering is a linear operation which, using antisymmetry, acts on $\kernel_\ell$ as
\begin{equation}
	\wickord{V_\ell}_\Lambda[ f ]
	=
	(-1)^m w_{\ell m}
	\left( V_\ell, P_\Lambda^{\otimes m} \otimes f \right)
	, \quad
	f \in \slow_{\ell - 2m}
	\label{eq:kernel_Wick}
\end{equation}
where 
\begin{equation}
	w_{\ell m} 
	:=
	\binom{\ell }{ 2m}
	\frac{(2m)!}{2^m m!}
	=
	\frac{\ell!}{(\ell-2m)! 2^m m!}
	\label{eq:w_lm_def}
\end{equation}
is the number of ways of choosing $m$ unordered pairs from a set with $\ell$ elements, cf.\ \cite[Equation~(1.5.14)]{GJ}; recall that $P_\Lambda := \int_0^\Lambda G_\Sigma \dd \Sigma$ is the propagator with a UV cutoff.
It seems unlikely that this is well-defined on all of \kernel, however (and especially that \cref{eq:kernel_Wick} defines a continuous map), so I proceed cautiously as follows.
Let us first consider the definition on elements of \kernel\ which are sequences with only a finite number of nonzero entries, which I call \kernelfin; this is dense in \kernel, bearing in mind the topology specified above.  
Then
\begin{equation}
	\left( \fW_\Lambda \cV \right)_\ell [f]
	:=
	\sum_{m=0}^\infty
	(-1)^m
	w_{\ell + 2m, m}
	V_{\ell + 2m} \left[ P_\Lambda^{\otimes m} \otimes f \right]
	\label{eq:fW_def}
\end{equation}
defines a linear operator coinciding with Wick ordering on \kernelfin, preserving $\potfin := \kernelfin \cap \pot$, and it is invertible with inverse
\begin{equation}
	\left( \fW^{-1}_\Lambda \cV \right)_\ell [f]
	:=
	\sum_{m=0}^\infty
	w_{\ell + 2m, m}
	V_{\ell + 2m} \left[ P_\Lambda^{\otimes m} \otimes f \right]
	\label{eq:fW_inverse}
\end{equation}
as can easily be verified by direct calculation.
I will consider $\fW_\Lambda^{\pm 1}$ to be defined on those $\cV$ for which 
\begin{equation}
	\sum_{m=0}^\infty
	w_{\ell + 2m, m}
	\| \fW_\Lambda V_{\ell + 2m} \|
	<
	\infty
	,
	\ \forall \ell \in \bN + 1
	\label{eq:fW_convergence}
\end{equation}
on some open set of $\Lambda$,
which implies that sums the right hand sides of \cref{eq:fW_def,eq:fW_inverse} are absolutely convergent for all $f$.  On this domain $\fW_\Lambda^{\pm 1} \cV$ is jointly continuous in $\Lambda$ and $\cV$.

Since $\frac{\dd }{\dd \Lambda} P_\Lambda =  G_\Lambda$, using antisymmetry
\begin{equation}
	\begin{split}
		\frac{\dd}{\dd \Lambda}
		\left( \fW_\Lambda^{\pm 1} \cV \right)_\ell [f]
		& =
		\sum_{m=1}^\infty
		(\mp 1)^{m} 
		m w_{\ell + 2m, m}
		V_{\ell + 2m} \left[ P_\Lambda^{\otimes m-1} \otimes G_\Lambda \otimes f \right]
		\\ &=
		\mp
		\frac{(\ell + 2)(\ell + 1)}{2}
		\sum_{m=0}^\infty
		(\mp 1)^{m} 
		w_{\ell + 2m, m}
		V_{\ell + 2m+2} \left[ P_\Lambda^{\otimes m} \otimes G_\Lambda \otimes f \right]
		\\ & =
		\mp
		\frac{(\ell + 2)(\ell + 1)}{2}
		(\fW^{\pm 1}_\Lambda \cV)_{\ell+2} [G_\Lambda \otimes f]
	\end{split}
\end{equation}
(all of the sums are absolutely convergent for $\cV$ in the domain of $\fW^{\pm 1}_\Lambda$)
and also using
\begin{equation}
	m w_{\ell + 2m, m}
	=
	\frac{(\ell + 2m)!}{\ell! 2^m (m-1)!}
	=
	\frac{(\ell + 2) (\ell + 1 )}{2} w_{\ell + 2m, m-1}
\end{equation}
this combines with \cref{eq:Lflow_def} to give
\begin{equation}
	\frac{\dd }{\dd \Lambda} 
	\fW^{\pm 1}_\Lambda
	\cV
	=
	\mp
	\fA^{(\Lambda)} 
	\fW^{\pm 1}_\Lambda
	\cV
	\label{eq:L_flow_Wick}
	.
\end{equation}
Alternatively, 
\begin{equation}
	\begin{split}
		m w_{\ell + 2m, m}
		&=
		\frac{(\ell + 2m)!}{\ell! 2^m (m-1)!}
		=
		\frac{(\ell+2m)(\ell+2m-1)}{2}
		\frac{(\ell+2m-2)!}{\ell! 2^{m-1} (m-1)!}
		\\ & =
		\frac{(\ell+2m)(\ell+2m-1)}{2}
		w_{\ell+2m-2,m-1}
	\end{split}
\end{equation}
and so
\begin{equation}
	\begin{split}
		\frac{\dd}{\dd \Lambda}
		\left( \fW_\Lambda^{\pm 1} \cV \right)_\ell [f]
		& =
		\sum_{m=1}^\infty
		(\mp 1)^{m} 
		m w_{\ell + 2m, m}
		V_{\ell + 2m} \left[ P_\Lambda^{\otimes m-1} \otimes G_\Lambda \otimes f \right]
		\\ &=
		\mp
		\sum_{n=0}^\infty
		\begin{aligned}[t]
			\frac{(\ell + 2n+ 2)(\ell  +2n +1)}{2}
			(\mp 1)^{n} 
			w_{\ell + 2n, n}
			&
			\\  \times 
			\vphantom{\sum}
			V_{\ell + 2n+2} \left[G_\Lambda \otimes P_\Lambda^{\otimes n} \otimes  f \right]
			&
		\end{aligned}
		\\ & =
		\mp
		(\fW^{\pm 1}_\Lambda\cV)_\ell [f]
	\end{split}
\end{equation}
(renumbering $n=m-1$) and so also
\begin{equation}
	\frac{\dd }{\dd \Lambda} 
	\fW^{\pm 1}_\Lambda
	\cV
	=
	\mp
	\fW^{\pm 1}_\Lambda
	\fA^{(\Lambda)} 
	\cV
	\label{eq:L_Wick_flow}
\end{equation}
(in fact it is also easy enough to check directly that $\fW^{\pm 1}_\Lambda$ and $\fA^{(\Lambda)}$ commute). 
As a result, if $\frac{\dd}{\dd \Lambda} \cV_\Lambda = \fA^{(\Lambda)} \cV_\Lambda - \cQ^{(\Lambda)}(\cV_\Lambda,\cV_\Lambda) $ (i.e.\ \cref{eq:flow_kernel}) and $\cW_\Lambda \equiv \fW_\Lambda \cV_\Lambda$, then 
\begin{equation}
	\begin{split}
		\frac{\dd \cW_\Lambda}{\dd \Lambda}
		&=
		-
		\frac12
		\fW_\Lambda
		\cQ^{(\Lambda)}(\cV_\Lambda,\cV_\Lambda)
		\\ &=
		-
		\frac12
		\fW_\Lambda
		\cQ^{(\Lambda)}(\fW^{-1}_\Lambda \cW_\Lambda,\fW^{-1}_\Lambda \cW_\Lambda)
		=:
		\frac12
		\cQ^{\fW,\Lambda} (\cW_\Lambda,\cW_\Lambda)
		.
	\end{split}
	\label{eq:Wick_flow}
\end{equation}
The reverse implication also follows using \cref{eq:L_flow_Wick}.
Note that this equivalence depends on the interactions involved being in the domain of $\fW_\Lambda^{-1}$, as defined above in terms of absolute convergence in norms.
On this domain, however, the fixed point problem is equivalent to
\begin{equation}
	\left.
	\frac{\dd}{\dd \Lambda} \left[ \fD_\Lambda \cW^* \right]
	\right|_{\Lambda = 1}
	=
	\frac12
	\cQ^{\fW} (\cW^*,\cW^*)
	\label{eq:Wick_stationary_formal}
\end{equation}
with the convention $\cQ^{\fW} := \cQ^{\fW,1}$, cf.\ \cref{eq:flow_Wick_formal}.
Note that when $\cV,\cW \in \potfin$, 
\begin{equation}
	\begin{aligned}
		\fW_\Lambda \cQ^{(\Lambda)}(\cV,\cW)(\psi)
		&
		=
		\sum_{m=0}^\infty
		\frac{(-1/2)^m}{m!}
		\left\langle \frac{\delta}{\delta\psi} , P_\Lambda \frac{\delta}{\delta \psi} \right\rangle^m
		\left\langle \frac{\delta}{\delta \psi} \cV(\psi), G_\Lambda \frac{\delta}{\delta \psi} \cW(\psi) \right\rangle
		\\
		& =
		\sum_{n=0}^\infty
		\frac{(-1)^n}{n!}
		\left\langle \frac{\delta}{\delta\psi} , P_\Lambda \frac{\delta}{\delta \varphi} \right\rangle^n
		\left\langle \frac{\delta}{\delta \psi} \fW_\Lambda\cV(\psi), G_\Lambda \frac{\delta}{\delta \varphi}\fW_\Lambda \cW(\varphi) \right\rangle
		\Big|_{\varphi=\psi}
	\end{aligned}
	\label{eq:flow_Wick_finite}
\end{equation}
(cf.\ \cref{eq:flow_Wick_formal}; this can also be verified directly using \cref{eq:fW_def} but the resulting calculation is more involved) since all of the sums involved are finite, and so recalling the definition of $\cQ^{(\Lambda)}$ in \cref{eq:Qflow_def}
\begin{equation}
	Q^{\fW,\Lambda}_\ell (\cV,\cW)[f]
	=
	\begin{aligned}[t]
		\sum_{n=0}^\infty
		&
		\frac{(-1)^n}{n!}
		\sum_{\substack{j,k \in 2\bN + 2 \\ j+k=\ell+2n+2\\ j,k \ge n+1}}
		(-1)^{\tfrac12j-1}
		\frac{j!}{(j-n-1)!}
		\frac{k!}{(k-n-1)!}
		\frac{1}{\ell!}
		\sum_{\pi \in \Pi_\ell}
		\sigma(\pi)
		\\ & 
		\times
		\int\int
		\begin{aligned}[t]
			G_\Lambda (x_1,y_1)
			\prod_{r=2}^{n+1}
			P_\Lambda (x_r,y_r)
			[\pi f](x_{n+2},\dots,x_j,y_{n+2},\dots,y_k)
			& \\
			\times
			V_j(\bs x)
			W_k (\bs y)
			\dd^{jd} \bs x
			\dd^{k d} \bs y
			&
		\end{aligned}
	\end{aligned}
	\label{eq:QW_integral_prelim}
\end{equation}
where I have used antisymmetry to rearrange the arguments, giving an additional sign;
also, for legibility I have omitted the discrete indices, which are arranged in the same fashion as the position arguments.
Using the antisymmetry of $W_k$ further, the product of $P_\Lambda$ factors can be replaced by
\begin{equation}
	\frac{1}{n!}
	\det \left[ P_\Lambda(x_{i+1},y_{j+1}) \right]_{i,j \in \intrange1n}
	=
	(-1)^{\tfrac12n-1}
	\frac{1}{n!}
	\Pf \check \fP_{n}(x_2,..,x_{n+1},y_2,\dots,y_{n+1})
\end{equation}
where $\check \fP_n$ is the $2n \times 2n$ matrix with elements
\begin{equation}
	\check \fP_{ij} (\bs z)
	=
	\begin{cases}
		P_\Lambda (z_i,z_j)
		, &
		i \le n < j \text{ or } j \le n <i
		\\
		0
		, &
		\text{otherwise,}
	\end{cases}
	\label{eq:fP_tilde_pair}
\end{equation}
so
\begin{equation}
	Q^{\fW,\Lambda}_\ell (\cV,\cW)[f]
	=
	\begin{aligned}[t]
		\sum_{n=0}^\infty
		&
		(-1)^n
		\sum_{\substack{j,k \in 2\bN + 2 \\ j+k=\ell+2n+2\\ j,k \ge n+1}}
		(-1)^{(j+n)/2}
		j k
		\binom{j-1}n
		\binom{k-1}n
		\frac{1}{\ell!}
		\sum_{\pi \in \Pi_\ell}
		\sigma(\pi)
		\\ & 
		\times
		\int\int
		\begin{aligned}[t]
			&
			G_\Lambda (x_1,y_1)
			\Pf \check \fP_n(x_2,..,x_{n+1},y_2,\dots,y_{n+1})
			\\
			&
			\times
			[\pi f](x_{n+2},\dots,x_j,y_{n+2},\dots,y_k)
			V_j(\bs x)
			W_k (\bs y)
			\dd^{jd} \bs x
			\dd^{k d} \bs y
			.
		\end{aligned}
	\end{aligned}
	\label{eq:QW_integral}
\end{equation}
When using this formula later, I will find it helpful to note that 
$ j k
\binom{j-1}n
\binom{k-1}n$
is the number of ways of choosing which of the arguments of $V_j$ and $W_k$ appear as arguments of which other factors in the integral.

At $\Lambda=1$, using a Gram-Hadamard bound (see \cref{lem:gram_Lambda_integral,thm:gram} in \cref{sec:Pfaff}) to bound $| \Pf \check \fP_n| \le C_\GH^n$, and similarly for its derivatives, so $| \check \fP_n \otimes f|\le C_\GH^n |f|$ referring to the norm of $\slow_{\ell+2n}$; using this with \cref{lem:contraction_bound_generic}
\begin{equation}
	\begin{aligned}
		\left\| Q_\ell^\fW (\cV,\cW) \right\|
		&\le
		\sum_{\substack{j,k \in 2\bN + 2 \\ j+k=\ell+2n+2\\ j,k \ge n+1}}
		j k
		\binom{j-1}n
		\binom{k-1}n
		C_\GH^{n}
		| G |_+ 
		\| V_j \|
		\| W_k \|
		\\ & 
		\le
		C_\GH^{- ( \ell/2) - 1}
		|G|_+
		\sum_{\substack{j,k \in 2 \bN + 2}}
		3^{j} 3^{k}
		C_\GH^{j/2}
		C_\GH^{k/2}
		\| V_j \|
		\| W_k \|
	\end{aligned}
	\label{eq:QW_summable}
\end{equation}
and so $\cQ^\fW$ is well defined and continuous whenever the sum on the right hand side is finite.

\section{Relevant and irrelevant parts}
\label{sec:relevant_parts}

In this section, I introduce a linear operator $\fL$ on \pot\ with the properties
\begin{itemize}
	\item $\fL(\pot)$, i.e.\ the range of $\fL$ is finite dimensional, and
	\item letting $\fI = 1 - \fL$, there are constants $\Cl{irrel-pre}<\infty,\Cl{irrel-power}>0$ such that
		\begin{equation}
			\| \fD_\Lambda \fI \cV \|_\ell
			\le
			\Cr{irrel-pre} \Lambda^{-\Cr{irrel-power}}
			\| \cV\|_\ell
			\quad 
			\forall \Lambda \ge 1, \ \ell \in 2 \bN +2
			.
			\label{eq:irrel_basic}
		\end{equation}
\end{itemize}
Following the standard renormalization group terminology reviewed in \cref{sec:example}, $\fL \cV$ identifies the local relevant part of some $\cV \in \pot$, while \cref{eq:irrel_basic} is a precise way of stating that the remaining part $\fI \cV$ is irrelevant (which is how I will refer to it).

Actually it will be more convenient to decompose $\fI = \fI_\many + \fI_\inter$, where $\fI_\many$ is simply the projection onto sequences whose first few terms vanish (see \cref{eq:fI_many_def}), both since the proof of \cref{eq:irrel_basic} is quite different in the two corresponding cases (see \cref{eq:V_j_scaling,lem:fI_few_scaling}) and because I will treat the resulting parts differently in the construction of the fixed point.

\medskip

For $\ell \ge 2$, let $\mi_\ell$ be the set of multiindices $\intrange1\ell \times \intrange1d \to \bN$, and let $\mi_\ell^-$ be the set of $\ul m \in \mi_\ell$ with $m_{\ell,j}=0$ for $j=1,\dots,d$.  For each $f \in \slow_\ell$, I denote
\begin{equation}
	\cL f 
	=
	\sum_{\substack{\ul m \in \mi_\ell^- : \\ \ell [\psi] + |\ul m| \le d}}
	f_{\ul m}^\cL
	\label{eq:fcL_def}
\end{equation}
with
\begin{equation}
	f_{\ul m}^\cL
	\colon
	\bs a, \bs x
	\mapsto
	\frac{1}{\ul m !}
	(\bs x - \bs x_\cL)^{\ul m}
	\partial^{\ul m}
	f_{\bs a } (\bs x_\cL)
	,
	\label{eq:f_cL_m_def}
\end{equation}
recall $\bs x_\cL := (x_\ell,\dots,x_\ell)$.  Since we assumed at the beginning of \cref{sec:formalism} that $2 [\psi] + R - d > 0$, \cref{eq:fcL_def} involves only $\ul m \in \mi_{\ell,R}$. 

Since (recalling the definition of $D_{\bs a}$ in \cref{eq:D_ell_def})
\begin{equation}
	D_{\bs a}[f_{\ul m}^{\cL}](\bs x)
	=
	\max_{\ul n \le \ul m}
	\left|f_{\ul n}^\cL (\bs a, \bs x)\right|
	\le
	\max_{\ul n \le \ul m}
	\frac{1}{\ul n!}
	\left|(\bs x - \bs x_\cL)^{\ul n}\right| 
	D_{\bs a}[f](\bs x_\cL)
\end{equation}
and
$ \left| (\bs x - \bs x_\cL)^{ \ul n}\right| \le (\cT (\bs x))^{|\ul n|}$ for all $\bs x$ (note that this bound is saturated when $x_1 = x_2 = \dots x_{\ell-1}$), whence using \cref{eq:eta_growth} there is a  $\Cl{f_cL_pre}$ such that
\begin{equation}
	\sup_{\bs x}
	\left| (\bs x - \bs x_\cL)^{ \ul n}\right| 
	e^{- \eta(\cT (\bs x))}
	\le
	\sup_{y \in [0,\infty)} [\eta^{-1}(y)]^{|\ul n|} e^{-y}
	\le
	\Clast
	\label{eq:monomial_tree_bound}
\end{equation}
uniformly for $|\ul n| \le R$,
we have
\begin{equation}
	\begin{split}
		\left| f_{\ul m}^{\cL}\right|
		&
		\le
		\int 
		\left[ 
			\max_{\ul n \le \ul m}
			\frac{1}{\ul n!}
			\sup_{\substack{\bs x \in (\bR^d)^\ell:}}
			e^{-\eta(\cT(\bs x))}
			\left|(\bs x - \bs x_\cL)^{\ul n}\right| 
		\right]
		\max_{\bs a \in \intrange1N^\ell}
		D_{\bs a} [f] (y,\dots,y) \dd^d y
		\\ &
		\le
		\Clast
		\ |f|
		,
	\end{split}
\end{equation}
and so \Cp{f_cL}
\begin{equation}
	| \cL f | \le \Clast |f|
	,
	\label{eq:f_cL_norm}
\end{equation}
and in particular $\cL f \in \slow_\ell$.  Since the sum in \cref{eq:fcL_def} is empty for $\ell > d/[\psi]$, \Clast\ can be taken to be uniform in $\ell$.

Then let $\fL$ be the linear map on \kernel\ with
\begin{equation}
	\left( \fL \cV \right)_\ell [f]
	:=
	\frac{1}{\ell!}
	\sum_{\pi \in \Pi_\ell}
	\sign \pi
	V_\ell [\pi \cL f]
	;
	\label{eq:fR_V_def}
\end{equation}
since I have explicitly antisymmetrized and \cref{eq:f_cL_norm} implies
\begin{equation}
	\| \fL \cV \|_\ell
	\le 
	\Clast
	\| \cV \|_\ell
	\label{eq:fR_cV_norm}
\end{equation}
(by the definition of \kernel, $\pi V_\ell \equiv V_\ell$ and so $\|\pi V_\ell\| \equiv \|V_\ell\|$)
this does in fact give a map $\kernel \to \kernel$.
For $V \in \pot$, by translation invariance
\begin{equation}
	\begin{aligned}
		V_\ell [ f_{\ul m}^{\cL}]
		& =
		\int \dd^{\ell d} \bs x
		V_\ell (\bs x) (\bs x - \bs x_\cL)^{\ul m}
		\partial^{\ul m}
		f(\bs x_\cL)
		=
		\int \dd^{\ell d} \bs x
		V_\ell (\bs x-\bs x_\cL) (\bs x - \bs x_\cL)^{\ul m}
		\partial^{\ul m}
		f(\bs x_\cL)
		\\ & =
		\int \dd^d y_1 \cdots \int \dd^d y_{\ell-1}
		V_\ell (y_1,\dots,y_{\ell-1},0) (\ul y)^{\ul m}
		\int \dd^d x_\ell \partial^{\ul m} f(x_\ell,\dots,x_\ell)
		,
	\end{aligned}
\end{equation}
that is, $\fL \cV$ can always be expressed as a  linear combination of the finite set $\sY \subset \pot$ whose elements have the form
\begin{equation}
	\cY^\ell_{\ul m, \bs a}
	:=
	\int
	\frac{\partial^{m_1}}{\partial x^{m_1}}	\psi_{a_1} (x) 
	\dots 
	\frac{\partial^{m_\ell}}{\partial x^{m_\ell}}	\psi_{a_\ell} (x) 
	\dd^d x
	\colon
	f \mapsto
	\begin{cases}
		\int \partial^{\ul m} f_{\bs a}(x,\dots,x) \dd^d x
		, &
		f \in \slow_\ell
		\\
		0
		, &
		\text{otherwise}
	\end{cases}
	\label{eq:Y_def}
\end{equation}
with $\ell [\psi] + |\ul m| \le d$,
and so $\fL(\pot)$ is finite dimensional.

Recalling the expression in \cref{eq:Delta_explicit} for $\Delta$ (the generator of a one-parameter semigroup of dilations),
evidently $\sY$ is a collection of eigenvectors of $\Delta$:
\begin{equation}
	\fD_\Lambda 
	\cY^\ell_{\ul m, \bs a}
	=
	\Lambda^{d - \ell [\psi] - |\ul m|}
	\cY^\ell_{\ul m, \bs a}
	\implies
	\Delta
	\cY^\ell_{\ul m, \bs a}
	=
	(d - \ell [\psi] - |\ul m|)
	\cY^\ell_{\ul m, \bs a}
	\label{eq:Y_fL_eig}
\end{equation}

It is clear from \cref{eq:fcL_def,eq:f_cL_m_def} that $\wt D_\Lambda \cL f \equiv \cL \wt D_\Lambda f$, so $\fL$ indeed commutes with $\fD_\Lambda$.

Taking into account the other symmetries included in the definition of $\pot$ in \cref{def:pot}, $\fL(\pot)$ is generically smaller.  For the model under consideration (and especially for the specified parameter range $\varepsilon \in [0,d/6)$, where $\sY$ includes only elements with $\ell=2$, $|\ul m| \le 1$ and $\ell = 4$, $\ul m = 0$), the requirement of invariance under symplectic and rotational symmetries means that $\fL(\pot)$ is actually spanned by the two elements
\begin{equation}
	\sY^2_\Omega
	:=
	\sum_{a,b \in \intrange1N}
	\Omega_{ab} \sY^2_{(a,b),0}
	, \quad
	\sY^4_\Omega
	:=
	\sum_{\bs a \in \intrange1N^4}
	\Omega_{a_1,a_2}
	\Omega_{a_3,a_4}
	\sY^4_{\bs a,0}
	.
	\label{eq:fLpot_smallbasis}
\end{equation}

It is also evident that $\cL \cL f = \cL f$, and so $\fL^2 = \fL$ on \kernel.  Letting 
\begin{equation}
	(\fI_\inter \cV)_\ell
	:=
	\begin{cases}
		V_\ell - (\fL \cV)_\ell
		, & 
		\ell [\psi] \le d
		\\
		0
		, &
		\ell [\psi] > d
	\end{cases}
	,
	\label{eq:fI_few_def}
\end{equation}
evidently $\| (\fI_\inter \cV)_\ell \| \le (1+\Clast) \|V_\ell\|$ and $\fI_\inter^2 = \fI_\inter$.  Also $\fI_\inter$ commutes with $\fD_\Lambda$.
Note that the first case of \cref{eq:fI_few_def} can equivalently be expressed as 
\begin{equation}
	(\fI_\inter \cV)_\ell [f]
	=
	\sum_{\pi \in \Pi_\ell}
	\sign \pi
	V_\ell [\pi \cR f]
	\label{eq:fI_few_alt}
\end{equation}
where $\cR f := f - \cL f$ is the Taylor remainder of order $N_\ell := \floor{d - \ell [\psi]}$ of $f(x)$ around $x_\cL$.
Finally letting $\fI_\many := 1 - \fI_\inter - \fL$ or equivalently 
\begin{equation}
	(\fI_\many \cV)_\ell
	=
	\begin{cases}
		0
		, &
		\ell [\psi] \le d
		\\
		V_\ell
		, &
		\ell [\psi] > d
	\end{cases}
	\label{eq:fI_many_def}
\end{equation}
we also have $\fI_\many^2 = \fI_\many$, and $\fI_\many$ also commutes with $\fD_\Lambda$.

Recalling \cref{eq:bf_norm},
\begin{equation}
	\left| \wt D_\Lambda f \right|
	=
	\int \dd^d y
	\max_{\bs a \in \intrange1{N}^n}
	\sup_{\substack{\bs x \in (\bR^d)^\ell: \\ x_\ell = y}}
	e^{- \eta(\cT (\bs x))}
	D_{\bs a}[\wt D_\Lambda f] (\bs x)
\end{equation}
and for $\Lambda \ge 1$, also $D_{\bs a}[\wt D_\Lambda f ](\bs x) \le D_{\bs a} [f] (\bs x / \Lambda)$ and thus
\begin{equation}
	\begin{split}
		\left| \wt D_\Lambda f \right|
		& \le
		\int \dd^d y
		\max_{\bs a \in \intrange1{N}^n}
		\sup_{\substack{\bs x \in (\bR^d)^\ell: \\ x_\ell = y}}
		e^{- \eta(\cT (\bs x))}
		D_{\bs a}[ f] (\bs x / \Lambda)
		\\ &=
		\Lambda^{d}
		\int \dd^d y
		\max_{\bs a \in \intrange1{N}^n}
		\sup_{\substack{\bs x \in (\bR^d)^\ell: \\ x_\ell = y}}
		e^{- \eta(\Lambda \cT (\bs x))}
		D_{\bs a}[ f] (\bs x)
		\le
		\Lambda^{d} |f|
		,
	\end{split}
	\label{eq:f_scaling}
\end{equation}
and consequently 
\begin{equation}
	\left\| \fD_\Lambda \cV \right\|_j
	\le
	\Lambda^{d- j[\psi]}
	\| \cV \|_j
	,
	\quad \forall \Lambda \ge 1,
	\label{eq:V_j_scaling}
\end{equation}
giving scaling behavior similar to that of \cref{eq:sextic_scaling} as an upper bound also for nonlocal terms.  In particular $\fI_\many$ is defined in \cref{eq:fI_many_def} exactly so that only negative exponents appear in \cref{eq:V_j_scaling} for the nonvanishing part, corresponding to the idea that such terms are dimensionally irrelevant.

\medskip

It is less obvious that the operation $\fI_\inter$, defined above as the remainder of the $\ell [\psi] \le d$ part after subtracting the ``local relevant'' part defined by $\fL$, is also irrelevant in a similar sense.  As discussed in the introduction, this has usually been handled elsewhere by redefining $\fI_\inter$ as a map in a larger space involving derivative labels (sometimes introduced via a jet extension, as in \cite{Duch}) so that different components have different scaling properties, but this leads to a number of inconveniences, including a spurious divergence which sometimes unnecessarily restricts the choice of cutoff.  
Thinking of the interaction as a functional or distribution as I do here -- rather than a collection of coefficient functions -- makes it possible instead to reformulate the necessary estimate on the basis of the following lemma.

\begin{lemma}
	For any $d, [\psi] > 0$ there exists $\Cl{fI_few_scaling}$ such that
	\begin{equation}
		\left| \cR \wt D_\Lambda f \right|
		\le
		\Clast \Lambda^{d - N_\ell -1}
		|f|
		\label{eq:f_cR_scaling}
	\end{equation}
	for all $f \in \slow_\ell$ and $\Lambda \in [1,\infty)$.
	\label{lem:f_cR_scaling}
\end{lemma}
\begin{proof}
	First, note that if $|\ul n| > N_\ell$, $\partial^{\ul n} \cR f_{\ul a} = \partial^{\ul n} f_{\ul a}$ and so
	\begin{equation}
		\partial^{\ul n} \cR \wt D_\Lambda f_{\ul a} (\bs x)
		= 
		\partial^{\ul n} \wt D_\Lambda f_{\ul a} (\bs x)
		=
		\Lambda^{-|\ul n|} \partial^{\ul n} f_{\ul a} (\Lambda^{-1} \bs x)
		.
		\label{eq:f_cR_easy_scaling}
	\end{equation}
	On the other hand if $\left|\ul n\right| \le N_\ell$, by the Taylor remainder formula
	\begin{equation}
		\partial^{\ul n} \cR f_{\ul a }(\bs x)
		=
		\sum_{\substack{\ul m \in \mi_\ell^- \\ |\ul m| + |\ul n | = N_\ell +1}}
		\frac{|\ul m|}{\ul m!}
		(\bs x - \bs x_\cL)^{\ul m}
		\int_0^1
		(1-t)^{|\ul m| - 1}
		\partial^{\ul m + \ul n}
		f_{\bs a} (\bs x_\cL + t[\bs x - \bs x_\cL])
		\dd t
	\end{equation}
	and so
	\begin{equation}
		\begin{split}
			\partial^{\ul n} \cR \wt D_\Lambda f_{\ul a }(\bs x)
			= &
			\Lambda^{-N_\ell-1}
			\sum_{\substack{\ul m \in \mi_\ell^- \\ |\ul m| + |\ul n | = N_\ell +1}}
			\frac{|\ul m|}{\ul m!}
			(\bs x - \bs x_\cL)^{\ul m}
			\\ & \times 
			\int_0^1
			(1-t)^{|\ul m| - 1}
			\partial^{\ul m + \ul n}
			f_{\bs a} (\Lambda^{-1} \bs x_\cL + t \Lambda^{-1}[\bs x - \bs x_\cL])
			\dd t
		\end{split}
	\end{equation}
	and thus
	\begin{equation}
		\begin{split}
			\left| 
			\partial^{\ul n} \cR \wt D_\Lambda f_{\ul a }(\bs x)
			\right|
			& \le
			\Lambda^{-N_\ell-1}
			\sum_{\substack{\ul m \in \mi_\ell^- \\ |\ul m| + |\ul n | = N_\ell +1}}
			\frac{1}{\ul m!}
			\left|
			(\bs x - \bs x_\cL)^{\ul m}
			\right|
			\sup_{\substack{\bs z \in (\bR^d)^\ell : \\ \Lambda \bs z_\cL =  \bs x_\cL, \\ \Lambda \cT(\bs z) \le  \cT(\bs x)}}
			\left|
			\partial^{\ul m + \ul n}
			f_{\bs a} (\bs z)
			\right|
			\\ & \vphantom{x}
			\\ & \le
			\Lambda^{-N_\ell-1}
			e^{(\ell-1)d}
			[\cT (\bs x)]^{N_\ell + 1 - |\ul n|}
			\sup_{\substack{\bs z \in (\bR^d)^\ell : \\ \Lambda \bs z_\cL =  \bs x_\cL, \\ \Lambda \cT(\bs z) \le  \cT(\bs x)}}
			D_{\bs a}[f](\bs z)
			.
		\end{split}
		\label{eq:cR_f_scaling_hard}
	\end{equation}
	For all $\Lambda \in [2,\infty)$ using the monotonicity and subadditivity of $\eta$
	\begin{equation}
		\eta(\cT(x))
		\ge
		\eta\left( \frac12 \cT(x) + \Lambda^{-1} \cT(x) \right)
		\ge
		\eta\left( \frac12 \cT(x)\right) + \eta\left(\Lambda^{-1} \cT(x) \right)
		.
	\end{equation}
	Consequently using \cref{eq:monomial_tree_bound} there is some $\C \ge 1$ (depending on $\ell,\ [\psi]$ and $d$) such that for all $\Lambda \in [2,\infty)$
	\begin{equation}
		e^{-\eta(\cT(\bs x))}
		\left| 
		\partial^{\ul n} \cR \wt D_\Lambda f_{\ul a }(\bs x)
		\right|
		\le
		\Clast
		\Lambda^{-N_\ell-1}
		\sup_{\substack{\bs z \in (\bR^d)^\ell : \\ \Lambda \bs z_\cL =  \bs x_\cL, \\ \Lambda \cT(\bs z) \le  \cT(\bs x)}}
		e^{- \eta(\cT (\bs z))}
		D_{\bs a}[f](\bs z)
	\end{equation}
	which also holds for larger $\ul n$ from \cref{eq:f_cR_easy_scaling}; then inserting this into \cref{eq:bf_norm} gives
	\begin{equation}
		\begin{aligned}
			\left| \cR \wt D_\Lambda f \right|
			&
			\le
			\Clast 
			\Lambda^{-N_\ell-1}
			\int \dd^d y
			\max_{\bs a \in \intrange1{N}^\ell}
			\sup_{\substack{\bs x \in (\bR^d)^\ell: \\ x_\ell = y}}
			\sup_{\substack{\bs z \in (\bR^d)^\ell : \\ \Lambda z_\ell =  x_\ell, \\ \Lambda \cT(\bs z) \le  \cT(\bs x)}}
			e^{- \eta(\cT (\bs z))}
			D_{\bs a}[f](\bs z)
			\\ & 
			=
			\Clast 
			\Lambda^{-N_\ell-1}
			\int \dd^d y
			\max_{\bs a \in \intrange1{N}^\ell}
			\sup_{\substack{\bs z \in (\bR^d)^\ell : \\ \Lambda z_\ell =  y}}
			e^{- \eta(\cT (\bs z))}
			D_{\bs a}[f](\bs z)
			\\ & 
			=
			\Clast
			\Lambda^{d-N_\ell-1}
			|f|
			,
		\end{aligned}
	\end{equation}
	again for all $\Lambda \in [2,\infty)$.
	Also \cref{eq:fR_cV_norm,eq:f_scaling} give
	\begin{equation}
		\left|\cR \wt D_\Lambda f\right|
		\le 
		(1+\Cr{f_cL})
		\left| \wt D_\Lambda f\right|
		\le
		(1+\Cr{f_cL} )
		\Lambda^d
		|f|
		,
		\label{eq:trivial_interpolation_bound}
	\end{equation}
	which trivially gives a similar bound for $\Lambda \in [1,2]$.
\end{proof}

Note that it is evident from \cref{eq:fcL_def,eq:f_cL_m_def} that ${\cL \wt D_\Lambda f \equiv \wt D_\Lambda \cL f}$, and so likewise for $\cR f$.
\begin{corollary}
	For the same $\Cr{fI_few_scaling}$, 
	\begin{equation}
		\| \fD_\Lambda \fI_\inter \cV \|_\ell 
		=
		\| \fI_\inter \fD_\Lambda  \cV \|_\ell 
		\le 
		\Cr{fI_few_scaling}
		\Lambda^{d - \ell [\psi] - N_\ell -1 }
		\| V_\ell \|
		\label{eq:fI_few_scaling}
	\end{equation}
	for all $\Lambda \in [1,\infty)$.
	\label{lem:fI_few_scaling}
\end{corollary}
Recalling $N_\ell := \floor{d - \ell [\psi]}$, the power of $\Lambda$ appearing here is always negative.

\begin{corollary}
	There is a $\Cl{fI_few_integrate} < \infty$ such that
	\begin{equation}
		\left\| \int_1^\infty \fD_\Lambda \fI_\inter \cV \frac{\dd\Lambda}{\Lambda}\right\|_\ell
		\le
		\Clast \|V_\ell\|
		.
		\label{eq:fI_few_integral_bound}
	\end{equation}
	$\Clast$ can be taken bounded when 
	\begin{equation}
		\Dmin^\inter :=
		\min_{\ell=2,4,\dots,M}
		(\floor{d - \ell [\psi]} + 1 + \ell [\psi] - d)
		\label{eq:Dmin_inter_def}
	\end{equation}
	is bounded away from zero.
	\label{lem:fI_few_integral_bound}
\end{corollary}

\section{The ansatz for the irrelevant part of the fixed point}
\label{sec:fp_construction}

In this section, I construct $\pot$-valued functions $\cI_\inter, \cI_\many$ which give the irrelevant part of a fixed point action when the associated relevant part is chosen appropriately.
Formulating a suitable definition, and checking that it has a number of important properties (which involve introducing a norm on \pot, or rather on a subspace of \pot) will be the main technical element remaining for the proof of \cref{thm:main_FP}, which I conclude in the next section.

Abbreviating
\begin{equation}
	\cF_\Lambda(\cV) 
	:=
	\frac12
	\cQ^{\fW,\Lambda} (\fD_\Lambda \cV, \fD_\Lambda \cV)
	\label{eq:cF_def}
\end{equation}
and $\cF = \cF_1$, so that $\cF_\Lambda = \Lambda^{-1} \fD_\Lambda \circ \cF $ via \cref{eq:Q_rescaling}.
\begin{theorem}
	There exist linear maps 
	$\cI_\inter \colon \sB \to \fI_\inter(\pot) $,
	$\cI_\many \colon \sB \to \fI_\many(\pot) $,
	defined on some
	$\sB \subset \pot$,
	satisfying
	\begin{align}
		\frac{\dd}{\dd \Lambda}  \fD_\Lambda \cI_\inter(\cW)
		& =
		\fI_\inter \cF_\Lambda (\cW + \cI_\inter(\cW) + \cI_\many(\cW))
		,
		\label{eq:cI_inter_stationary}
		\\[0.5ex]
		\frac{\dd}{\dd \Lambda}  \fD_\Lambda \cI_\many(\cW)
		& =
		\fI_\many \cF_\Lambda (\cW + \cI_\inter(\cW) + \cI_\many(\cW))
		\label{eq:cI_many_stationary}
		,
	\end{align}
	and a function
	$\rho \colon \pot \to [0,\infty]$, such that (letting $\cI = \cI_\inter + \cI_\many$)
	\begin{enumerate}
		\item \label{it:rho_prenorm}
			The set of $\cV \in \pot$ with $\rho[\cV] < \infty$ is a Banach space with norm $\rho$; in particular $\rho[\cV] = 0$ iff $\cV=0$,
		\item \label{it:rho_local_finite}
			$\rho [\fL \cV]$ is always finite,
		\item \label{it:B_from_rho}
			There is some $B > 0$ such that
			$
			\sB \supset \left\{ \cV \in \pot : \, \rho[\cV] < B \right\}
			$
			and $\cI(\cW)$ is in the domain of $\fW^{-1}$ for all $\cW$ such that $\rho[\cW] < B$.
		\item \label{it:remainder_bounds}
			For some constants $\Cl{F_cubic},\Cl{rho_Lip}$
		\begin{align}
			\rho[ \fL \cQ^\fW (\cW ,  \cI(\cW)  )]
			, \ 
			\rho[ \fL \cQ^\fW (\cI(\cW)   ,  \cI(\cW) )]
			&
			\le
			\Cr{F_cubic}
			\rho[\cW]^3
			\label{eq:F_cubic}
			,
			\\ 
			\rho[\fL \cQ^\fW (\cI(\cV), \cI(\cV)) -\fL \cQ^\fW (\cI(\cW), \cI(\cW))]
			&
			\le
			\Cr{rho_Lip} (\rho[\cV]+\rho[\cW])^2 \rho[\cV - \cW]
			\label{eq:rho_Lip}
		\end{align}
		for all $\cV,\cW \in \pot$.
	\end{enumerate}
	\label{thm:FP_irrel}
\end{theorem}
The proof is fairly complex, and so I first present some preliminaries in \cref{sec:tree_def,sec:regroup,sec:tree_bounds} (in particular, $\rho$ is defined in \cref{eq:convergence_requisite}) before stating the definition of $\cI$ and completing the proof of \cref{thm:FP_irrel} in \cref{sec:concluding_ansatz}.

\subsection{The tree-expansion function $\cS$}
\label{sec:tree_def}

For a finite ordered set $V$, let $\fT_V$ be the set of trees on $V$; I denote $V(\tau) = V$ and $n(\tau) = \# V$ for $\tau \in \fT_V$, where $\#$ denotes the cardinality.
I call the edge set $\edge(\tau)$ and will abbreviate this as $\edge$ where it does not cause confusion.

I define a quantity by iterating on trees as follows.
When $\# V = 1$ (i.e.\ for a tree consisting of only a single vertex), for $\Lambda \in (0,\infty)$ and $\cV \in \pot_{\le M}:= \bigcup_{m=2}^{M} \pot_m$ let
\begin{equation}
	\cS[\tau,\Lambda,\cV]
	:=
	\fD_\Lambda 
	\cV
	\label{eq:S_tau_Lambda_trivial_case}
\end{equation}
and for other trees, for $\ul \Lambda = (\Lambda_{e_1},\dots) \in (0,\infty)^{\edge(\tau)}$ with no repeated elements, iteratively let 
\begin{equation}
	\cS[\tau, \ul \Lambda,\cV]
	:=
	\frac{n(\tau_1)! n(\tau_2)!}{ n(\tau)!}
	\fI_\many
	\cQ^{\fW,\Lambda_0}(\cS[\tau_1, \ul \Lambda_1,\cV], \cS[\tau_2,\ul \Lambda_2,\cV])
	\label{eq:S_tau_Lambda_iter_def}
\end{equation}
where $\Lambda_0$ is the minimum of $\Lambda_e$ over $e \in \edge(\tau)$, $\tau_1,\tau_2$ are the two subtrees of $\tau$ obtained by removing the corresponding edge, and $\ul \Lambda_1,\ul \Lambda_2$ are either the restriction of $\ul \Lambda$ if the associated subtree has more than one vertex or $\Lambda_0$ if it has only one.  Note that for any $\tau$ this gives an element of \potfin\ so all of the Wick orderings are well defined.

\medskip

Now for $\Lambda \in (0,\infty) $ let
\begin{equation}
	\cS^{(\Lambda)}[\tau,\cV]
	:=
	\int_{\ul \Lambda \in (\Lambda,\infty)^{\edge(\tau)}}
	\cS[\tau, \ul \Lambda,\cV]
	\dd \ul \Lambda
	;
	\label{eq:C_tau_def}
\end{equation}
note that $\fD_\lambda \cS^{(\Lambda)}[\tau,\cV] = \cS^{(\lambda\Lambda)}[\tau, \cV]$.
If all of the sums and integrals in \cref{eq:C_ell_tau_explicit,eq:C_tau_def} are absolutely convergent, this defines $\cS^{(\Lambda)}[\tau,\cV] \in \pot$ and for $n(\tau) \ge 2$
\begin{equation}
	\begin{split}
		\frac{\dd}{\dd \Lambda}
		\cS^{(\Lambda)}[\tau,\cV]
		&=
		\sum_{e \in \edge}
		\int_{\substack{\ul \Lambda \in [\Lambda,\infty)^\edge \\ \Lambda_e = \Lambda}}
		\cS [\tau, \ul \Lambda, \cV]
		\dd \ul \Lambda
		\\ &
		=
		\frac12
		\sum_{\tau_1}^*
		\binom{ n(\tau) }{ n(\tau_1) }^{-1}
		\fI_\many
		\cQ^{\fW,\Lambda} (
			\cS^{(\Lambda)}[\tau_1,\cV]
			,
			\cS^{(\Lambda)}[\tau_2,\cV]
		)
	\end{split}
	\label{eq:Cc_derivative}
\end{equation}
using \cref{eq:S_tau_Lambda_iter_def},
where the sum is over $\tau_1$ which are subtrees of $\tau$ such that there is a unique pair $v_1\in V(\tau_1), \ v_2 \in V(\tau) \setminus V(\tau_1)$ such that $\{v_1,v_2\}$ is an edge of $\tau$, and $\tau_2$ is the subtree made up of the remaining part of $\tau$; each suitable pair of subtrees appears twice in the sum, which is compensated by the factor of $1/2$.
Noting that
\begin{equation}
	\sum_{\tau \in \fT_V}
	\sum_{\tau_1}^*
	f(\tau,\tau_1)
	=
	\sum_{\substack{V_1 \subset V : \\ V_1, \ V \setminus V_1 \neq \emptyset}}
	\sum_{\tau_1 \in \fT_{V_1}}
	\sum_{\tau_2 \in \fT_{V \setminus V_1}}
	\sum_{v_1 \in V_1}
	\sum_{v_2 \in V \setminus V_1}
	f(\tau_1 \cup \tau_2 \cup \{\{ v_1,v_2 \}\}, \tau_1)
\end{equation}
for all functions $f$, and bearing in mind the independence of labelling of the vertices we have
\begin{equation}
	\sum_{\substack{\tau \in \fT \\ n(\tau) \ge 2}}
	\frac{\dd}{\dd \Lambda}
	\cS^{(\Lambda)}[\tau,\cV]
	=
	\frac12
	\sum_{\tau_1 \in \fT}
	\sum_{\tau_2 \in \fT}
	\fI_\many
	\cQ^{\fW,\Lambda} (
		\cS^{(\Lambda)}[\tau_1,\cV]
		,
		\cS^{(\Lambda)}[\tau_2,\cV]
	)
	\label{eq:dC_dLam_treesum}
\end{equation}
where $\fT_n$ denotes the set of trees with vertex set $\intrange1n = \left\{ 1,\dots,n \right\}$ and $\fT= \bigcup_{n=1}^\infty \fT_n$, again assuming that all the sums involved are absolutely convergent.
In the next section I give a series of estimates which make it possible to identify a set of $\cV$ on which this absolute convergence holds, and as a result letting
\begin{equation}
	\cS^{(\Lambda)}(\cV) 
	:=
	\sum_{\substack{\tau \in \fT \\ n(\tau) \ge 2}}
	\cS^{(\Lambda)}[\tau,\cV]
	\label{eq:cS_def}
\end{equation}
and $\cS(\cV) := \cS ^{(1)}(\cV)$,
this will satisfy 
\begin{equation}
	\fD_\Lambda \cS(\cV)
	=
	\sum_{\substack{\tau \in \fT \\ n(\tau) \ge 2}}
	\cS^{(\Lambda)}[\tau,\cV]
\end{equation}
and so also (recalling the definition of $\cF_\Lambda$ in \cref{eq:cF_def},
\begin{equation}
	\frac{\dd}{\dd \Lambda}
	\fD_\Lambda \cS(\cV)
	=
	\fI_\many
	\cF_\Lambda (\cV + \cS(\cV))
	\label{eq:cS_stationary}
\end{equation}
which, after introducing some other elements, I will use to construct $\cI_\many$ satisfying \cref{eq:cI_many_stationary}. 

\medskip

The first step in proving the aforementioned estimates will be to write out $\cS$ more explicitly as follows, starting with $\cS[\tau, \ul \Lambda, \cV]$.
To read and interpret what follows it may be helpful to associate the expression with a sum over Feynman diagrams, whose vertices are the vertices of $\tau$, although I prefer to validate the expression algebraically, so I will not specify all the details of such a representation.
Contractions are associated with factors of either $G_{\Lambda_e}$ (one for each edge $e$ of $\tau$, so these are the ``tree lines'' of \cite{DR}) or $P_{\Lambda(e)}$ (arising from the Pfaffians in \cref{eq:QW_integral}, the ``loop lines'' of \cite{DR}).
Given $\ul \Lambda$, each edge $e$ of $\tau$ is associated with the set of vertices, which I call a \emph{cluster} below, which are either incident to $e$ or connected to it by other edges $f$ with $\Lambda_f > \Lambda_e$ (see \cref{fig:cluster}); 
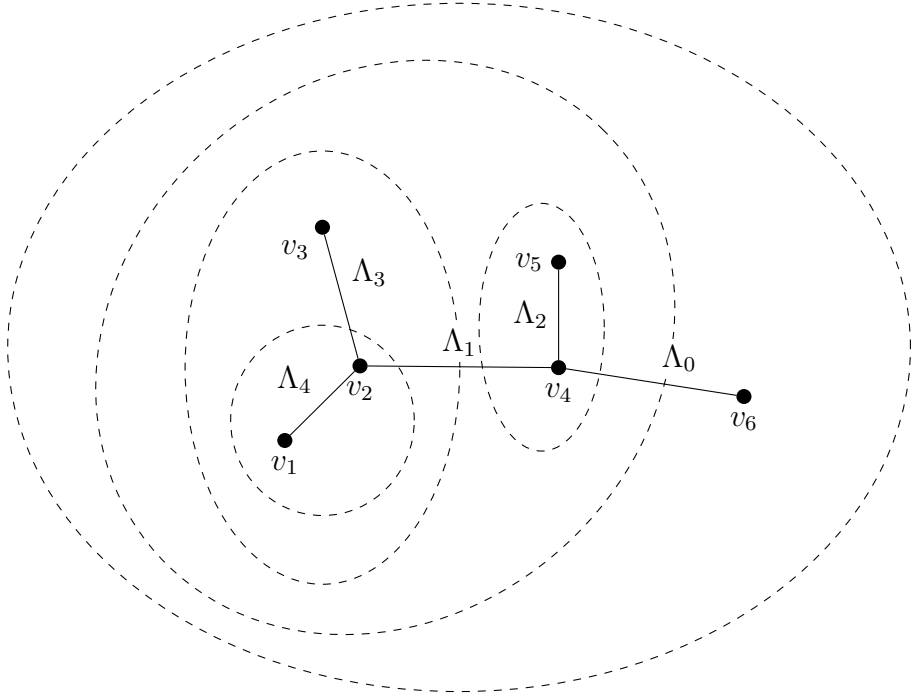
\begin{figure}[t]
	\centering
		\begin{tikzpicture}[node distance = 12mm,
				every node/.style={circle, inner sep = 2pt, outer sep = 0pt, fill=black},
				every label/.style={rectangle,fill=white, outer sep=1pt},
				every fit/.style={ellipse,draw,dashed, fill=none,on background layer,outer sep=0},
				edgeLabel/.style={rectangle,auto,fill=white,inner sep=2pt,outer sep=2pt}
			]
			\node (v1) [label={[name=label1] below:$v_1$}] {} ;
			\node (v2) [label={[name=label2] below:$v_2$},above right=of v1] {};
			\draw (v1) edge node[edgeLabel] {$\Lambda_4$} (v2);
			\node (c5) [fit = (v1) (label1) (v2) (label2)] {};
			\node (v3) [above=of c5,label={[name=label3] below left:$v_3$}] {};
			\draw (v3) edge node[edgeLabel] {$\Lambda_3$} (v2);
			\node (c4) [fit = (v3) (c5)] {};
			\node (v4) [right=of c4,label={[name=label4] below:$v_4$}] {};
			\draw (v2) edge node[edgeLabel] {$\Lambda_1$} (v4);
			\node (v5) [above=of v4,label={[name=label5] left:$v_5$}] {};
			\draw (v4) edge node[edgeLabel,name=L3] {$\Lambda_2$} (v5);
			\node (c3) [fit = (v4) (label4) (v5) (label5) (L3)] {};
			\node[rotate fit=45] (c2) [fit = (c4) (c3)] {};
			\node (v6) [below right=of c2,label={[name=label6] below:$v_6$}] {};
			\draw (v4) edge node[edgeLabel,name=L1] {$\Lambda_0$} (v6);
			\node[inner sep=10pt] (c1) [fit = (c2) (v6)] {};
		\end{tikzpicture}
	\caption{Example of a tree with $6$ vertices, a choice of $\ul \Lambda$, and the associated nontrivial clusters indicated by dashed curves.  The values of $\Lambda_e$ are written next to the associated edges, labelled such that $\Lambda_0 < \Lambda_1 <\dots < \Lambda_4$.}
	\label{fig:cluster}
\end{figure}
along with the trivial clusters consisting of a single vertex, these are the vertex sets of the subtrees appearing in \cref{eq:S_tau_Lambda_iter_def} when it is applied iteratively.  $\Lambda_e$ also assigns a scale to the cluster associated with $e$, and I will assign to trivial clusters the lowest scale of any incident edge.
Expanding the integral in each factor of $P$, each contraction is associated with a scale label in a way that respects the following: 
\begin{itemize}
	\item For each edge $e= \{v_1,v_2\}$ of $\tau$, one leg of $v_1$ and one leg of $v_2$ are contracted at scale $\Lambda_e$; I will sometimes call these \emph{tree legs}.
	\item The scale for any other contraction is less than $\Lambda_e$ for any $e$ associated with a cluster containing only one of the two legs involved; I will call these \emph{loop legs}.
	\item For each $e$ except for the one with the lowest scale, the number of legs in the associated cluster which are not contracted at scales $\ge \Lambda_e$ is more than $M$ (due the action of $\fI_\many$); I will say that these legs are {\emph external} to the cluster in question.  Note, however, that this does not mean that these legs are not contracted with another leg in the same cluster, only that the scale at which they are contracted is associated with another cluster.
\end{itemize}

In regrouping the diagrams, I will associate each leg $l$ with an edge $K(l)$, which in this context labels the largest cluster which contains $l$, such that $l$ is contracted at a scale $\Lambda_l \le \Lambda_e$.  Note that $K(l) = K(l')$ for all the pairs $l,l'$ which are contracted in a given diagram.
For completeness, for legs which are not contracted (external legs of the diagram) I let $K(l)$ be the lowest-scale edge, whose cluster is the entire vertex set. 
\medskip

I now give a more precise formulation implementing this idea.
For $\ul \ell: V(\tau) \to \{2,4,\dots,M\}$ (specifying the number of legs of each vertex), let
\begin{equation}
	\leg(\tau,\ul \ell):= \left\{ (v_1,1), (v_1,2),\dots,(v_1,\ell_1) , (v_2,1),\dots, (v_{n(\tau)},\ell_{n(\tau)}) \right\} 
\end{equation}
To begin with, explicitly iterating \cref{eq:S_tau_Lambda_iter_def} there is an action of $\fI_\many \cQ^{\fW,\Lambda}$ associated with each edge $e$, and writing these out using \cref{eq:QW_integral} and introducing indices $E(l) \in \edge(\tau)$ to indicate in which of these a particular argument appears in a factor of $G$ or $P$, 
\begin{equation}
	\begin{split}
		S_\ell[\tau,\ul \Lambda,\cV][f]
		&=
		\frac{1}{\ell!}
		\frac{1}{n(\tau)!}
		\sum_{\el : \intrange1\ell \rightarrowtail \leg}
		\sum_{\ul \ell : V(\tau) \to \{2,4,\dots,M\}}
		\sum_{T : \edge \to \leg^2}^*
		\sigma(\el,T)
		\\ & \qquad \times
		\sum_{E : \leg \to \edge  }
		\int
		f(y_{\el(1)},\dots,y_{\el(\ell)})
		\prod_{e \in \edge}
		G_{\Lambda_e} (y_{T_1(e)},y_{T_2(e)})
		\\ & \qquad \times
		\chitree (E,T)
		\chimono (\ul \Lambda, E)
		\chirel(\ul \Lambda, E)
		\Pf \check\fP_{\LL,E} (\bs y)
		\\ & \qquad \times
		\prod_{v \in V} \fD_{\Lambda(v)} V_{\ell(v)} (\bs y_v)
		\dd^{|\ul \ell| d} \bs y
		\label{eq:C_ell_tau_explicit_prelim}
	\end{split}
\end{equation}
for $\ell \in 2 \bN +2$, $f \in \slow_\ell$, $\tau \in \fT_V$,
where I have abbreviated $\leg = \leg (\tau, \ul \ell)$, $\edge = \edge(\tau)$ is the edge set of $\tau$,
\begin{itemize}
	\item The sum over $T$ is restricted to those maps with the property that $T_1(e) = (v_1,i)$, $T_2(e) = (v_2,j)$ with $e = \{v_1,v_2\}$, $v_1 < v_2$ (note that  $T$ indicates the pair of legs which produces the factor of $G_{\Lambda_e}$ associated with $e$) 
	\item $\sigma(\el,T)$ is the sign of the permutation taking the elements of $\range(\el)$ and $\bigcup_{e \in \edge} \left\{ T_1 (e),T_2(e)\right\}$ from their natural order into the order in which they appear in the product on the right hand side of \cref{eq:C_ell_tau_explicit}
	\item \LL\ is the set of $l \in \leg$ which do not appear in \el\ or $T$, and $\check \fP_{\LL,E} (\bs y)$ is the antisymmetric matrix with elements
		\begin{equation}
			\check \fP_{ij}(\bs y)
			=
			\begin{cases}
				P_{\Lambda(e)} (y_i,y_j)
				, &
				E(i) = E(j) = e, \ e \text{ separates } v(i),v(j)
				\\
				0,
				\textup{otherwise}
			\end{cases}
			\label{eq:fP_check_big}
		\end{equation}
		indexed by $i,j \in \LL$, where by ``separates'' I mean that $e$ is part of the unique path from $v(i)$ to $v(j)$ passes through $e$.  Up to reordering the indices, this matrix is block diagonal, with each block of the form $\check \fP_n$ in \cref{eq:fP_tilde_pair}.
	\item $\chitree(E,T)$ is the indicator function that $E(l) = e$ if $T_j(e) = l$ for some $j=1,2$, and if $l \in \range(\el)$ then  $E(l)=e^*$, where $e^* = e^*(\tau, \ul \Lambda)$ is the edge that minimizes $\Lambda_e$. 
	\item $\chimono(\ul \Lambda,E)$ is the indicator function that, for all $(v,k) \in \leg$, $\Lambda_{e_1} > \Lambda_{e_2} > \dots$ on the path from $v$ to $E(v,k)$ in $\tau$ (that is, the cluster with $E(v,k)$ as the edge of lowest scale includes $v$).
	\item $\chirel(\ul \Lambda,E)$ is the indicator function that $\EL(e) > M $ for each $e\in \edge$,
		where $\EL(e^*) = \ell$ for the edge $e^*$ minimizing $\Lambda_{e^*}$,
		and for the other $e \in \edge$, $\EL(e)$ is the number of $(v,k) \in \leg$ such that the path from $v$ to $K(v,k)$ passes through $e$ but does not end there (i.e.\ $K(v,k) \neq e)$; this corresponds to the action of $\fI_\many$ in \cref{eq:S_tau_Lambda_iter_def}.
	\item $\Lambda(v)$ is the maximum of $\Lambda_e$ over the edges $e$ incident to $v$.
\end{itemize}
Note that many of these quantities have names associated with the discussion of Feynman diagrams above to aid in identification (so \el\ corresponds to the choice of external legs, $\LL$ to the loop legs, and so on).

I further rewrite this by decomposing $\check \fP_{\LL,E}$ using the telescoping sum
\begin{equation}
	P_{\Lambda(e)}
	=
	(P_{\Lambda(e)} - P_{\Lambda(e')})
	+
	(P_{\Lambda(e')} - P_{\Lambda(e'')})
	+ \dots
	+
	(P_{\Lambda_0} - P_0)
	\label{eq:P_slices}
\end{equation}
where $e'$ is the edge maximizing $\Lambda_{e'} \equiv \Lambda'(e)$ among those with $e' \cap e \neq \emptyset$ and $\Lambda(e') < \Lambda_e$;
note that $P_0 = 0$.
Introducing a label $K(l)$ to indicate the edge appearing in the term associated with the leg $l$, and noting that in rewriting the term in \cref{eq:C_ell_tau_explicit_prelim} with a given $E$ gives a collection of $K$ such that $K(ell)$ is always on the path from $E(l)$ to $e^*$ in $\tau$, so that the indicator functions are unchanged, this gives
\begin{equation}
	\begin{split}
		S_\ell[\tau,\ul \Lambda,\cV][f]
		&=
		\frac{1}{\ell!}
		\frac{1}{n(\tau)!}
		\sum_{\el : \intrange1\ell \rightarrowtail \leg}
		\sum_{\ul \ell : V(\tau) \to \{2,4,\dots,M\}}
		\sum_{T : \edge \to \leg^2}^*
		\sigma(\el,T)
		\\ & \qquad \times
		\sum_{K : \leg \to \edge  }
		\int
		f(y_{\el(1)},\dots,y_{\el(\ell)})
		\prod_{e \in \edge}
		G_{\Lambda_e} (y_{T_1(e)},y_{T_2(e)})
		\\ & \qquad \times
		\chitree (K,T)
		\chimono (\ul \Lambda, K)
		\chirel(\ul \Lambda, K)
		\Pf \fP_{\LL,K} (\bs y)
		\\ & \qquad \times
		\prod_{v \in V} \fD_{\Lambda(v)} V_{\ell(v)} (\bs y_v)
		\dd^{|\ul \ell| d} \bs y
		,
		\label{eq:C_ell_tau_explicit}
	\end{split}
\end{equation}
where $\fP_{\LL,K}(\bs y)$ is the matrix with elements 
	\begin{equation}
		\fP_{ij}(\bs y)
		=
		\begin{cases}
			P_{\Lambda'(e)} (y_i,y_j)
			-
			P_{\Lambda(e)} (y_i,y_j)
			, &
			K(i)=K(j)=e
			\textup{ and } v(i) \neq v(j)
			\\
			0
			, &
			K(i) \neq K(j) \textup{ or } v(i) = v(j)
		\end{cases}
		\label{eq:fP_elems}
	\end{equation}
	indexed by $i,j \in \LL$,
	where $\Lambda'(e) = \max \left\{ \Lambda_{e'} \middle| e' \cap e \neq \emptyset, \Lambda_{e'} < \Lambda_e \right\} $ or zero if this maximum is over an empty set (note also $P_0 = 0$). 

\subsection{Resummation of the loop-leg assignment $K$}
\label{sec:regroup}

The main problem in using the above construction is that the number of terms in the sum over $E$ or $K$ is quite large, typically larger than $n!$, as is the number of $\tau \in \fT_n$.  I will using the same approach as \cite[Section~IV.2]{DR}, re-elaborated here for the current context.
The goal is to regroup the right hand side of \cref{eq:C_ell_tau_explicit} into a more manageable number of terms, while preserving the Pfaffian structure.
This would be straightforward without \chirel:
let
\begin{equation}
	\begin{split}
		\wt \fP_{ij} (\bs y)
		:=
		&
		\begin{cases}
			- P_{\Lambda(i,j)} (y_i,y_j)
			, &
			v(i) \neq v(j),
			\\
			0
			, &
			v(i) = v(j)
		\end{cases}
		\\ 
		= &
		\begin{cases}
			P_{\Lambda(i,j)} (y_i,y_j)
			, &
			v(i) = v(j)
			,\\
			0
			, &
			v(i) \neq v(j)
		\end{cases}
		-
		P_{\Lambda(i,j)} (y_i,y_j),
	\end{split}
	\label{eq:wt_fP_def}
\end{equation}
and let \chiduo\ be the indicator function of a set which is characterized equivalently by 
\begin{itemize}
	\item For any pair $v,w \in V(\tau)$, $\Lambda_{e_1} > \dots > \Lambda_{e_j} < \dots  <\Lambda_{e_m}$ for some $j$, where $e_1,\dots,e_m$ is the (unique) path in $\tau$ from $v_1$ to $v_2$, or
	\item $\Lambda_e$ has a unique minimum at some edge $e_*$, and is strictly increasing as one moves away from $e^*$ along $\tau$,
\end{itemize}
and 
with $\Lambda(i,j)$ the minimum of $\Lambda_e$ along the path from $v(i)$ to $v(j)$ if they are distinct, or the maximum over the incident edges if $v(i)=v(j)$.
Then writing out $\wt \fP$ as a sum over $K$ of $\fP_{\LL,K}$ defined in \cref{eq:fP_elems} $\fP$ and expanding the Pfaffian in rows and columns, 
\begin{equation}
	\chiduo (\ul \Lambda)
	\Pf \wt \fP_{\LL,\ul \Lambda} (\bs y)
	=
	\sum_{K : \leg \to \edge  }
	\chitree (K,T)
	\chimono (\ul \Lambda, K)
	\Pf \fP_{\LL,K} (\bs y)
\end{equation}
reproducing the correct sum over $K$.

The last expression in \cref{eq:wt_fP_def} is the difference of two terms each with a suitable Gram form:
concretely, let $\hat e_l (\Lambda)$ be one of the unit basis vectors of $\bR^{n(\tau)}$, chosen in such a way that $\hat e_k (\Lambda) = \hat e_l (\Lambda)$ iff $v(k)$ and $v(l)$ are part of the same connected component of the graph obtained from $\tau$ by deleting all edges with $\Lambda_e < \Lambda$; then
\begin{equation}
	\hat e_k(\Lambda) \cdot \hat e_l(\Lambda)
	=
	\ind_{[0,\Lambda(k,l)]} (\Lambda)
	, \quad
	\Lambda \in [0,\Lambda_k \wedge \Lambda_l]
\end{equation}
where $\Lambda_l=\Lambda_{v(l)}$ is the minimum of $\Lambda_e$ over the edges $e$ incident to $v(l)$.  Then
\begin{equation}
	\begin{split}
		P_{\Lambda(i,j)}
		& =
		\int_0^{\Lambda(i,j)}
		G_\Lambda
		\dd \Lambda
		=
		\int_0^{\min(\Lambda_i,\Lambda_j)}
		\hat e_i (\Lambda) \cdot \hat e_j (\Lambda)
		G_\Lambda 
		\dd \Lambda
		\\ &=
		\int_0^\infty
		\ind_{[0,\Lambda_i]}(\Lambda)
		\ind_{[0,\Lambda_j]}
		\hat e_i (\Lambda) \cdot \hat e_j (\Lambda)
		G_\Lambda 
		\dd \Lambda
		.
	\end{split}
	\label{eq:P_tree_Gram_trick_basic}
\end{equation}

\medskip

The presence of $\chirel$ (which, remember, corresponds to the presence of $\fI_\many$ in \cref{eq:S_tau_Lambda_iter_def}, and is needed for the integral in \cref{eq:C_tau_def} to be absolutely convergent) makes it impossible to regroup the sum into a single term without breaking the Pfaffian structure, but it is possible to allocate it as follows into a manageable number of terms each which can be so regrouped by the following procedure, which unfortunately requires more preliminary definitions.  

Given $\tau, \ul \Lambda, \ul l, \el, T$ for which $\chiduo(\ul \Lambda) =1$ and $\supp \chitree(T,\cdot) \neq \emptyset$, let $\IG$ be the set of $K$ for which $\chimono \chitree \chirel = 1$.
I define  $\omega : \IG \to (V \cup \edge)^\LL$ as follows.  As will become clearer, $\omega_l (K)$ can be thought of as specifying the last (lowest-scale) cluster to which $l$ is required to be an external leg in order for $K$ to be in the support of $\chirel$.

I define a partial order on $V \cup \edge$ which is the minimal one for which $v \succ e$ whenever $e$ is incident to $v$, and $e \prec e'$ whenever $e \cap e' \neq \emptyset$ and $\Lambda_e < \Lambda_{e'}$ (this corresponds to ordering by inclusion of the associated clusters, see \cref{fig:Hesse}).
	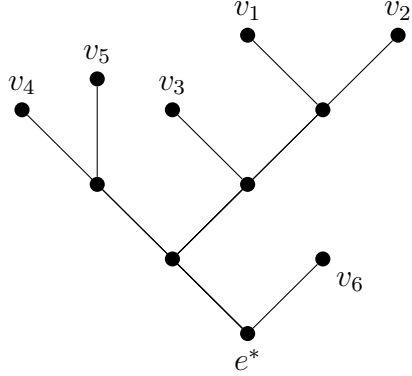
\begin{figure}[t]
		\centering
			\begin{tikzpicture}[node distance = 12mm,
				every node/.style={circle, inner sep = 2pt, outer sep = 0pt, fill=black},
				every label/.style={rectangle,fill=white, outer sep=1pt}
			]
			\node (e1) [label=below:$e^*$]{} ;
			\node (v6) [above right=of e1,label=below right:$v_6$] {};
			\draw (e1) edge (v6) ;
			\node (e2) [above left = of e1] {};
			\node (e3) [above left = of e2] {};
			\node (v4) [above left = of e3, label = above:$v_4$] {};
			\draw (e1) edge (e2) edge (e3) edge (v4);
			\node (v5) [above = of e3, label=above:$v_5$] {};
			\draw (e3) edge (v5);
			\node (e4) [above right= of e2] {};
			\node (v3) [above left =of e4,label=above:$v_3$] {};
			\draw (e4) edge (v3);
			\node (e5) [above right=of e4] {};
			\node (v1) [above left = of e5,label=above:$v_1$] {};
			\node (v2) [above right = of e5,label=above:$v_2$] {};
			\path (e2) edge (e4) edge (e5) edge (v2);
			\draw (e5) edge (v1);
		\end{tikzpicture}
		\caption{The Hesse diagram $\fH$ of the partial order $\succ$ associated with the tree and scale assignment in \cref{fig:cluster} (that is, the elements of $V \cup \edge$ are drawn as nodes, with edges such that $j \succ k$ iff there is an upwards path from $k$ to $j$).  The unlabelled nodes correspond to edges of the tree $\tau$, other than the minimal edge $e^*$ which is labelled.}
		\label{fig:Hesse}
	\end{figure}
I abbreviate $\omega_l = \omega_l (K)$.
For each $v$ let $\leg(v)$ be the set of $(w,k) \in \leg(\tau)$ such that $w \succeq v$.
Proceeding through $ V\cup \edge$ according to this partial order,
\begin{enumerate}
	\item For the edge minimizing $\Lambda_e$ (and so also the partial order) do nothing.
	\item For any other edge $e$,
		let $W_e$ be the set of $l \in \leg(e) = \bigcup_{v \succ e} \leg(v)$ for which either
		\begin{itemize}
			\item $l \in \range(\el)$,
			\item $l = T_j(e')$ with $j\in \{1,2\}$ and $e' \prec e$,
				or
			\item $\omega_l \prec e$ (note that this will already have been decided).
		\end{itemize} 
		Then if $\# W_e \le M$, let $\omega_l =e$ for the first $M + 1 - \# W_e$ elements (in the lexicographic order) $l$ of $(\leg(e) \cap \LL) \setminus W_e$ such that $K(l) \prec e$.

		Note that this is always possible on the support of \chirel.
	\item For a vertex $v$ (so that $\nexists k \in V \cup \edge : k \succ v$), let $\omega_l = v$ for any $l \in \LL \cap \leg(v)$ for which it has not yet been assigned.
\end{enumerate}

\begin{remark}
	\label{rem:omega_characterization}
	Note that this has the following properties:
	\begin{enumerate}
		\item $\omega_l \neq e^*$ for all $l \in \LL$, where $e^*$ is the edge minimizing $\Lambda_e$.
			\label{it:omega_e0}
		\item If $e \in \edge \setminus \{e^*\}$,  then $\LL(e) := \# \left\{ l \in \LL \middle| \omega_l \preceq e \right\} \le M$ (since $W_e$ is always nonempty).
			\label{it:omega_few}
		\item \label{it:omega_enough}
			On the other hand, letting 
			\begin{equation}
				\TL(e)
				:=
				\# \leg(e) \cap \bigcup_{e' \prec e} \left\{ T_1(e'),T_2(e') \right\} 
				=
				\# \left\{ e'  \prec e, j \in \{1,2\} \middle| T_j (e') \succ e \right\}
				\label{eq:TL_def}
			\end{equation}
			and $\EL(e) := \# [\leg(e) \cap \range (\el)]$,
			\begin{equation}
				\LL(e) + \TL(e) + \EL(e)
				\ge
				M +1
				.
				\label{eq:omega_enough}
			\end{equation}
		\item \label{it:omega_preimage}
			For any $\lis \omega \in \omega(\IG)$, 
			$\lis \omega_l \prec v(l) \ \forall l$ and
			$\omega^{-1} (\{\lis \omega\})$ is exactly the set of $K : \leg \to V \cup \edge$ such that
			\begin{itemize}
				\item $\chitree(T,K) =1$, and
				\item $K (l) \prec \lis \omega_l$ (in particular this implies $\chimono=1$), and
				\item If $k,l \in \LL$ with $k > l$, $l \in \leg(\lis \omega_k)$ and $\lis \omega_k \prec \lis \omega_l$, then $K(l) \succ \lis \omega_k$ (otherwise $l$ rather than $k$ will be picked to have $\omega_l (K) = \lis \omega_k$),
			\end{itemize}
			and this set is nonempty.
	\end{enumerate}
\end{remark}

This makes it possible to reorganize the sum without disrupting the Pfaffian form and prove useful bounds:
as a result of \cref{it:omega_preimage} of \cref{rem:omega_characterization}, there exist $\omega^-_l =  \omega^-_l (\lis \omega) \in \edge $ such that
\begin{equation}
	\begin{split}
		\sum_{K : \leg \to \edge  }
		\chitree (K,T)
		\chimono (\ul \Lambda, K)
		&
		\chirel(\ul \Lambda, K)
		\Pf \fP_{\LL,K} (\bs y)
		\\ &
		=
		\chiduo (\ul \Lambda)
		\sum_{\lis \omega \in \omega(\IG)}
		\Pf \wt \fP_{\LL,\ul \Lambda}^{\lis \omega} (\bs y)
	\end{split}
	\label{eq:K_sum_to_omega_sum}
\end{equation}
with
\begin{equation}
	\wt \fP^{\lis \omega}_{kl}
	=
	\begin{cases}
		\wt P^{\lis \omega,\ul \Lambda}_{kl} (\bs y)
		, &
		v(k) = v(l)
		, \\
		0
		, &
		v(k) \neq v(l)
	\end{cases}
	-
	\wt P^{\lis \omega,\ul \Lambda}_{kl} (\bs y)
\end{equation}
and
\begin{equation}
	\begin{split}
		\wt P^{\lis \omega,\ul \Lambda}_{kl} (\bs y)
		&=
		\sum_{\substack{e \in \edge: 
				\\ \omega^-_k \preceq e \prec \lis \omega_k , 
				\\ \omega^-_l \preceq e \prec \lis \omega_l  
		}}
		[P_{\Lambda' (e)}- P_{\Lambda(e)}]
		=
		\int_{\max(\Lambda'(\omega^-_k),\Lambda'(\omega^-_l))}^{\min(\Lambda(\lis \omega_k), \Lambda(\lis \omega_l),\Lambda(k,l))}
		G_\Lambda \dd \Lambda
		\\ & 
		=
		\int_0^\infty
		\ind_{[\Lambda'(\omega^-_k),\Lambda(\lis \omega_k)]}(\Lambda)
		\ind_{[\Lambda'(\omega^-_l),\Lambda(\lis \omega_l)]}(\Lambda)
		\hat e_i (\Lambda) \cdot \hat e_j (\Lambda)
		G_\Lambda
		\dd \Lambda
	\end{split}
	\label{eq:P_tree_Gram_advanced}
\end{equation}
using the same definitions as in \cref{eq:P_tree_Gram_trick_basic}.
This Gram form makes it possible to efficiently bound the individual terms on the right hand side of \cref{eq:K_sum_to_omega_sum} using the methods presented in \cref{sec:Pfaff}.
For this to be useful we also need to check that there are not too many terms in the resulting sum, as expressed in the following bound.
\begin{lemma}
	For any $\tau \in \fT_n$ and $\ul \Lambda$ with $\chiduo=1$,
	\begin{equation}
		\sum_{\el : \intrange1\ell \rightarrowtail \leg}
		\sum_T^* 
		\# \omega (\IG)
		\le
		n^\ell
		2^{n-1}
		(2^{M+1})^{n-2}
		\ul \ell!
		\label{eq:omega_count}
	\end{equation}
	\label{lem:omega_count}
\end{lemma}
\begin{proof}
	Enumerate as follows:
	\begin{enumerate}
		\item For each $j=1,\dots,\ell$, chose the $v$ for which $\el(j) \in \leg(v)$ ($n^l$ choices).
		\item For each $e$, choose which of the incident vertices provides $T_1(e)$ and which provides $T_2(e)$ ($2^{n-1}$ choices).
		\item Pick the multiset $\mu$ of pairs $(v(l),\omega_l)$ for $l \in \LL$ with $\omega_l \neq v(l)$.
			\label{it:path_enumerate}
		\item Finally, for each vertex $v$, choose which elements of $\leg(v)$ to assign to each of the roles above (at most $\ell_v!$ choices).
	\end{enumerate}
	I now bound the number of choices in step~\ref{it:path_enumerate} by $(2^{M+1})^{n-2}$ to complete the proof.
	It is convenient to visualize the argument in terms of the Hesse diagram $\fH$ of the partial order $\prec$ on $V \cup \edge$ (see \cref{fig:Hesse});
	this is a  tree, and each vertex has at most one line entering from below and two lines leaving above (this resembles a Gallavotti-Nicolò tree set on its side, except that the ordering of the ``endpoints'' is not taken into account). 
	Recalling \cref{it:omega_e0,it:omega_few} of \cref{rem:omega_characterization},
	if I associate each element $(v,u)$ of the multiset $\mu$ with the path from $v$ to $u$ in $\fH$, there are at most $M$ such paths starting at or passing through each $e$ (a n edge of $\tau$ which corresponds to a vertex of $\fH$) and none at $e^*$ (the lowest vertex of $\fH$).  I can generate all of the multisets of upwards paths with this property iteratively, starting at the neighbors of $e^*$ and moving up: at each $e$ I decide for each path arriving from below which of the two upwards directions it continues in, decide how many paths will start at $e$, and which direction each of them starts in; with $i$ incoming paths this gives at most
	\begin{equation}
		\sum_{j=i}^M 2^j
		<
		2^{M+1}
	\end{equation}
	possibilities.
\end{proof}

\subsection{Estimates on tree terms}
\label{sec:tree_bounds}

I will now describe some estimates on the right hand side of \cref{eq:C_ell_tau_explicit}, which show that the integrals are well defined (the sums are all finite, thanks to the restricted choice  of $\ul \ell$ corresponding to the restriction to $\cV \in \pot_{\le M}$) and that the integrals in \cref{eq:C_tau_def} and sums over trees such as those in \cref{eq:dC_dLam_treesum} are absolutely convergent for a certain set of $\cV$.

\begin{remark}
	Unless otherwise stated, all of the numbered constants introduced in this section depend only on $M \in 2\bN+2$, $[\psi], |G|_+, C_{\GH} \in (0,\infty)$, 
	\begin{equation}
		\Dmin^\many := 
		(M+1)[\psi] - d
		> 0
		, \text{ and }
		\Dmin^\inter :=
		\min_{\ell=2,4,\dots,M}
		(\floor{d - \ell [\psi]} + 1 + \ell [\psi] - d)
		> 0,
	\end{equation}
	and can be chosen to be uniform on compact subintervals of $(0,\infty)$, hence uniformly for $\varepsilon \in [0,d/6) $.
	\label{rem:uniform_constants}
\end{remark}

The first step is to estimate the individual terms.  

Consequently, using \cref{lem:Gram_diff}, 
\begin{equation}
	\label{eq:decomp_Gram_bound}
	\left| 
	\Pf \wt \fP_{\LL,\ul \Lambda}^{\lis \omega} (\bs y)
	\right|
	\le
	\prod_{l \in \LL}
	\sqrt 2 [\Lambda(\lis \omega)]^{[\psi]}
	C_{\GH}^{\LL}
	.
\end{equation}
Noting that 
\begin{equation}
	\Lambda(\omega_l)
	=
	\Lambda(v(l))
	\prod_{v(l) \succ e \succeq \omega_l}
	\frac{\Lambda'(e)}{\Lambda(e)}
	\label{eq:Lambda_as_product_quotient}
\end{equation}
this becomes
\begin{equation}
	\begin{split}
		\left| 
		\Pf \wt \fP_{\LL,\ul \Lambda}^{\lis \omega} (\bs y)
		\right|
		&\le
		C_{\GH}^{\LL}
		\prod_{e \in \edge \setminus \{e^*\}}
		\left[ \frac{\Lambda'(e)}{\Lambda(e)} \right]^{[\psi] \LL(e)}
		\prod_{v \in V}
		\left( 
			\sqrt2
			[\Lambda(v)]^{[\psi]}
		\right)^{\LL(v)}
		\\ & \le
		C_{\GH}^{\LL}
		2^{(|\ul \ell| - \ell - 2n + 2)/2}
		\begin{aligned}[t]
			&
			\prod_{e \in \edge \setminus \{e^*\}}
			\left[ \frac{\Lambda'(e)}{\Lambda(e)} \right]^{[\psi] \max(M+1 - \EL(e)-\TL(e),0)}
			\\ & \times \prod_{v \in V}
			[\Lambda(v)]^{[\psi] (\ell_v - \EL(v) - \TL(v))}
		\end{aligned}
	\end{split}
	\label{eq:omega_Pfaff_bound}
\end{equation}
using \cref{it:omega_enough} of \cref{rem:omega_characterization}.  Note that the last bound is uniform in $\lis \omega$.

Using \cref{lem:omega_count,eq:omega_Pfaff_bound} along with \cref{lem:contraction_bound_generic,eq:V_j_scaling} gives
\begin{equation}
	\begin{split}
		\left\|
		S_\ell^{(\Lambda)}[\tau,\cV]
		\right\|
		&\le
		\begin{aligned}[t]
			\frac{\C^n}{n!}
			2^{-\ell/2}
			\sum_{\ul \ell : V(\tau) \to \{2,4,\dots,M\}}
			&
			\max_{\el : \intrange1\ell \rightarrowtail \leg}
			C_{\GH}^{\LL}
			\\
			\times
			\int_{(\Lambda,\infty)^\edge}
			&
			\chiduo(\tau,\ul \Lambda)
			\prod_{e \in \edge} \Lambda(e)^{2[\psi]-1-d}
			\\  \times &
			\prod_{e \in \edge \setminus \{e^*\}}
			\left[ \frac{\Lambda(e)}{\Lambda'(e)} \right]^{[\psi] \min (\EL(e)-\TL(e)-M-1,0)}
			\\  \times &
			\prod_{v \in V}  2^{\ell_v/2} \ell_v ! \Lambda(v)^{d-(\EL(v)+\TL(v))[\psi]} \|V_{\ell(v)}\|
			\dd \ul \Lambda
		\end{aligned}
		\\ & \le
		\begin{aligned}[t]
			\frac{\Cr{C_tree_basic}^n}{n!}
			2^{-\ell/2}
			C_\GH^{-\ell}
			&
			\begin{aligned}[t]
				\int_{(\Lambda,\infty)^\edge}
				&
				\max_{\substack{ \wt \EL : V \to \bN \\ \sum_{v \in V} \wt \EL(v) = \ell}}
				\chiduo(\tau,\ul \Lambda)
				\Lambda(e^*)^{d - 1 - \ell [\psi]}
				\\ & \times
				\prod_{e \in \edge \setminus \{e^*\}}
				\Lambda(e)^{-1}
				\left[ \frac{\Lambda(e)}{\Lambda'(e)} \right]^{-D(e)}
				\dd \ul \Lambda
			\end{aligned}
			\\ &  \times 
			\prod_{v \in V} 
			\sum_{\ell_v = 2,4,\dots,M}
			2^{\ell_v/2} \ell_v ! C_{\GH}^{\ell_v} \|V_{\ell(v)}\|
		\end{aligned}
	\end{split}
	\label{eq:C_sum_integral_bound_basic}
\end{equation}
where $D(e) :=  [\psi] \max(M+1,\wt\EL(e) + \TL(e)) - d > 0$ with $\wt \EL(e) := \sum_{v \succ e} \wt \EL(v)$,
and $\Cl{C_tree_basic} = e 2^{M+1} |G|_+$, which depends only on $M$ and $|G|_+$.
Note that  
\begin{equation}
	D(e) 
	\ge 
	\Dmin
	:=
	\min_{\substack{\ell \in 2 \bN, \\ \ell > M}}
	(\ell [\psi]-d)
	\label{eq:Dmin_def}
\end{equation}
and there is a $\Cl{dim_TL} > 0$ depending only on \Dmin such that 
\begin{equation}
	D(e) \ge \Clast [\psi] \TL (e).
\end{equation}
If $ \ell > M$, then \cref{eq:C_sum_integral_bound_basic} implies
\begin{equation}
	\left\|
	S_\ell^{(\Lambda)}[\tau,\cV]
	\right\|
	\le
	\begin{aligned}[t]
		&
		\frac{\Cr{C_tree_basic}^n}{n!}
		2^{-\ell/2}
		C_\GH^{-\ell}
		\Lambda^{d - 1 - \ell [\psi]}
		\prod_{v \in V} 
		\sum_{\ell_v = 2,4,\dots,M}
		2^{\ell_v/2} \ell_v ! C_{\GH}^{\ell_v} \|V_{\ell(v)}\|
		\\ & \times
		\begin{aligned}[t]
			\int_{(1,\infty)^\edge}
			&
			\chiduo(\tau,\ul \Lambda)
			\left[ \Lambda(e^*) \right]^{-\Cr{dim_TL} [\psi]/2}
			\\ & \times
			\prod_{e \in \edge}
			\Lambda(e)^{-1}
			\left[ \frac{\Lambda(e)}{\Lambda'(e)} \right]^{-\Cr{dim_TL} [\psi] \TL(e)}
			\dd \ul \Lambda
		\end{aligned}
	\end{aligned}
	\label{eq:C_sum_integral_bound}
\end{equation}
note that the minimal edge $e^*$ is the only edge with $\TL(e)=0$.

\begin{lemma}
	For any choice of $M,[\psi]$ there exist
	$\Cl{tree_integral_start}$ and $\Cl{tree_integral} $
	such that for all $\tau \in \fT_n$, $n \ge 2$,
	\begin{equation}
		\begin{split}
			\int_{(1,\infty)^\edge}
			\chiduo(\tau,\ul \Lambda)
			\left[ \Lambda(e^*) \right]^{-\Cr{dim_TL} [\psi]/2}
			\prod_{e \in \edge}
			\frac1{\Lambda(e)}
			\left[ \frac{\Lambda(e)}{\Lambda'(e)} \right]^{-\Cr{dim_TL} [\psi] \TL(e)}
			\dd \ul \Lambda
			& \\[1ex]
			\le
			\Cr{tree_integral_start} 
			\Cr{tree_integral}^{n-2}
			.
			&
		\end{split}
		\label{eq:tree_integral_bound}
	\end{equation}
	\label{lem:tree_integral_bound}
\end{lemma}
\begin{proof}
	I iterate in $n$, starting with $n=2$ where there is only one edge $e=e^*$, $\TL(e) = 0$ and the estimate reduces to
	\begin{equation}
		\int_1^\infty
		\Lambda(e^*)^{-1-\Cr{dim_TL} [\psi]/2}
		\dd \Lambda
		\le 
		\Cr{tree_integral_start}
	\end{equation}
	for which it suffices to choose $\Cr{tree_integral_start}$ appropriately.

	Let 
	\begin{equation}
		\Cr{tree_integral}
		:=
		2
		\int_1^{\infty}
		\Lambda^{-1 - \Cr{dim_TL}[\psi]/2}
		\dd \Lambda
		=
		\int_0^\infty
		\frac{1}{\Lambda}
		\min \left( \frac{\Lambda'}{\Lambda}, \frac{\Lambda}{\Lambda'} \right)^{\Cr{dim_TL}[\psi]/2};
	\end{equation} the last equality holds for any $\Lambda' \in (0,\infty)$.

	Proceeding iteratively, for any $\tau \in \fT_n$ let $v$ be a vertex of degree one (there are always at least two) and let $v'$ be its (unique) neighbor.  Let $\wt e^*$ be the edge minimizing of $\Lambda_e$ considering only $e \in \edge(\wt \tau)$, and let $\wt \Lambda (v')$ be the maximum of $\Lambda_e$ over $e \in \edge(\wt \tau)$ incident to $v'$, and let $\tilde e$ be the associated edge.
	If $\Lambda_{\{v,v'\}} < \Lambda_{\tilde e}$, then $\{v,v'\} = e^* \prec \tilde e^* \prec \tilde e$, 
	and
	\begin{equation}
		\wt \TL(e) 
		=
		\TL(e)
		-
		\begin{cases}
			1
			, &
			\{v,v'\} \prec e \prec v'
			, \\
			0
			, &
			\text{otherwise}
		\end{cases}
	\end{equation}
	and $\TL(e^*) = 0$,
	so 
	\begin{equation}
		\begin{split}
			\left( \frac{\Lambda}{\Lambda_{e^*}} \right)^{1/2}
			\prod_{e \in \edge(\tau)}
			\left[ \frac{\Lambda(e)}{\Lambda'(e)} \right]^{-\TL(e)}
			&=
			\left( \frac{\Lambda}{\Lambda_{\{v,v'\}}} \right)^{1/2}
			\frac{\Lambda_{\{v,v'\}}}{\Lambda_{\tilde e}}
			\prod_{e \in \edge(\tilde \tau)}
			\left[ \frac{\Lambda(e)}{\Lambda'(e)} \right]^{-\wt\TL(e)}
			\\ &\le
			\left( \frac{\Lambda}{\Lambda_{\tilde e^*}} \right)^{1/2}
			\left( \frac{\Lambda_{\{v,v'\}}}{\Lambda_{\tilde e}} \right)^{1/2}
			\prod_{e \in \edge(\tilde \tau)}
			\left[ \frac{\Lambda(e)}{\Lambda'(e)} \right]^{-\wt\TL(e)}
			.
		\end{split}
	\end{equation}
	If instead $\Lambda_{\{v,v'\}} \in (\wt \Lambda_0, \wt \Lambda (v'))$, $e^* = \tilde e^*$ and the same estimate holds because
	\begin{equation}
		\begin{split}
			\left( \frac{\Lambda}{\Lambda_{e^*}} \right)^{1/2}
			\prod_{e \in \edge(\tau)}
			\left[ \frac{\Lambda(e)}{\Lambda'(e)} \right]^{-\TL(e)}
			&=
			\left( \frac{\Lambda}{\Lambda_{\tilde e^*}} \right)^{1/2}
			\frac{\Lambda_{\{v,v'\}}}{\Lambda_{\tilde e}}
			\prod_{e \in \edge(\tilde \tau)}
			\left[ \frac{\Lambda(e)}{\Lambda'(e)} \right]^{-\wt\TL(e)}
			.
		\end{split}
	\end{equation}

	On the other hand if $\Lambda_{\{v,v'\}} > \wt \Lambda (v')$, then $\TL(\{v,v'\})=1$, $\Lambda'(v) = \Lambda_{\{v,v']\}}$ and $\wt \TL(e) = \TL(e)$ in $\tilde \tau$, and so
	\begin{equation}
		\begin{split}
			\left( \frac{\Lambda}{\Lambda_{e^*}} \right)^{1/2}
			\prod_{e \in \edge(\tau)}
			\left[ \frac{\Lambda(e)}{\Lambda'(e)} \right]^{-\TL(e)}
			&=
			\left( \frac{\Lambda}{\Lambda_{\{v,v'\}}} \right)^{1/2}
			\frac{\Lambda_{\tilde e}}{\Lambda_{\{v,v'\}}}
			\prod_{e \in \edge(\tilde \tau)}
			\left[ \frac{\Lambda(e)}{\Lambda'(e)} \right]^{-\wt\TL(e)}
			\\ &\le
			\left( \frac{\Lambda}{\Lambda_{\tilde e^*}} \right)^{1/2}
			\left( \frac{\Lambda_{\{v,v'\}}}{\Lambda_{\tilde e}} \right)^{1/2}
			\prod_{e \in \edge(\tilde \tau)}
			\left[ \frac{\Lambda(e)}{\Lambda'(e)} \right]^{-\wt\TL(e)}
			.
		\end{split}
	\end{equation}
	Then the desired estimate follows immediately from the case with $n-1$.
\end{proof}

Bearing in mind that $\# \fT_n = n^{n-2}$ (Cayley's formula),
\begin{corollary}
	For some $\Cl{tree_sum_power},\Cl{tree_sum_pref}$,
	for all positive even $\ell > M$,
	\begin{equation}
		\Cr{tree_sum_power}^{\ell}
		\sum_{n=2}^\infty
		\sum_{\tau \in \fT_n}
		\left\|
		S_\ell^{(\Lambda)}[\tau,\cV]
		\right\|
		\le
		\Lambda^{d - \ell [\psi]}
		\sum_{n=2}^\infty
		\left[ 
			\Cr{tree_sum_pref}
			\sum_{k = 2,4,\dots}
			\Cr{tree_sum_power}^k
			k! \ 
			\| V_k \|
		\right]^n
		.
	\end{equation}
	\label{lem:tree_sum_nice}
\end{corollary}
For later convenience, I will take $\Cr{tree_sum_power} \ge 4 C_\GH^{1/2}$, although this is larger than necessary here.

As a result of \cref{lem:tree_sum_nice}, the sum over trees defining $\cS(\cV)$ in \cref{eq:cS_def} is absolutely convergent whenever
\begin{equation}
	\rho[\cV]
	:=
	\Cr{tree_sum_pref}
	\sum_{k = 2,4,\dots}
	\Cr{tree_sum_power}^k
	k! \ 
	\| V_k \|
	< 1
	,
	\label{eq:convergence_requisite}
\end{equation}
justifying \cref{eq:cS_stationary}.  Furthermore, with this $\Cr{tree_sum_power}$ all of the sums and integrals defining $\fW^{-1}_\Lambda \cS^{(\Lambda)}[\cV]$ and $\tfrac{\dd}{\dd \Lambda} \fW^{-1}_\Lambda \cS^{(\Lambda)}[\cV]$ are absolutely convergent near $\Lambda=1$.

When $\rho[\cV] \le 1/2$, indeed
\begin{equation}
	\mu[\cS^{(\Lambda)}[\cV]]
	:=
	\sup_{\ell = 2,4,\dots}
	\Cr{tree_sum_power}^{\ell}
	\left\|
	S_\ell^{(\Lambda)}[\cV]
	\right\|
	\le 
	2
	\Lambda^{d - \ell [\psi]}
	\rho[\cV]^2
	.
	\label{eq:cSl_norm_bound}
\end{equation}

\begin{remark}
	This is also true without restricting to $\cV \in \pot_{\le M}$; however the property \cref{eq:convergence_requisite} is not preserved by $\cS$.
	The main advantage of using a tree expansion-based ansatz is that it can accommodate this.
	\label{rem:why_trees}
\end{remark}

Proceeding similarly,
\begin{lemma}
	For any $\cV,\cW \in \pot,$ 
	\begin{equation}
		\begin{split}
			\Cr{tree_sum_power}^{\ell}
			&
			\left\|
			\cS_\ell^{(\Lambda)} [\cV]
			-
			\cS_\ell^{(\Lambda)} [\cW]
			\right\|
			\le 
			\Cr{tree_sum_power}^{\ell}
			\sum_{n=2}^\infty
			\sum_{\tau \in \fT_n}
			\left\|
			S_\ell^{(\Lambda)}[\tau,\cV]
			-
			S_\ell^{(\Lambda)}[\tau,\cW]
			\right\|
			\\ &
			\le
			\Lambda^{d - \ell [\psi]}
			\sum_{n=2}^\infty
			\sum_{m=0}^{n-1}
			\rho[\cV]^m
			\rho[\cV - \cW]
			\rho[\cW]^{n-m-1}
			\\ &
			\le
			\Lambda^{d - \ell [\psi]}
			\left( 
			\sum_{m=0}^\infty
			\rho[\cV]^m
			\sum_{n=0}^\infty
			\rho[\cW]^n
			-1
			\right)
			\rho[\cV - \cW],
		\end{split}
	\end{equation}
	with the same constants as in \cref{lem:tree_sum_nice}.
	\label{lem:cS_diff_estimate}
\end{lemma}

\begin{corollary}
	For $\rho[\cV],\rho[\cW] \le 1/3$,
	\begin{equation}
		\mu
		\left[
			\cS [\cV]
			-
			\cS [\cW]
		\right]
		\le 
		2
		(\rho[\cV]+\rho[\cW])
		\rho[\cV-\cW]
		.
		\label{eq:cS_diff_simple_bound}
	\end{equation}
	\label{lem:cS_diff_simple_bound}
\end{corollary}

Recalling \cref{eq:QW_summable} and that $\Cr{tree_sum_power} \ge 4 C_\GH^{1/2}$,
\begin{equation}
	\left\| Q_\ell^\fW (\cV,\cW) \right\|
	\le
	C_\GH^{- ( \ell/2) - 1}
	|G|_+
	\sum_{\substack{j,k \in 2 \bN + 2}}
	\left( \frac{3}{4} \right)^{j+k}
	\Cr{tree_sum_power}^{j+k}
	\| V_j \|
	\| W_k \|
	\le
	C_\GH^{- ( \ell/2) - 1}
	|G|_+
	\mu[\cV]
	\mu[\cW]
	\label{eq:QW_rho_bound}
\end{equation}
which along with \cref{eq:cSl_norm_bound} shows that $\cS[\cV]$ is in the domain of $\cF$ whenever $\rho[\cV] \le 1/2$, and that the sums on the right hand side of \cref{eq:dC_dLam_treesum} are absolutely convergent, justifying \cref{eq:cS_stationary}.

As another consequence of \cref{eq:QW_rho_bound}, there is a $\Cl{Q_rho}$ such that
\begin{equation}
	\rho[ \fL \cQ^\fW(\cV,\cW) ]
	, \ 
	\rho[ \fI_\inter \cQ^\fW(\cV,\cW) ]
	\le 
	\Clast
	\mu[\cV] \mu [\cW]
	\label{eq:Q_rho_bound}
\end{equation}
where I use the fact that $\fL$ and $\fI_\inter$ produce kernels which vanish for degree larger than $M$ to bound the factorial in \cref{eq:convergence_requisite}.  Note that this can be combined with \cref{eq:cSl_norm_bound,lem:cS_diff_estimate}.

\subsection{Concluding the ansatz}
\label{sec:concluding_ansatz}

I will now define the interpolated part $\cI_\inter$ of the fixed point ansatz, which will also be the last missing piece for the definition of $\cI_\many$.  As mentioned in the introduction, it would be possible to incorporate this in the tree expansion as well (as in \cite{DR}) but I find it simpler to construct via a fixed point argument as below.

\begin{theorem}
	There exist $\Cl{cI_cV},\Cl{cI_norm}$ such that for all $\cV \in \pot_{\le M}$ with $\rho[\cV] < \Cr{cI_cV}$ there exists a unique $\cI_\inter =: \cI_\inter(\cV)$ such that
	\begin{equation}
		\cI_\inter
		=
		-
		\int_1^\infty
		\fI_\inter
		\cF_\Lambda ( \cV + \cI_\inter + \cS (\cV + \cI_\inter))
		\dd \Lambda
		\label{eq:cI_stationary}
	\end{equation}
	and
	\begin{equation}
		\rho[\cI_\inter] \le \Cr{cI_norm} \rho[\cV]^2
		.
		\label{eq:cI_inter_quad}
	\end{equation}
	\label{thm:interpolated_stationary}
\end{theorem}

\begin{proof}
	For $\cV, \cI \in \pot_{\le M} $ let
	\begin{equation}
		\Phi_\inter^\cV (\cI)
		:=
		-
		\int_1^\infty \fD_\Lambda \fI_\inter \cQ^\fW(\cV + \cI + \cS (\cV + \cI),  \cV + \cI + \cS (\cV + \cI))
		\frac{\dd \Lambda}{\Lambda}
		.
		\label{eq:F_inter_def}
	\end{equation}
	Then using \cref{eq:Q_rho_bound,lem:fI_few_integral_bound}
	\begin{equation}
		\rho [\Phi_\inter^\cV (\cI)]
		\le
		\C 
		\mu[\cV + \cI + \cS(\cV + \cI)]^2
		;
	\end{equation}
	assuming $\rho[\cV + \cI] \le 1/2$ so as to use \cref{eq:cSl_norm_bound}, this implies
	\begin{equation}
		\rho [\Phi_\inter^\cV (\cI)]
		\le
		\C 
		\rho[\cV+\cI]^2
		,
		\label{eq:Phi_rho_squared_bound}
	\end{equation}
	and using \cref{lem:cS_diff_simple_bound} and proceeding as above also
	\begin{equation}
		\rho [\Phi_\inter^\cV (\cI - \cI')]
		\le
		\C
		\rho[\cV] \rho[\cI - \cI']
	\end{equation}
	and letting $\Cr{cI_norm} < 1/\Clast$ we can apply the Banach fixed point theorem to conclude that $\Phi_\inter^\cV$ has a unique fixed point $\cI_\inter$, which is the unique solution of \cref{eq:cI_stationary}, and the bound on $\rho[\cI_\inter]$ follows from \cref{eq:Phi_rho_squared_bound}.
\end{proof}

\begin{proof}[Proof of \cref{thm:FP_irrel}]
	Letting $\cI_\many (\cW) := \cS (\cW + \cI_\inter(\cW))$, \cref{eq:cI_inter_stationary} follows from \cref{eq:cI_stationary}, and \cref{eq:cI_many_stationary} follows from \cref{eq:cS_stationary}.  

	Recall the definition of $\rho$ in \cref{eq:convergence_requisite}.  Since each $\pot_k$ is a Banach space with norm $\| \cdot \|$, \cref{it:rho_prenorm} clearly holds, and \cref{it:rho_local_finite} is also obviously true since $\fL \cV$ always has only a finite number of nonzero components.
	Using \cref{eq:cI_inter_quad} with the comments around \cref{eq:convergence_requisite}, \cref{it:B_from_rho} holds with $1/(1+\Cr{cI_norm})$.
	Finally, the bounds in \cref{it:remainder_bounds} follow from the estimate in \cref{thm:interpolated_stationary} using \cref{eq:Q_rho_bound,lem:cS_diff_simple_bound}.
\end{proof}

\section{Conclusions}
\label{sec:conclusions}

The proof of \cref{thm:main_FP} can be concluded as follows, combining \cref{thm:FP_irrel} with some considerations on the relevant part from \cref{sec:relevant_parts} and the relationship between different forms of the fixed point problem from \cref{sec:flow_formal,sec:Wick}.  

\begin{proof}[Proof of \cref{thm:main_FP}]
	Using \cref{eq:cI_inter_stationary,eq:cI_many_stationary}, we see that for any $\cW_\fL^* \in \fL(\pot)$ which satisfies
	\begin{equation}
		\Delta \cW_\fL^*
		=
		\fL \cF_1 (\cW_\fL^* + \cI_\inter(\cW_\fL^*) + \cI_\many(\cW_\fL^*))
		\label{eq:rel_stationary_full}
	\end{equation}
	(recall that $\Delta$ is the generator of dilations reintroduced in \cref{eq:Delta_explicit}), $\cW^* := \cW_\fL^* + \cI_\inter(\cW_\fL^*) + \cI_\many(\cW_\fL^*)$ satisfies \cref{eq:Wick_stationary_formal} whenever it is well defined, i.e.\ whenever $\cV_\fL^* \in \sB$.

	To see that \cref{eq:rel_stationary_full} has a nontrivial solution, recalling the basis of $\fL(\pot)$ introduced in \cref{eq:fLpot_smallbasis} we can parameterize
	\begin{equation}
		\label{eq:FP_local}
		\cW_\fL^*
		=
		\wt \nu 
		\cY_\Omega^2
		+
		\lambda
		\cY_\Omega^4
		,
	\end{equation}
	making it clear that \cref{eq:rel_stationary_full} is in fact a finite-dimensional problem.
	Recalling \cref{eq:Y_fL_eig}
	\begin{equation}
		\Delta \cW_\fL^*
		=
		(d - 2[\psi])
		\wt \nu 
		\cY_\Omega^2
		+
		(d - 4 [\psi])
		\lambda
		\cY_\Omega^4
		=
		\tfrac12(d - 2 \varepsilon)
		\wt \nu 
		\cY_\Omega^2
		+
		2 \varepsilon
		\lambda
		\cY_\Omega^4
		.
	\end{equation}
	The right hand side of \cref{eq:rel_stationary_full} is less explicit, but it can be expanded in in powers of $\cW_\fL^*$ as
	\begin{equation}
		\begin{aligned}
			\fL \cF (\cW_\fL^* + & \cI_\inter(\cW_\fL^*) + \cI_\many(\cW_\fL^*))
			\\
			&=
			\begin{aligned}[t]
				\tfrac12
				\fL \cQ^\fW (\cW_\fL^* , \cW_\fL^* )
				&
				+
				\fL \cQ^\fW (\cW_\fL^* ,  \cI_\inter(\cW_\fL^*) + \cI_\many(\cW_\fL^*) )
				\\ &
				+
				\tfrac12
				\fL \cQ^\fW (\cI_\inter(\cW_\fL^*) + \cI_\many(\cW_\fL^*), \cI_\inter(\cW_\fL^*) + \cI_\many(\cW_\fL^*))
				.
			\end{aligned}
			\\ &
			=: 
			\tfrac12
			\fL \cQ^\fW (\cW_\fL^* , \cW_\fL^* )
			+
			\delta \cF(\cW_\fL^*)
			;
		\end{aligned}
		\label{eq:rel_RHS_split}
	\end{equation}
	writing out $\fL \cQ^\fW (\cW_\fL^* , \cW_\fL^* ) $ explicitly in terms of the parameters $\wt \nu, \lambda$ gives the right hand side of \cref{eq:truncated_FP_prelim}, and so \cref{eq:rel_RHS_split} is equivalent to \cref{eq:beta_with_error} with
	\begin{equation}
		\delta \cF(\cW_\fL^*)
		=
		B_{\wt \nu} (\wt \nu, \lambda)
		\cY_\Omega^2
		+
		B_{\lambda} (\wt \nu, \lambda)
		\cY_\Omega^4
		\label{eq:B_implicit_def}
	\end{equation}
	and the bounds \cref{eq:B_small,eq:B_Lipschitz} then follow immediately by applying the bounds in \cref{it:remainder_bounds} of \cref{thm:FP_irrel}. 
	As noted above, this immediately implies that \cref{eq:beta_with_error} has a solution of the form \cref{eq:rel_FP_asymp}, in particular with $\lambda \neq 0$, for small positive $\varepsilon$.  Retracing the steps above, this corresponds to a nonzero solution $\cW_\fL^*$ of \cref{eq:rel_stationary_full}, and so also a nonzero solution $\cW^* = \cW_\fL^* + \cI(\cW_\fL^*)$ of \cref{eq:Wick_stationary_formal}.
	Finally, recalling \cref{it:B_from_rho} of \cref{thm:FP_irrel}, $\cV^* := \fW^{-1} \cW^*$ is well-defined (and nonzero, since $\fW^{-1}$ is invertible) and $\Lambda \mapsto \fD_\Lambda \cV^*$ is indeed a solution of \cref{eq:flow}.  
\end{proof}

Note that the statement that $\cV^* \in \pot$ includes the requirement that it be invariant under the symplectic and rotation symmetries of the model.  The behavior of the local part of $\cV^*$ for small $\varepsilon$ can be obtained as discussed in \cref{sec:example}.  This gives $\rho[\cV^*]= \cO(\varepsilon), \ \varepsilon \to 0^+$, which together with \cref{eq:cI_inter_quad} implies
\begin{equation}
	\| \fI \cV^* \|_{\ell} = \cO(\varepsilon^2)
	, \quad
	\varepsilon \to 0_+
	, \ 
	\ell = 2,4
	.
\end{equation}
Furthermore, since the sum defining $S_\ell[\tau,\ul \Lambda,\cW^*_\cL]$ in \cref{eq:C_ell_tau_explicit} is empty unless the tree $\tau$ has at least $(\ell -2)/2$ vertices,
\begin{equation}
	\| \cW^* \|_\ell
	=
	\cO \left( \rho[\cW^*_\fL]^{(\ell-2)/2} \right)
	=
	\cO \left( \varepsilon^{(\ell-2)/2} \right)
	,
	\quad
	\varepsilon \to 0_+
	, \ 
	\ell = 6,8,\dots
\end{equation}
which in turn, recalling \cref{eq:fW_inverse}, implies bounds of the same order for $\cV^*$.  Further bounds can also be obtained in a similar fashion.

\bigskip

Although the correlation functions associated with the fixed point (which are automatically invariant under rescaling) can presumably be constructed using the methods of \cite{GMRS.correlations} (indeed, if the fixed point I construct here can be shown to be the element of \Grassmann\ equivalent to the fixed point kernel constructed there, the correlation functions are necessary the same), it should also be possible to provide a construction based directly on the Polchinski equation.  It will be interesting to see whether it is possible to provide a proof which involves the linearized flow near the fixed point, giving a more systematic treatment of all relevant observables.

The construction of the correlation functions is probably a prerequisite for approaching a number of open problems, such as the relationship between different choices of the cutoff and conformal invariance.  Recall that the cutoff, in the setup of this article, corresponds to choosing a particular $G_\Lambda$ among those corresponding to the same $P_\infty = \int_0^\infty G_\Lambda \dd \Lambda$.  Different choices of the cutoff are expected to be unimportant, at least in the sense of producing the same behavior at long distances, but this is not evident at the level of the effective action corresponding to the fixed point, which depends on the form of the cutoff, as can be seen from \cref{eq:rel_FP_asymp}.  The clearest notion of ``long range behavior'' is the asymptotic behavior of the correlation functions; for the fixed point studied in \cite{GMR,GMRS.correlations} (presumably the same one constructed here) there is no obvious dependence on the cutoff (in particular the explicit term in \cite[Equation~(2.42)]{GMRS.correlations} is independent of the cutoff $\chi$) but it has yet to be conclusively demonstrated even in this case \cite[Footnote~4]{GMRS.correlations}.

Conformal invariance is also a property of the correlation functions, and although it is very often the case that quantum field theories which are invariant under rescaling are also invariant under general conformal transforms, this is not always the case (see \cite{GNR24.noconformal,Nakayama24.noconformal} and the references therein).  
Conformal invariance can be approached via renormalization group equations \cite{Schafer.conformal}; the starting point would be to consider solutions where the dilation operator is extended to include general conformal transformations, which also involves a problem similar to studying a variety of cutoffs, since the usual classes of cutoffs are generated by scaling only.  
There is no obvious obstacle to reformulating the methods given here for a larger class of cutoffs, and it will be interesting to see whether this produces the desired invariance.

\section*{Acknowledgements}
I wish to thank Margherita Disertori, Paweł Duch, Luca Fresta, Alessandro Giuliani, and an anonymous referee for discussions or comments which led to various clarifications, improvements, and corrections.

The present work was supported in part by the MIUR Excellence Department Project MatMod@TOV awarded to the Department of Mathematics, University of Rome Tor Vergata, CUP E83C23000330006, and during the revision of the manuscript also by the ERC Starting Grant MaTCh (Grant agreement 101117299, P.I.\ Serena Cenatiempo).

The author is a member of GNFM-INdAM (Italy).

\appendix
\section{The Steiner diameter}
\label{sec:Steiner}

For $x,y \in \bR^d$, let $\lis{xy}$ denote the closed line segment from $x$ to $y$.  
Then for $\bs x = (x_1,\dots,x_\ell) \in \bR^{\ell d}$ 
let $T(\bs x)$ be the set of spanning trees $\bs x$, that is
\begin{equation}
	T(\bs x)
	:=
	\left\{ t \subset \{x_1,\dots,x_{\ell}\}^2 : \ \bigcup_{e \in t} \lis{e_1 e_2} \textup{ is simply connected} \right\},
\end{equation}
and for $t \in T(\bs x)$ let
\begin{equation}
	|t| 
	:=
	\sum_{e \in t}
	|e_1 - e_2|
	.
\end{equation}
Then the Steiner diameter of the set $\{ x_1,\dots,x_\ell\}$ or equivalently the tuple $\bs x$ is  
\begin{equation}
	\cT(\bs x)
	:=
	\cT(x_1,\dots,x_\ell)
	:=
	\inf_{k =0,1,\dots}
	\ 
	\inf_{x_{\ell+1},\dots,x_{\ell+k} \in \bR^d}
	\ 
	\inf_{t \in T(\bs x)} |t|
	,
	\label{eq:Steiner_def}
\end{equation}
that is the infimum of the length of all trees which connect all of the points in $\bs x$, with additional vertices allowed.  In fact the infimum is realized, but this is not important here.
What is important is that this obviously has the following properties.
For brevity, I will write $\cT(\bs x_1,\dots,\bs x_n)$ for the Steiner diameter of the sequence obtained by concatenating $\bs x_1,\dots,\bs x_n$.
Then for any tuples $\bs x, \bs y, \bs z$,
\begin{enumerate}
	\item For $\ell =2$, $\cT(x_1,x_2) = |x_1 - x_2|$,
	\item $\cT(\bs x) \le \cT(\bs x, \bs y)$,
	\item $\cT(\bs x, \bs y, \bs z) \le \cT(\bs x, \bs  y) + \cT(\bs y, \bs z)$.
\end{enumerate}
The inequalities follow immediately from inclusion relationships over the sets of trees involved.

\section{Pfaffians and the Gram-Hadamard bound}
\label{sec:Pfaff}

Let $g : X^2 \to \bC$ and $K \in (0,\infty)$; in this section, I will say that $g$ \emph{has a Gram decomposition with constant $K$} (or has Gram constant $K$) if there is a Hilbert space $\cH$ with inner product $ \left\langle \cdot,\cdot \right\rangle$ containing families of vectors $\{\gamma_x\}_{x \in X}, \ \{\tilde \gamma_y\}_{y \in X}$ such that
\begin{equation}
	g(x,y) = \left\langle \gamma_x, \tilde \gamma_y \right\rangle
	\text{ and }
	|\gamma_x|^2 , |\tilde \gamma_y|^2 
	\le 
	K
	\quad
	\forall x,y \in X
	.
	\label{eq:gram_def}
\end{equation}
For example, for the $g^\hk$ defined in \cref{eq:gt_def} can be written as an inner product in $L^2(\bR^d)$ as, for example,
\begin{equation}
	\begin{split}
		g^\hk(x-y)
		& =
		\frac{2 }{\Gamma\left( 1+\tfrac{d}4 + \tfrac{\varepsilon}2 \right)}
		\int \frac{\dd^d k}{(2 \pi)^d} 
		k^2
		e^{i \; k \cdot (x-y)}
		\\ & =
		\int 
		\begin{aligned}[t]
			&
			\left(  
				\sqrt{
					\frac{2 }{\Gamma\left( 1+\tfrac{d}4 + \tfrac{\varepsilon}2 \right)}
				}
				|k| e^{-k^2/2}
				e^{i k \cdot x}
			\right)
			\\ & \times
			\left(  
				\sqrt{
					\frac{2 }{\Gamma\left( 1+\tfrac{d}4 + \tfrac{\varepsilon}2 \right)}
				}
				|k| e^{-k^2/2}
				e^{- i k \cdot y}
			\right)
			\dd^d k 
		\end{aligned}
	\end{split}
\end{equation}
giving $ x,y \mapsto g^\hk(x-y)$ a decomposition with constant
\begin{equation}
	\C
	=
	\frac{2 }{\Gamma\left( 1+\tfrac{d}4 + \tfrac{\varepsilon}2 \right)}
	\int
	k^2 e^{-k^2}
	\dd^d k
	;
\end{equation}
also noting that
\begin{equation}
	\partial_x^{\ul m} \partial_y^{\ul n}g^\hk(x-y)
	=
	\frac{2 }{\Gamma\left( 1+\tfrac{d}4 + \tfrac{\varepsilon}2 \right)}
	\int \frac{\dd^d k}{(2 \pi)^d} 
		k^2
		(ik)^{\ul m} e^{i k \cdot x}
		(-ik)^{\ul n} e^{-i k \cdot y}
\end{equation}
there is a similar decomposition for $(x, \ul m),(y, \ul n) \mapsto\partial_x^{\ul m} \partial_y^{\ul n}g^\hk(x-y) $ (restricting to $|\ul m|, |\ul n| \le 2$) with a different constant $\Cl{gram_g_base}$. 
Following the same procedure with the decomposition $ \wt \xi = \sqrt{\wt \xi}\sqrt{\wt \xi}$ gives a similar decomposition for $g^\xi$

Given two functions $g_1 : X_1^2 \to \bC$, $g_2 : X_2^2 \to \bC$ with Gram compositions with respective constants $K_1, K_2$ the function $ (x_1,x_2),(y_1,y_2)\mapsto g_1(x_1,y_1)g_2(x_2,y_2)$ clearly has a Gram decomposition given by simply taking the tensor products, which has constant $K_1 K_2$.
For example, since $\Omega_{ab}$ defined in \cref{eq:Omega_matrix} trivially has a Gram decomposition with constant 1, $G = \Omega g$ has a Gram decomposition with the same constant $\Cr{gram_g_base}$ as $g$.

More generally, if we assume (as in \cref{sec:formalism}) that  $G$ (with up to $R$ derivatives on each argument) has Gram constant $\Cr{gram_g}$, then for $\Lambda \in [0,1]$, it is easy enough to see that $G_\Lambda(x,y) = \Lambda^{2[\psi]-1} G (\Lambda x, \Lambda y)$ does with $\Cr{gram_g} \Lambda^{2[\psi]-1}$ (for $\Lambda > 1$ the derivatives make this larger).
Then
\begin{lemma}
	For a given $R,s \in \bN+1$, $[\psi] >0$, and  $G \in \diffy_2$ with Gram constant $\Cr{gram_g}$, 
	all of the functions
	\begin{equation}
		(x, a, \ul m),(y, b, \ul n)
		\mapsto 
		\partial_x^{\ul m}
		\partial_y^{\ul n}
		\int_0^1
		f(\Lambda)
		G_\Lambda (x,y)
		\dd \Lambda
	\end{equation}
	with $f:[0,1] \to \bC$, $\sup_{x \in [0,1]} |f(x)| \le 1$
	have a Gram decomposition with constant (at most) $C_\GH =\Cr{gram_g}/2[\psi] $.
	\label{lem:gram_Lambda_integral}
\end{lemma}

This includes $P$ and also some other related functions introduced in \cref{sec:regroup}.

\begin{proof}
	Denoting the Gram decomposition of $G_\Lambda$ by
	\begin{equation}
		\partial_x^{\ul m}
		\partial_y^{\ul n}
		[G_\Lambda(x,y)]_{ab}
		=
		\left\langle \gamma^{(\Lambda)}_{x,a,\ul m} , \tilde \gamma^{(\Lambda)}_{y,b,\ul n}\right\rangle
		=
		\Lambda^{2 [\psi] + |\ul m| + |\ul n| -1}
		\left\langle \gamma_{\Lambda x,a,\ul m} , \tilde \gamma_{\Lambda y,b,\ul n}\right\rangle
	\end{equation}
	(I write this out to emphasize that these are all on the same Hilbert space $\cH$ independent of $\Lambda$),
	the function of interest trivially has a Gram decomposition on $\cH \otimes L^2 (0,1)$ with the advertised constant.
\end{proof}

\begin{lemma}
	If $f, g : X^2 \to \bC$ have Gram decompositions with the same constant $K$, then $f-g$ has a Gram decomposition with constant $2K$.
	\label{lem:Gram_diff}
\end{lemma}
\begin{proof}
	Writing $f(x,y) = \left\langle \varphi_x , \tilde \varphi_y \right\rangle$ and $g(x,y) = \left\langle \gamma_x , \tilde \gamma_y \right\rangle$, then
	$$(f-g)(x,y) = \left\langle \varphi_x \oplus \gamma_x , \tilde \varphi_y \oplus (-\tilde \gamma_y) \right\rangle.$$	
\end{proof}

The key application of these decompositions is the following bound:
\begin{theorem}[Gram-Hadamard bound]
	Let $M$ be a $2n \times 2n$ antisymmetric matrix such that $M_{ij} = g(x_i,x_j)$ for some $x_1,\dots,x_{2n} \in X$, $g: X^2 \to \bC$ with Gram constant $K$.

	Then $\Pf M \le K^n$.
	\label{thm:gram}
\end{theorem}

\begin{proof}
	Let $G : \bC^{2n} \to \cH$, $z \mapsto \sum_{j=1}^{2n} z_j \gamma_{x_j}$ and similarly $\tilde G$, so that $M = G^* \tilde G$.
	Letting $P$ be the projection onto the span of $\left\{ \gamma_{x_1},\dots,\gamma_{x_{2n}} \right\}$,
	\begin{equation}
		M 
		= G^* \tilde G 
		= G^* P \tilde G
		= (G^* P) (P \tilde G)
		=: \hat G^* \hat{ \tilde G}
	\end{equation}
	where $\hat G, \hat{\tilde G}$ are $2n \times 2n$ matrices whose rows are vectors with norm at most $\sqrt K$.  Then by Hadamard's inequality
	\begin{equation}
		|\det \hat G| , |\det \hat{\tilde G}|
		\le 
		K^{n}
		;
	\end{equation}
	thus $|\det M| \le K^{2n}$, and since $(\Pf M)^2 = \det M$ this gives the desired bound.
\end{proof}

\section*{Data availability}
Data sharing is not applicable to this article as no new data were created or
analyzed in this study.

\section*{Conflict of interest}
I am unaware of any conflicts of interest relevant to this article.

\defbibnote{preprints}{Note that for the preprints cited below I indicate the latest version I have consulted as of writing the relevant passages.}
\printbibliography[heading=bibintoc,prenote=preprints]
\end{document}